\RequirePackage[final]{graphicx}
\documentclass[UKenglish,cleveref, autoref]{lmcs}
\pdfoutput=1

\usepackage{lastpage}
\lmcsdoi{20}{4}{5}
\lmcsheading{}{\pageref{LastPage}}{}{}%
{Oct.~05,~2023}{Oct.~07,~2024}{}

\usepackage[utf8]{inputenc}
\usepackage{cite}
\usepackage[final,inline,nomargin]{fixme} \fxsetup{theme=color,mode=multiuser}
\FXRegisterAuthor{MB}{aMB}{MB}
\FXRegisterAuthor{JL}{aJL}{JL}
\FXRegisterAuthor{GZ}{aGZ}{GZ}
\usepackage{graphicx}
\usepackage{mathpartir}
\usepackage{apxproof}
\usepackage{nicefrac}
\usepackage[inline]{enumitem}
\usepackage{amsmath}
\usepackage{amssymb}
\usepackage{mathtools}
\usepackage{thmtools}
\usepackage{thm-restate}
\usepackage{xspace}
\usepackage{proof-dashed}
\usepackage{afterpage}
\usepackage{tikz}
\usetikzlibrary{arrows,shapes,automata,spy,positioning,
  calc,shadows,fit,arrows.meta,patterns,backgrounds,decorations.pathreplacing}
\usepackage{tikz-qtree}
\usepackage{wrapfig}
\usepackage[arrow,matrix,curve]{xy}
\usepackage{listings}
\usepackage{thm-restate}
\usepackage{tabularx,booktabs} 
\usepackage{ifthen}

\makeatletter
\RequirePackage[bookmarks,unicode,colorlinks=true]{hyperref}\def\@citecolor{blue}\def\@urlcolor{blue}\def\@linkcolor{blue}
\def\orcidID#1{\smash{\href{http://orcid.org/#1}{\protect\raisebox{-1.25pt}{\protect\includegraphics{orcid_color.eps}}}}}
\makeatother

\newcommand{\dual}[1]{\overline{#1}}

\newcommand{\auxarrow}
        {\mathop{\longrightarrow}}

\newcommand{\arrow}[1]
        {\, \auxarrow\limits^{#1} \,}

\newcommand{\Tbranchindex}[4]{\&\{{#1}_{#3}:{#2}_{#3}\}_{#3\in #4}}
\newcommand{\TbranchindexNoidx}[5]{\&\{{#1}_{#3}:{#2}, \ #5 \}_{ #4}}
\newcommand{\Tbranch}[2]{\Tbranchindex{#1}{#2}{i}{I}}

\newcommand{\Tbranchsingle}[2]{\&\{{#1}:{#2}\}}

\newcommand{\Tbranchset}[3]{\&\{{#1}\!:\!{#2}\}_{#1\in #3}}

\newcommand{\Tbra}[1]{\&\{#1\}}
\newcommand{\Tselectindex}[4]{\oplus\{{#1}_{#3}:{#2}_{#3}\}_{#3\in #4}}
\newcommand{\Tselect}[2]{\Tselectindex{#1}{#2}{i}{I}}
\newcommand{\Tselectsingle}[2]{\oplus\{{#1}:{#2}\}}
\newcommand{\Tselectset}[3]{\oplus\{{#1}:{#2}\}_{#1\in #3}}
\newcommand{\Tsel}[1]{\oplus\{ #1 \}}

\newcommand{\TselectindexNoidx}[5]{\oplus\{{#1}_{#3}:{#2}, \ #5 \}_{ #4}}

\newcommand{\Trec}[1]{\mu \mathbf{#1}}
\newcommand{\Tvar}[1]{\mathbf{#1}}
\newcommand{\Tend}{\mathbf{end}}
\newcommand{\context}[3]{\mathcal{#1}[{#2}]^{#3}}

\newcommand{\unfold}[2]{\mathsf{unfold}^{#1}(#2)}
\newcommand{\unf}[1]{\mathsf{unfold}(#1)}

\newcommand{\selunf}[1]{\mathsf{selUnfold}(#1)}

\newcommand{\dec}[1]{\mathsf{selRepl}(\Tvar t,\Tvar{\hat t},#1)}

\newcommand{\conf}{s}
\newcommand{\parconf}{|}
\newcommand{\cnfg}[4]{[#1,#2]\parconf[#3,#4]}
\newcommand{\append}{\!\cdot\!}
\newcommand{\refine}{\sqsubseteq}

\newcommand{\calA}{\mathcal{A}}

\newcommand{\grepl}[3]{#1\langle #2 \rangle^{#3}}
\newcommand{\crepl}[2]{#1\lfloor #2 \rfloor}

\newcommand{\tlab}{\lambda}
\newcommand{\holes}{\mathit{holes}}

\DeclareMathOperator{\subtype}{\leq}

\newcommand{\semT}[2]{[\![{#1}]\!]^{#2}}

\newcommand{\semcurly}[1]{\{\!\!\{{#1}\}\!\!\}}
\newcommand{\semthree}[1]{[\![\![{#1}]\!]\!]}
\newcommand{\semTt}[2]{[\![\![{#1}]\!]\!]^{#2}}

\newcommand{\unicontro}{\,\mathsf{ok}}
\newcommand{\contro}{\,\mathsf{ctrl}}
\newcommand{\treetrans}{\twoheadrightarrow}
\newcommand{\simtree}[2]{\mathit{simtree}(#1,#2)}
\newcommand{\treepair}[2]{( #1 , \; #2)}

\newcommand{\exTM}{\msg{tm}}
\newcommand{\exTC}{\msg{tc}}
\newcommand{\exDONE}{\msg{done}}
\newcommand{\exOVER}{\msg{over}}

\newcommand{\exGR}{T'_G}
\newcommand{\exG}{T_G}
\newcommand{\exS}{T_S}
\newcommand{\exSR}{T'_S}
\newcommand{\exSRR}{T''_S}

\newcommand{\ivar}[3]{#1_{#2,#3}}

\newif\iflong 

\longtrue 

\newcommand{\appendixref}{\iflong the appendix\else \cite{techreport}\fi}

\newcommand{\vt}{\Tvar t}
\newcommand{\free}[1]{\mathsf{free}(#1)}

\DeclareMathAlphabet{\mathpzc}{OT1}{pzc}{m}{it}

\newcommand{\xqedhere}[2]{\rlap{\hbox to#1{\hfil\llap{\ensuremath{#2}}}}}

\newcommand{\dst}{.\ }
\newcommand{\subs}[2]{\{\nicefrac{#1}{#2}\}}

\newcommand{\lempty}{\epsilon}

\newcommand{\word}{\omega}

\newcommand{\msg}[1]{\mathit{#1}}
\newcommand{\snd}[1]{!{\msg{#1}}}
\newcommand{\rcv}[1]{?\msg{#1}}

\newcommand{\ctx}[1]{\mathcal{#1}}

\makeatletter
\newcommand\reduline{\bgroup\markoverwith {\lower3.5\p@\hbox{\sixly \textcolor{red}{\char58}}}\ULon}\font\sixly=lasy6 \makeatother 

\newcommand{\lstCodeSize}{\normalsize}
\newcommand{\lstPrimitiveStyle}{\color{blue}\bfseries}

\newcommand{\lstNumberStyle}{\tiny\sffamily\color{gray}}

\lstdefinestyle{nonumber}{numbers=none,
  xleftmargin=1em,
  framexleftmargin=1em,
}

\tikzset{
every state/.style={minimum size=1pt,inner sep=1.5pt, initial text={}},
  mycfsm/.style={
    font=\scriptsize,
    initial where=left,
    initial distance=0.25cm,
    ->,>=stealth,auto, node distance=0.8cm and 0.8cm,
    scale=1, every node/.style={transform shape},
    baseline=(current  bounding  box.center)
  },
  ogate/.style = {
    diamond, draw, fill=white,
    minimum size=4mm,
    inner sep=0pt,
    postaction={path picture={\draw[black]
        ([yshift=\gatedistancein]path picture bounding box.south) -- ([yshift=-\gatedistancein]path picture bounding box.north)
        ([xshift=-\gatedistancein]path picture bounding box.east) -- ([xshift=\gatedistancein]path picture bounding box.west)
        ;}}, drop shadow},
  agate/.style={draw,rectangle,
    minimum size=3mm,
    inner sep=0pt,
    fill=white,
    postaction={path picture={\draw[black]
        ([yshift=\gatedistanceinand]path picture bounding box.south) --
        ([yshift=-\gatedistanceinand]path picture bounding box.north) ;}}, drop shadow},
  source/.style={draw,circle,fill=white,
    minimum size=3mm,
    inner sep=0pt, drop shadow},
  sink/.style={draw,circle,double,fill=white,
    minimum size=3mm,
    inner sep=0pt, drop shadow},
  intera/.style = {rectangle, draw=black, align=center, fill=white, rounded corners=0.1cm,
    minimum height=12,
    inner sep=2pt, drop shadow},
  line/.style = {draw,->, rounded corners=0.07cm,>=latex},
  venn/.style={preaction={fill, #1},opacity=0.6},
  cnode/.style={rectangle,draw=black,inner sep=2pt},
  ancestor/.style={densely dashed,->},
  silentedge/.style={>=latex,->},
  nlabel/.style={fill=white,inner sep=0pt,font=\footnotesize},
  notexplo/.style={fill=gray!10},
  echnode/.style={rectangle,draw=black,inner sep=2pt},
  schnode/.style={diamond,draw=black,inner sep=0pt},
}

\newcommand{\inference}[3]{\infer[\ifthenelse{\equal{#1}{}}{}{\inferrule{#1}}]{#3}{#2}}
\newcommand{\coinference}[3]{\infer=[\ifthenelse{\equal{#1}{}}{}{\inferrule{#1}}]{#3}{#2}}

\makeatletter
\def \rightarrowfill{\m@th\mathord{\smash-}\mkern-6mu\cleaders\hbox{$\mkern-2mu\mathord{\smash-}\mkern-2mu$}\hfill
  \mkern-6mu\mathord\rightarrow}
\makeatother

\newcommand{\trans}[1]{\stackrel{#1}{\rightarrow}}

\newcommand{\pic}{$\pi$-calculus}

\newcommand{\node}[1]{\raisebox{-.2ex}{\mbox{\large {\tt [}}} \, #1 \, \raisebox{-.2ex}{\mbox{\large {\tt ]}}}}

\newcommand{\pa}{{|}}

\newcommand{\sem}[1]{[\![#1]\!]}

\begin{document}

\title{Fair Asynchronous Session Subtyping}

\thanks{
  This work has been partially supported by the research project
  FREEDA (CUP: I53D23003550006) funded by the framework PRIN 2022
  (MUR, Italy), the French ANR project SmartCloud ANR-23-CE25-0012,
  and the H2020-MSCA-RISE project ID 778233 ``Behavioural Application
  Program Interfaces (BEHAPI)'' }

\author[M.~Bravetti]{Mario Bravetti\lmcsorcid{0000-0001-5193-2914}}[a]
\author[J.~Lange]{Julien Lange\lmcsorcid{0000-0001-9697-1378}}[c]
\author[G.~Zavattaro]{Gianluigi Zavattaro\lmcsorcid{0000-0003-3313-6409}}[b]

\address{University of Bologna, ITALY}
\email{mario.bravetti@unibo.it}

\address{University of Bologna / INRIA OLAS Team, ITALY}
\email{gianluigi.zavattaro@unibo.it}

\address{Royal Holloway, University of London, Egham, UK}	\email{julien.lange@rhul.ac.uk}

\begin{abstract}
Session types are widely used as abstractions of asynchronous 
  message passing systems.
Refinement for such abstractions is crucial as it allows
  improvements of a given component without compromising
 its compatibility with the rest of the system.
In the context of session types, the most general notion of
 refinement is asynchronous session subtyping, which allows
 message emissions to be anticipated w.r.t. a bounded
 amount of message consumptions.
In this paper we investigate the possibility to anticipate 
  emissions w.r.t. an unbounded amount of consumptions: 
to this aim we propose to consider fair compliance over
  asynchronous session types and fair refinement as the relation that
  preserves it.
This allows us to propose a novel variant of session subtyping
  that leverages the notion of controllability from service contract theory and that is a sound characterisation of fair refinement.
In addition, we show that both fair refinement and our novel
  subtyping are undecidable. We also present a sound algorithm which
  deals with examples that feature potentially unbounded buffering.
Finally, we present an implementation of our algorithm and an
  empirical evaluation of it on synthetic benchmarks.

\end{abstract}

\maketitle 

\section{Introduction}

The coordination of software components via message-passing techniques
is becoming increasingly popular in modern programming languages and
development methodologies based on actors and microservices, e.g.,
Rust, Go, and the Twelve-Factor App methodology~\cite{twelvefactor}.
Often the communication between two concurrent or distributed
components takes place over point-to-point \textsc{fifo} channels.

Abstract models such as communicating finite-state
machines~\cite{BZ83} and asynchronous session types~\cite{HYC16} are
essential to reason about the correctness of such systems in a
rigorous way.
In particular these models are important to reason about
mathematically grounded techniques to improve concurrent and
distributed systems in a compositional way.
The key question is whether a component can be \emph{refined}
independently of the others, without compromising the correctness of
the whole system.
In the theory of session types, the most general notion of refinement
is the asynchronous session subtyping~\cite{ESOP09, CDY2014,
  MariangiolaPreciness}, which leverages asynchrony by allowing the
refined component to anticipate message emissions, but only under
certain conditions.
Notably asynchronous session subtyping rules out candidate
subtypes that occur naturally in communication protocols where, e.g.,
two parties simultaneously send each other a finite but unspecified
amount of messages before removing them from their buffers.

We illustrate this key limitation of the asynchronous session
subtyping with Figure~\ref{fig:runex-types}, which depicts possible
communication protocols between a spacecraft and a ground station that
communicate via two unbounded asynchronous channels (one in each direction).
For convenience, the protocols are represented as
session types (bottom) and equivalent communicating finite-state
machines (top).
Consider $\exS$ and $\exG$ first.
Session type $\exS$ is the abstraction of the spacecraft. It may send
a finite but unspecified number of telemetries ($\exTM$), followed by a
message $\exOVER$ --- this phase of the protocol typically 
models a \texttt{for} loop and its exit.
In the second phase, the spacecraft receives a number of telecommands
($\exTC$), followed by a message $\exDONE$.
Session type $\exG$ is the abstraction of the ground station. It is
the \emph{dual} of $\exS$, written $\dual{\exS}$, as required in
standard binary session types without subtyping.
Since $\exG$ and $\exS$ are dual of each other, the theory of session
types guarantees that they form a \emph{correct composition}, namely
no communication errors can be generated and the communication protocol
can always terminate successfully, with empty queues.

However, it is clear that this protocol is not efficient: the
communication is half-duplex, i.e., it is never the case that more
than one party is sending at any given time.
Using full-duplex communication is crucial in distributed systems with
intermittent connectivity, e.g., in this case ground stations are not
always visible from low orbit satellites.

The abstraction of a more efficient ground station is given by type
$\exGR$, which sends telecommands before receiving telemetries.
In this way $\exGR$ and $\exS$ interact in a symmetric manner:
they first send all of their messages and then consume the messages 
sent from the other partner. No communication error can occur, and the 
communication protocol can always terminate successfully, with empty queues.
Unfortunately $\exGR$ is not an asynchronous subtype of
$\exG$ according to earlier definitions of session
subtyping~\cite{ESOP09, MariangiolaPreciness, CDY2014}.
Hence they cannot formally guarantee that $\exGR$ is a safe replacement
for $\exG$. 
Note that the composition of $\exGR$ and $\exS$ is not existentially
bounded, hence it cannot be verified by techniques based on
communicating finite-state machines~\cite{LangeY19, BouajjaniEJQ18,
  GenestKM06, GenestKM07}.

\begin{figure}[t]
  \centering
  \begin{tabular}{c@{\qquad\quad}c@{\qquad\quad}c}
    \begin{tikzpicture}[mycfsm, node distance = 0.5cm and 0.9cm
      ,scale=1.2, every node/.style={transform shape}]
      \node[state, initial, initial where=left] (s0) {$0$};
      \node[state, right =of s0] (s1) {$1$};
      \node[state, right=of s1] (s2) {$2$};
\path 
      (s0) edge [loop above] node {$\snd{\exTC}$} (s0)
      (s0) edge node {$\snd{\exDONE}$} (s1)
      (s1) edge [loop above] node [above] {$\rcv{\exTM}$} (s1)
      (s1) edge node [below] {$\rcv{\exOVER}$} (s2)
      ; 
    \end{tikzpicture}
    &
    \begin{tikzpicture}[mycfsm, node distance = 0.5cm and 0.9cm
      ,scale=1.2, every node/.style={transform shape}]
      \node[state, initial, initial where=left] (s0) {$0$};
      \node[state, right =of s0] (s1) {$1$};
      \node[state, right=of s1] (s2) {$2$};
\path 
      (s0) edge [loop above] node [above] {$\rcv{\exTM}$} (s0)
      (s0) edge node [below] {$\rcv{\exOVER}$} (s1)
      (s1) edge [loop above] node {$\snd{\exTC}$} (s1)
      (s1) edge node {$\snd{\exDONE}$} (s2)
      ; 
    \end{tikzpicture}
    &
    \begin{tikzpicture}[mycfsm, node distance = 0.5cm and 0.9cm
      ,scale=1.2, every node/.style={transform shape}]
      \node[state, initial, initial where=left] (s0) {$0$};
      \node[state, right =of s0] (s1) {$1$};
      \node[state, right=of s1] (s2) {$2$};
\path 
      (s0) edge [loop above] node [above] {$\snd{\exTM}$} (s0)
      (s0) edge node [below] {$\snd{\exOVER}$} (s1)
      (s1) edge [loop above] node {$\rcv{\exTC}$} (s1)
      (s1) edge node {$\rcv{\exDONE}$} (s2)
      ; 
    \end{tikzpicture}
    \\ $\exGR$  & $\exG = \dual{\exS}$ & $\exS $  
  \end{tabular}
  \begin{tabular}{lcl}
    $\exGR$ & = & $\Trec  t .\Tsel{ \exTC : \Tvar t  , \exDONE :
                  \Trec{t'} . 
                  ~\Tbra{ \exTM : \Tvar{t'} ,  \exOVER : \Tend } } $
    \\ 
    $\exG$ & = & $
                 \Trec{t}.~\Tbra{ \exTM : \Tvar{t} ,  \exOVER : \Trec{t'}
                 . \Tsel{ \exTC : \Tvar{t'}  , \exDONE : \Tend } } $
    \\
    $\exS$ & = &  $\Trec{t} . 
                 \Tsel{ \exTM : \Tvar{t} ,  \exOVER : \Trec{t'} .
                ~\Tbra{ \exTC : \Tvar{t'}  , \exDONE : \Tend } } $ 
  \end{tabular}
  \caption{Satellite protocols. $\exGR$ is the refined session type of the ground station, $\exG$ is the session type of ground station, and $\exS$ is the session type of the spacecraft.}\label{fig:runex-types}
\end{figure}

Technically speaking, previous asynchronous session subtyping relations
do not capture our spacecraft example due to the notion of correct
composition that they consider.
For instance, the notion of correct composition considered in~\cite{MariangiolaPreciness} 
imposes that all sent messages are guaranteed to be consumed
along \emph{all} possible computations of the receiver.
Following this approach the above type $\exGR$ is not a correct 
refinement of $\exG$
because
$\exGR$ can start by performing 
infinitely many outputs without consuming any incoming message.

The alternative notion of correct composition that we consider 
is weaker in that we do not impose a sent message to be consumed along all 
possible paths of the receiver, but we only require that, for all 
possible computation of the receiver either the message has been 
already consumed or there exists a continuation of the computation in 
which the message will be consumed.
More precisely, our notion of correctness is as follows:
given the composition of two session types, for every computation
there always exists a continuation of such computation reaching successful
termination (with empty queues).
This is a 
reasonable assumption, e.g., for programs that can conceptually run
indefinitely but must account for graceful termination (e.g., to
release acquired resources).

According to this notion of correct composition, 
$\exGR$ and $\exS$ are correct partners in that for every reachable state,
we can always find a way to terminate successfully the interaction.
This way to termination can be selected by exiting from the initial loops 
of outputs of both $\exGR$ and $\exS$.
The theory that we will develop will allow us to conclude that $\exGR$ 
is a correct refinement of $\exG$ for every possible partner, not only 
for the partner $\exS$.

The use of this notion of correct composition is new in the context
of asynchronous session types, but it has been already considered
in several related contexts. First of all, we observe that according to the terminology in \cite{GlabbeekH19}, our notion of correctness coincides with 
imposing that successful termination is a liveness-property
which holds under the assumption of {\em full fairness}.
For this reason, we name {\em fair compliance} our notion of correct composition.
Fair compliance has been already considered in 
the context of {\em synchronous} session types~\cite{Padovani16, CicconeP22}, 
in the definition of should testing \cite{RV07} where 
``every reachable state is required to be on a path to success'', and 
applied also to behavioural contracts~\cite{BravettiZ08,wsfm08}.

Given our notion of fair compliance defined on an operational
model for asynchronous session types, we define \emph{fair refinement}
the refinement relation that preserves it.
Then, we propose a novel variant of session subtyping called
\emph{fair asynchronous session subtyping}, that leverages the notion
of controllability 
from service contract theory, and
which is a sound characterisation of fair refinement.
We show that both fair refinement and fair asynchronous session
subtyping are undecidable, but give a sound algorithm for the
latter.
Our algorithm covers session types that exhibit complex behaviours
(including the spacecraft example and variants).
Our algorithm has been implemented in a tool available
online~\cite{tool}.

\paragraph{Structure of the paper}
The rest of this paper is structured as follows.
In \S~\ref{sec:refinement} we recall syntax and semantics
of asynchronous session types, we define 
\emph{fair compliance} and the corresponding
\emph{fair refinement}.
In \S~\ref{sec:subtyping} we introduce \emph{fair asynchronous
  subtyping}, the first relation of its kind to deal with examples
such as those in Figure~\ref{fig:runex-types}.
In \S~\ref{sec:algorithm} we propose a sound algorithm for subtyping
that supports examples with unbounded accumulations, including the
ones discussed in this paper.  In \S~\ref{sec:tool} we discuss the
implementation of this algorithm.
In \S~\ref{sec:eval} we present an evaluation of our implementation on generated session types.
Finally, in \S~\ref{sec:ending} we
discuss related and future work.
The paper includes also an~\appendixref\ containing details
of proofs that are not necessary in order to understand the 
main results that we have proved and the corresponding proof 
techniques.

This paper is based on the conference publication~\cite{BravettiLZ21}.
The main novelties w.r.t.~\cite{BravettiLZ21} are:
the inclusion of
all the proofs of our results, 
a completely new  
empirical evaluation of the implementation of our algorithm for checking 
fair asynchronous session subtyping (see\ \S~\ref{sec:eval}), 
an enriched and more comprehensive related work section.

 \section{Fair Refinement for Asynchronous Session Types}\label{sec:refinement}

In this section we first recall the syntax of two-party session types,
their reduction semantics, and a notion of compliance centred on the
successful termination of interactions.
We define our notion of refinement based on this compliance and show
that it is generally undecidable whether a type is a refinement of another.

\subsection{Preliminaries: Binary Session Types} \label{sub:prelm}

\paragraph{Syntax}
The formal syntax of two-party session types is given below. We follow
the simplified notation used in, e.g., \cite{BravettiCZ17,BCZ18}, without
dedicated constructs for sending an output/receiving an input.
Additionally we abstract away from message payloads since they are
orthogonal to the results of this paper.

\begin{defi}[Session Types]\label{def:sessiontypes}
  Given a set of labels $\mathcal{L}$, ranged over by $l$, the syntax
  of two-party session types is given by the following grammar:
\begin{displaymath}
    \begin{array}{lrl}
      T\ \ &::=&\ \   \Tselect{l}{T} 
                 \quad \mid\quad  \Tbranch{l}{T} 
                 \quad \mid\quad  \Trec t.T
                 \quad \mid\quad \Tvar t
                 \quad \mid\quad \Tend
\end{array}
  \end{displaymath}
\end{defi}

Output selection $\Tselect{l}{T}$ represents a guarded internal
choice, specifying that a label $l_i$ is sent over a channel, then
continuation $T_i$  is executed.
Input branching $\Tbranch{l}{T}$ represents a guarded external choice,
specifying a protocol that waits for messages. If message $l_i$ is
received, continuation $T_i$ takes place.
In selections and branchings each branch is tagged by a label $l_i$,
taken from a global set of labels $\mathcal{L}$. In each
selection/branching, these labels are assumed to be pairwise distinct.
In what follows, we leave implicit the index set $i \in I$ in input
branchings and output selections when it is clear from the
context.
Types $\Trec t.T$ and $\Tvar t$ denote standard recursion
constructs.
We assume recursion to be guarded in session types,
i.e., in $\Trec t.T$, the recursion variable $\Tvar t$ 
occurs within the scope of a selection or branching.
Session types are closed, i.e., all recursion variables $\Tvar t$ occur under the scope of a
corresponding binder $\Trec t.T$.  
Terms of the session syntax that are not closed
are dubbed (session) terms.   
Type $\Tend$ denotes the end of the interactions.

The dual of session type $T$, written $\dual{T}$, is inductively
defined as follows:
$\dual{\Tselect{l}{T}} = \Tbranch{l}{\dual{T}}$,
$\dual{\Tbranch{l}{T}} = \Tselect{l}{\dual{T}}$,
$\dual{\Tend} = \Tend$, $\dual{\Tvar t} = \Tvar t$, and
$\dual{\Trec t.T} = \Trec t.\dual{T}$.

\subsection{Asynchronous Fair Refinement}
We now define our notion of fair refinement.
We first define a reduction semantics formalizing the
interaction between two binary session types assuming asynchronous
communication via FIFO buffers. Then we formalize the notion of successful
final configuration; intuitively a configuration is successful
if both communicating types have completed their
send/receive operations and the buffers are empty.
Compliance is then defined as follows: two session types are compliant if, 
for every reachable configuration (according to the reduction semantics),
the interaction can continue to reach a successful configuration.
Finally, we say that a type $T$ refines another type $S$ if
it can safely replace $S$, i.e., if $S$ is compliant with a type $S'$
then also $T$ is compliant with $S'$.

In the definition of the reduction semantics for types we need some auxiliary
notation. Hereafter, we let $\word$ range over words in $\mathcal{L}^\ast$,
write $\lempty$ for the empty word, and write
$\word_1 \append \word_2$ for the concatenation of words $\word_1$ and
$\word_2$, where each word may contain zero or more labels.
Also, we write $T\subs{T'}{\Tvar t}$ for $T$ where
every free occurrence of $\Tvar{t}$ is replaced by $T'$.

We give an asynchronous semantics of session types via transition
systems whose states are configurations of the form:
$\cnfg{T_1}{\word_1}{T_2}{\word_2}$ where
$T_1$ and $T_2$ are session types equipped with two sequences
$\word_1$ and $\word_2$ of incoming messages (representing
unbounded buffers). We use $\conf$, $\conf'$, etc.\ to range over
configurations.

In this paper, we use explicit unfoldings of session types, as defined below.
\begin{defi}[Unfolding]\label{def:unfolding}
Given session type $T$, we define $\unf T$:
\[
\unf T = 
\begin{cases}
  \unf{T'\subs{T}{\Tvar{t}}} & \text{if $T=\Trec t.T'$}
  \\
  T & \text{otherwise}
\end{cases}
\]
\end{defi}
Definition~\ref{def:unfolding} is standard --- an equivalent
function is used in the first session subtyping~\cite{GH05}.
Notice that $\unf{T}$ unfolds all the recursive definitions in front
of $T$, and it is well defined for session types with guarded
recursion (c.f.\ assumptions in Section~\ref{sub:prelm}).

\begin{defi}[Transition Relation]
\label{def:transrel}
The transition relation $\rightarrow$ over configurations 
is the minimal relation satisfying the rules below (plus symmetric
ones):
\begin{enumerate}
\item \label{it:trans-send}
if $j \in I$ then
$\cnfg{\Tselect{l}{T}}{\word_1}{T_2}{\word_2} \rightarrow
\cnfg{T_j}{\word_1}{T_2}{\word_2\append l_j}$;
\item \label{it:trans-rcv}
if $j \in I$ then
$\cnfg{\Tbranch{l}{T}}{l_j \append \word_1}{T_2}{\word_2} \rightarrow
\cnfg{T_j}{\word_1}{T_2}{\word_2}$;
\item \label{it:trans-unfold} if
  $\cnfg{\unf {T_1}}{\word_1}{T_2}{\word_2} \rightarrow
  \conf$ then 
$\cnfg{T_1}{\word_1}{T_2}{\word_2} \rightarrow \conf$.
\end{enumerate}
We write $\rightarrow^*$ for the reflexive and transitive closure of
the $\rightarrow$ relation.
\end{defi}
Intuitively a configuration $\conf$ reduces to configuration $\conf'$
when either
\eqref{it:trans-send} a type outputs a message $l_j$, which is added
at the end of its partner's queue; 
\eqref{it:trans-rcv} a type consumes an expected message $l_j$ from
the head of its queue; or
\eqref{it:trans-unfold} the unfolding of a type can execute one of the
transitions above.

Next, we define successful configurations as those configurations
where both types have terminated (reaching $\Tend$) and both queues
are empty.
We use this to give  our definition of compliance which holds when
it is possible to reach a successful configuration from 
all reachable configurations.
\begin{defi}[Successful Configuration]\label{succomp}
The notion of \textit{successful configuration} is formalised
by a predicate $\conf\surd$ defined as follows:
\[
\cnfg{T}{\word_T}{S}{\word_S}\surd \;\; \mbox{iff} \;\; \unf
{T}\!=\!\unf{S}\!=\!\Tend \ \text{ and } \  \word_T\!=\!\word_S\!=\!\lempty
\]
\end{defi}

\begin{defi}[Compliance]
\label{def:compliance} \label{def:compatibility}
Given a configuration $\conf$ we say that it is a correct composition
if, whenever $\conf\rightarrow^*\conf'$, there exists a configuration
$\conf''$ such that
$ \conf' \rightarrow^\ast \conf''$ and $\conf'' \surd$.

\noindent
Two session types $T$ and $S$ are \emph{compliant} if $\cnfg{T}{\lempty}{S}{\lempty}$
is a correct composition.
\end{defi}

Observe that our definition of compliance is stronger than what is
generally considered in the literature on session types,
e.g.,~\cite{LangeY19, LY17, DY13}, where two types are deemed
compliant if all messages that are sent are eventually received, and
each non-terminated type can always eventually make a move.
Compliance is analogous to the notion of \emph{correct session}
in~\cite{Padovani16} but in an asynchronous setting.

A consequence of Definition~\ref{def:compatibility} is that it is
generally \emph{not} the case that a session type $T$ is compliant with its
dual $\dual{T}$, as we show in the example below.
\begin{exa}
  The session type 
  $ T = \Tbra{l_1 : \Tend, \  l_2 : \Trec{t} . \Tsel{l_3 : \Tvar{t}}}$
  and its dual $\dual{T} = \Tsel{l_1 : \Tend, \  l_2 : \Trec{t}
    . \Tbra{l_3 : \Tvar{t}}}$
  are not compliant.
Indeed, when $\dual{T}$ sends label $l_2$, the configuration
  $\cnfg{\Tend}{\lempty}{\Tend}{\lempty}$ is no longer reachable.
\end{exa}

We introduce a notion of refinement that preserves
compliance. This follows previous work done in the context of
behavioural contracts \cite{BravettiZ08} and \emph{synchronous}
multi-party session types \cite{Padovani16}.  
The key difference with these works is that we are considering
asynchronous communication based on (unbounded) \textsc{fifo} queues.
Asynchrony makes fair refinement undecidable, as we show below.

\begin{defi}[Refinement]\label{def:refine}
  A session type $T$ refines $S$, written $T \refine S$, if for every
  $S'$ s.t. $S$ and $S'$ are compliant then $T$ and $S'$ are also
  compliant.
\end{defi}
\noindent 
In contrast to traditional (synchronous and asynchronous) subtyping
for session types~\cite{GH05, MariangiolaPreciness, ESOP09},
this refinement is not covariant on outputs, i.e.,
it does not always allow a refined type to have output selections
with less labels.\footnote{The synchronous
subtyping in~\cite{GH05} follows a channel-oriented approach; 
hence it has the opposite direction and is contravariant on outputs.} 

\begin{exa}\label{ex:dual-not-refine}
  Let $T = \Trec{t} . \Tsel{l_1 : \Tvar{t}}$ and $S = \Trec{t}
  . \Tsel{l_1 : \Tvar{t}, \ l_2 : \Tend}$.
We have that $T$ is a synchronous (and asynchronous) subtype of
  $S$. However $T$ is \emph{not} a refinement of $S$.
In particular, the type $\dual{S} = \Trec{t}
  . ~\Tbra{l_1 : \Tvar{t}, \ l_2 : \Tend}$ is compliant with $S$
  but not with $T$, since $T$ does not terminate.
\end{exa}

\subsection{Undecidability of Fair Refinement}\label{subsec:undecidabilityRefinement}

Next, we show that the refinement relation $\refine$ is generally
undecidable.
The proof of undecidability exploits results from the tradition
of computability theory, i.e., Turing completeness of queue machines.
The crux of the proof is to reduce the problem of checking the
reachability of a given state in a queue machine to the problem of
checking the refinement between two session types.

\paragraph{Preliminaries}
Below we consider only state reachability in queue machines, and not
the typical notion of the language recognised by a queue machine
(see, e.g.,~\cite{BravettiCZ17} for a formalisation of queue machines).
Hence, we use a simplified formalisation, where no input string is considered.

\begin{defi}[Queue Machine]\label{def:queuemachines}
  A queue machine $M$ is defined by a five-tuple
  $(Q , \Gamma , \$ , s , \delta )$ where:
  \begin{itemize}
  \item $Q$ is a finite set of states;
\item $\Gamma$ is a finite set denoting the queue alphabet (ranged
    over by $A,B,C,X$);
  \item $\$ \in \Gamma$ is the initial queue symbol;
  \item $s \in Q$ is the start state;
  \item $\delta : Q \times \Gamma \rightarrow Q\times \Gamma ^{*}$ is
    the transition function ($\Gamma ^{*}$ is the set of sequences of
    symbols in $\Gamma$).
  \end{itemize}
\end{defi}

\noindent 
Considering a queue machine
$M=(Q , \Gamma , \$ , s , \delta )$,
a {\em configuration} of $M$ is an ordered pair
$(q,\gamma)$ where $q\in Q$ is its {\em current state} and
$\gamma\in\Gamma ^{*}$ is the {\em queue}.  The
starting configuration is $(s , \$)$, consisting of
the start state $s$ and the initial queue symbol $\$$.

Next, we define the transition relation ($\rightarrow
_{M}$), leading a configuration to another, and the related notion of
state reachability.
\begin{defi}[State Reachability]\label{queue_computation}  
  Given a machine $M\!\!=\!\!(Q , \Gamma , \$ , s , \delta )$,
the
  transition relation $\rightarrow _{M}$ over configurations
  $Q \times \Gamma ^{*}$ is defined as follows. For $p,q \in Q$, $A \in \Gamma$,
  and $\alpha,\gamma \in \Gamma ^{*}$, we have
  $(p,A\alpha )\rightarrow _{M}(q,\alpha \gamma)$ whenever
  $\delta (p,A)=(q,\gamma)$. Let $\rightarrow _{M}^{*}$ be the reflexive and
  transitive closure of $\rightarrow _{M}$.

  \noindent
  A target state $q_f \in Q$ is \emph{reachable} in $M$ if there
  is $\gamma \in\Gamma ^{*}$ s.t.
  $(s,\$)\rightarrow _{M}^{*}(q_f,\gamma)$.
\end{defi}

Since queue machines can deterministically encode Turing machines
(see, e.g., \cite{BravettiCZ17}),
checking state reachability for queue machines is undecidable.

\medskip

To prove the undecidability of fair refinement, we consider an
arbitrary queue machine $M$, and a target state $q_f$ for which we
define two session types $T$ and $S$ such that $T \refine S$ if and
only if state $q_f$ is reachable in $M$.
Hereafter, we use convenient notations for denoting output
selections and input branchings.
Instead of using labels indexed on an indexing set $I$, as in the
input branching syntax $\Tbranch{l}{T}$, we also use explicitly distinct
labels, as in $\&\{l:T_l, m: T_m\}$ (we use the same notation for
output selections).
We also use the union operator to combine disjoint sets of labels, for
instance, instead of writing $\oplus\{l_k:T_k\}_{k\in I\cup J}$, we
use the notation
$\oplus\{l_i:T_i\}_{i\in I} \cup \{l_j:T_j\}_{j\in J}$
(we use the same notation for
input branchings).

We start by defining the type $T=\sem{M,q_f,E}$.\footnote{\label{foot:labels}In
the definition of the type $T=\sem{M,q_f,E}$,
as well as in the definition $S=\sem{M,E}$,
we make the non restrictive assumption that the set of 
labels $\mathcal{L}$ of the Definition \ref{def:sessiontypes}
of the syntax of session types includes the symbols in the 
considered queue machine alphabet $\Gamma$ plus the additional 
symbol $E$.}
This type reproduces the 
finite control of the queue machine $M$, with a couple of differences:
($i$) it initialises the queue with symbol $\$$, and
($ii$) the state $q_f$ produces the additional ending symbol $E$ to communicate the 
end of the computation, then it consumes all symbols in the queue
and successfully terminates when $E$ is read from the queue.
In this way, the queue is empty when the type $T$ successfully terminates.

\begin{defi}[Finite Control Encoding]\label{def:controlEncoding}
  Let $M= (Q , \Gamma , \$ , s , \delta )$ be a queue machine,
  $q_f \in Q$, and $E \not\in \Gamma$ be the additional ending symbol;
we define $\sem{M,q_f,E}$ as follows:
  \[
  \sem{M,q_f,E}\ = \oplus \{\$: \semT{s}{\emptyset}\}
  \]
  where, given $q\in Q\setminus \{q_f\}$ and $\mathcal S\subseteq Q$,  
  $\semT{q}{\mathcal S}$ is defined as follows:
  \[
    \begin{array}{l}
      \semT{q}{\mathcal S} = 
      \left \{
      \begin{array}{l}
          \Trec{q}. 
          \Tbranchset{A}{
          \Tselectsingle{B^A_1}{\cdots
\Tselectsingle{B^A_{n_A}}{\semT{q'}{\mathcal S \cup q}
          }
}
          }                  
          {\Gamma} 
        \\[1mm]
        \hspace{0.9cm}\text{if }q\not\in {\mathcal S} \text{ and } \delta(q,A)=(q',B^A_1\cdots B^A_{n_A})
        \\
        \\
        \Tvar{q}\qquad \mbox{if $q \in {\mathcal S}$}
      \end{array}
      \right.
    \end{array}
  \]
  while $\semT{q_f}{\mathcal S} = \oplus \big\{E: \big(
  \Trec{\Tvar t}.\Tbranchset{A}{\Tvar t}{\Gamma} \cup \{E:\Tend\}
  \big)\ \big\}$
\end{defi}

We now define the type $S=\sem{M,E}$, that repeatedly behaves like 
a producer/consumer for all the symbols of the queue alphabet plus the
ending symbol $E$, with the difference that after producing and
consuming the ending symbol $E$, the type becomes $\Tend$.

\begin{defi}[Producer/consumer]\label{def:queueenc}
Let $M=
  (Q , \Gamma , \$ , s , \delta)$ be a queue machine and $E \not\in \Gamma$ be 
  the ending symbol.
  We define $\sem{M,E}$ as
  \[
  \sem{M,E} =
  \Trec{\Tvar t}.\Tselectset{A}{\&\{A:\Tvar t\}}{\Gamma} \cup \{E:\&\{E:\Tend\}\}
  \]
\end{defi} 

While $T=\sem{M,q_f,E}$ and $S=\sem{M,E}$ may appear unrelated, we
have that under some conditions $T \refine S$ holds.
Namely, $T \refine S$ if and only if $q_f$ is reachable in $M$. 
To prove this, we first characterize the set of types that are compliant with $S$.
This set consists of types that have the same behaviour (according to 
type bisimilarity) of $\dual S$, i.e., the dual of $S$.
The type $\dual S$, instead of being a producer/consumer, is a consumer/producer
which sends  the messages it receives back to the partner. This simulates a
FIFO queue that receives messages and sends messages in the same order
of reception.
Hence, the finite control encoding $T$, when combined with such consumer/producer
(i.e. any type having the same behaviour of $\dual S$), faithfully reproduces
the same behaviour of the encoded queue machine. A successful configuration
can be reached only if the type modeling the finite control terminates, 
and this is possible only if the final state $q_f$ is reached.

As mentioned above, the proof relies on the notion of type bisimilarity.

\begin{defi}[Type bisimilarity]\label{def:bisim}

  A relation 
$\,\mathcal R\!\!\;$ on session types is a bisimulation
  whenever
$(T,S)\in\mathcal R$ implies:
\begin{enumerate}
  
   \item \label{item:end}
  
    if $T=\Tend$ then 
    $\unf S = \Tend$;

   \item \label{item:internal}
  
    if $T=\Tselect{l}{T}$ then ${\unf S} = \Tselect{l}{S}$
     with
      $\forall i\in I.\,
      (T_i,S_i)\in\mathcal R$;
    
   \item \label{item:external}
  
    if $T=\Tbranch{l}{T}$ then ${\unf S} = \Tbranch{l}{S}$
     with
      $\forall i\in I.\,
      (T_i,S_i)\in\mathcal R$;
  \item \label{item:rec}
  
  if $T= \Trec t.{T'} $ then $(T'\{T/\Tvar{t}\}, S)\in\mathcal R$.
\end{enumerate}
   \noindent 
   $T$ is bisimilar to $S$, written
   $T \sim S$, if there is a bisimulation $\mathcal R$ such that $(T,S) \in \mathcal R$.
\end{defi}

Session type bisimilarity will be used only in the proof of undecidability 
of refinement and will not be involved in further developments in the 
remainder of the paper. Namely, we need bisimilarity in Lemma \ref{lem:dualQueue}
to characterise the session types that are compliant with $S=\sem{M,E}$.
Notice also that the relation $\sim$ is symmetric, i.e., if $(S,T) \in\ \sim$
then also $(T,S) \in\ \sim$.
In fact, the first three items of the above Definition 
simply check whether the l.h.s. and the r.h.s. terms
are either both $\Tend$ or have the same branching structure (i.e., the same 
set of labels) up-to unfolding of the r.h.s. But the
same effect of unfolding on the r.h.s. can be obtained on the l.h.s.
by (possibly repeated) application of the fourth item of the above definition.

In the proof of undecidability of refinement we need a result about 
bisimilar session types,
i.e., bisimilarity preserves compliance. Namely, we have that $T$ is compliant with $S$ if and only if $T'$ is compliant with $S'$ assuming $T \sim T'$ and $S \sim S'$.
This is an immediate corollary of
the following Lemma (which directly follows from the
bisimilarity of the considered types $T$ and $R$).

\begin{lem}
Consider the configuration $\cnfg{T}{\word_T}{S}{\word_S}$ and the session type $R$ s.t. $T \sim R$.
We have that:
\begin{itemize}
\item
$\cnfg{T}{\word_T}{S}{\word_S}\surd$ if and only if $\cnfg{R}{\word_T}{S}{\word_S}\surd$;
\item
$\cnfg{T}{\word_T}{S}{\word_S} \trans{} \cnfg{T'}{\word_T'}{S'}{\word_S'}$ if and only if
there exists $R' \sim T'$ s.t. 
$\cnfg{R}{\word_T}{S}{\word_S} \trans{} \cnfg{R'}{\word_T'}{S'}{\word_S'}$.
\end{itemize}
\end{lem}

\begin{cor}\label{cor:bisimilar}
Consider two pairs of bisimilar session types: $T \sim T'$ and $S \sim S'$. We have
that $T$ is compliant with $S$ if and only if $T'$ is compliant with $S'$. Moreover,
we have that $T \refine S$ if and only if $T' \refine S'$.
\end{cor}
 
\noindent 
As informally mentioned above, type bisimilarity allows us to characterize the
set of types that are compliant with a producer/consumer type $S=\sem{M,E}$,
for some queue machine $M$ and additional ending symbol $E$.
This result is formalized by the following Lemma (proof in Appendix \ref{subsec:queuemachines}).

\begin{restatable}{lem}{bisimilarityAndDual}\label{lem:dualQueue}
Let $M= (Q , \Gamma , \$ , s , \delta )$ be a queue machine 
and $E \not\in \Gamma$ the additional ending symbol.
Posing $S=\sem{M,E}$, for every session type $S'$ 
with input/output labels in $\Gamma \cup \{E\}$
we have that $S'$ 
is compliant with $S$ if and only if $S' \sim \dual{S}$.
\end{restatable}

\noindent 
The type $\dual S$ behaves like a FIFO queue,
which simply returns the messages it has received from 
the partner (in the same order).
Hence a type simulating the finite control $T=\sem{M,q_f,E}$, 
for the same queue machine $M$ and additional ending symbol $E$ as above,
turns out to be compliant with $\dual S$
if and only if the final state $q_f$ is reachable in $M$
(remember that only the encoding of $q_f$ allows to reach $\Tend$).
This result is formalized in the next theorem (proof in Appendix \ref{subsec:queuemachines}).

\begin{restatable}{thm}{encodingQueue}\label{th:encodingQueue}
Let $M= (Q , \Gamma , \$ , s , \delta )$ be a queue machine,
  $q_f \in Q$, $E \not\in \Gamma$ the additional ending symbol.
Posing $T=\sem{M,q_f,E}$ and $S=\sem{M,E}$, we have that $T$ is compliant with $\dual S$
if and only if $q_f$ is reachable in $M$.
\end{restatable}
\noindent 
Notice that the above theorem formalizes a reduction from the reachability
problem in queue machines to the verification of compliance between
session types. Hence, we can already conclude that the 
compliance relation is undecidable.

We now combine Corollary \ref{cor:bisimilar}, Lemma \ref{lem:dualQueue} and
Theorem \ref{th:encodingQueue} to prove the undecidability of refinement.
Consider the two above types $T=\sem{M,q_f,E}$ and $S=\sem{M,E}$.
By Lemma \ref{lem:dualQueue} we have that $S$ is compliant only with $\dual S$ 
and its bisimilar types.
Given that bisimulation preserves compliance (Corollary \ref{cor:bisimilar}) 
we have that $T$ refines $S$ if and only if it is compliant with $\dual S$.
But the latter holds if and only if 
$q_f$ is reachable in $M$ (Theorem \ref{th:encodingQueue}).
In this way we reduce the reachability
problem in queue machines to the verification of refinement between
session types.
We formally state this result in the theorem below (proof in Appendix \ref{subsec:queuemachines}).

\begin{restatable}{thm}{correfinementundec}\label{thm:undec}
Let $M= (Q , \Gamma , \$ , s , \delta )$ be a queue machine,
  $q_f \in Q$, $E \not\in \Gamma$ the additional ending symbol.
Posing $T=\sem{M,q_f,E}$ and $S=\sem{M,E}$,  we have that $T \refine S$ if
and only if
$q_f$ is reachable in $M$.
\end{restatable}

\noindent 
As a direct consequence of the above theorem and the undecidability of reachability 
in queue machines, we can conclude that refinement (Definition~\ref{def:refine}) 
is also undecidable.
\begin{restatable}{cor}{subtypingUndecidable}\label{cor:subtype-undec}
Given two session types $T$ and $S$, it is in general undecidable to 
check whether $T \refine S$ holds.
\end{restatable}

\subsection{Controllability and its Decidability} 

Given a notion of compliance, controllability amounts to checking the
existence of a compliant partner (see, e.g., \cite{Loh08,Wei08,BZ09a}).
In our setting, a session type is \emph{controllable} if there exists
another session type with which it is compliant.

Checking for controllability algorithmically is not trivial as it
requires to consider infinitely many potential partners. For the
synchronous case, an algorithmic characterisation was studied
in~\cite{Padovani16}.
In the asynchronous case, the problem is even harder
because each of the infinitely many potential partners may generate an
infinite state computation (due to unbounded buffers): specifically this reflects in the proof of its algorithmic characterisation.
The main contribution of this subsection is, thus, to give an algorithmic
characterisation of controllability in the asynchronous setting that is proven to be sound and complete.
Doing this is important because 
controllability is an essential ingredient for defining 
fair asynchronous subtyping, see Section~\ref{sec:subtyping}.

\iflong
\begin{figure}[t]\centering
  \begin{tikzpicture}[mycfsm, node distance = 0.5cm and 0.9cm
    ,scale=1, every node/.style={transform shape}]
    \node[state, initial, initial where=left] (s0) {$0$};
    \node[state, right =of s0] (s1) {$1$};
    \node[state, right=of s1] (s2) {$2$};
    \node[state, right=of s2] (s3) {$3$};
    \node[state, below right=of s2] (s4) {$4$};
\path 
    (s0) edge [bend left]  node [above] {$\rcv{l_1}$} (s1)
    (s1) edge [bend left] node [below] {$\rcv{l_3}$} (s0)
    (s1) edge node {$\rcv{l_2}$} (s2)
    (s2) edge node {$\snd{l_4}$} (s3)
    (s2) edge node [below] {$\snd{l_5}$} (s4)
    (s4) edge [loop right] node [right]  {$\snd{l_6}$} (s4)
    ; 
  \end{tikzpicture}
  \caption{Example of an uncontrollable session type, see Example~\ref{ex:uncontrol-ex}.}\label{fig:uncontrol-ex}
\end{figure}  
\fi

\begin{defi}[Characterisation of Controllability, $T \contro$]\label{def:controllability}
  We preliminarly define judgement $T\unicontro$ for session types $T$ having single input choices, i.e.\ such that all their input branches include just one possible choice. $T\unicontro$ is defined
  inductively as follows:
  \begin{mathpar}
    \inferrule
    {\,}
    {\Tend \unicontro}
    
    \inferrule
    {
      \Tend \in T
      \and 
      T\subs{\Tend}{\Tvar{t}} \unicontro
    }
    {\Trec t.T \unicontro}

    \inferrule
    {
      T\unicontro
    }
    { \Tbranchsingle{l}{T}  \unicontro}

        \inferrule
    {
      \forall i \in I \dst T_i \unicontro
    }
    {  \Tselect{l}{T}   \unicontro}
  \end{mathpar}
  where $ \Tend \in T$ holds if $\Tend$ occurs in $T$.

We now define predicate $T \contro$ over arbitrary session types $T$ as follows. $T \contro$ holds true if and only if there exists $T'$ such that:

\begin{enumerate}[($i$)]
\item $T'$ is obtained from $T$ by syntactically replacing every
  input choice $\Tbranch{l}{T}$ occurring in $T$ with a term
  $\Tbranchsingle{l_j}{T'_j}$ (with $j \in I$). Formally this is denoted by $T \; \mathsf{sin} \; T'$, where $\mathsf{sin}$ (standing for ``single input choices'') is defined as the smallest relation over session types such that:
  \begin{mathpar}
    \inferrule
    {\,}
    {\Tend \; \mathsf{sin} \; \Tend}

    \inferrule
    {\,}
    {\Tvar{t} \; \mathsf{sin} \; \Tvar{t}}
    
    \inferrule
    {
      T \; \mathsf{sin} \; T'
    }
    {\Trec t.T \; \mathsf{sin} \; \Trec t.T'}

    \inferrule
    {
      T_j \; \mathsf{sin} \; T'_j
      \and
      j \in I
    }
    { \Tbranch{l}{T} \; \mathsf{sin} \; \Tbranchsingle{l_j}{T'_j}}

        \inferrule
    {
      \forall i \in I \dst T_i \; \mathsf{sin} \; T'_i
    }
    {  \Tselect{l}{T} \; \mathsf{sin} \; \Tselect{l}{T'\!\!} }
  \end{mathpar}
In the following we use $\mathsf{sin}(T)$ to denote the set of single input choice types $T'$ such that $T \; \mathsf{sin} \; T'$.
\item $T' \unicontro$ holds true.
\end{enumerate}
\end{defi}
\noindent 
A type $T$ such that $T \contro$ is indeed controllable, in that
$\dual{T'}$, the dual of type $T'$ considered above, is compliant with $T$ (the predicate $\Tend\!\in\! T$ in the premise of 
the rule for recursion guarantees that a successful configuration is always reachable while looping).
Moreover the above definition
naturally yields a simple algorithm that
decides whether or not $T \contro$ holds for a type $T$, i.e., we
first pick a single branch for each input prefix syntactically
occurring in $T$ (there are finitely many of them) and then we
inductively check if $T'\unicontro$ holds.

\begin{exa}\label{ex:uncontrol-ex}
Consider the session type $T$ (see
  Figure~\ref{fig:uncontrol-ex} for a graphical representation):
  $$T = \Trec{t} .~\Tbra{l_1 : \Tbra{l_2 : \Tsel{l_4: \Tend, \ l_5 :
        \Trec{t'} . \Tsel{l_6 : \Tvar{t'}}}, \ l_3 : \Tvar{t} } } $$

  $T \contro$ does \emph{not} hold because it is not possible to
  construct a $T'$ as specified in
  Definition~\ref{def:controllability} for which $T' \unicontro$ holds.
In this case we have just two possible types $T'$ that can be obtained by input choice replacement: 
  $T' = \Trec{t} .~\Tbra{l_1 : \Tbra{l_3 : \Tvar{t} } } $
  and
  $T' = \Trec{t} .~\Tbra{l_1 : \Tbra{l_2 : \Tsel{l_4: \Tend, \ l_5 :
        \Trec{t'} . \Tsel{l_6 : \Tvar{t'}}} } } $. For the former $T' \unicontro$ does not hold because there is no 
$\Tend$ in the body of $\Trec{t}$; for the latter, instead, $T' \unicontro$ does not hold because there is no 
$\Tend$ in the body of $\Trec{t'}$.

As a result of Theorem~\ref{th:controlcharacter} (below), there is no
session type $S$ that is compliant with $T$. Hence $T$ is not
controllable.
\end{exa}

The following theorem shows that the judgement $T \contro$, as defined
above, precisely characterises controllability (i.e., the existence of
a compliant type). Its proof is rather complex (it requires introducing significant auxiliary technical machinery) and can be found in
Appendix \ref{subsec:controlcharacter}.

\begin{restatable}{thm}{thcontrolcharacter}\label{th:controlcharacter}
  $T \contro$ holds if and only if there exists a session type $S$ such that
  $T$ and $S$ are compliant.
\end{restatable}
\noindent{\em Sketch of the proof.}
 The proof relies on expressing session types via a set of equations, where
 each of the variables $\vt$ is mapped to an equation.
In essence, from $T$ \emph{controllable} we show that there exists a
 compliant type by considering the type $\dual{T'}$ (in equation set
 notation), where $T'$ is the type with single input branches obtained
 from $T$ by input choice replacement. The more difficult part of the proof
 is the opposite implication, where from the existence of any
 compliant $S$ we show that $T$ is controllable. 
This amounts to show that it is possible to build $T'$ from the
 transition system of the correct composition
 $\cnfg{T}{\lempty}{S}{\lempty}$ (in equation set notation), which is,
 in general, infinite state.
\qed

\section{Fair Asynchronous Session Subtyping}\label{sec:subtyping}
In this section, we present our novel variant of asynchronous
subtyping which we call \emph{fair asynchronous subtyping}.

First, we need to define a distinctive notion of unfolding.
As anticipated in the introduction (see the discussion about Figure \ref{fig:runex-types}), our subtyping will
identify the type $\exGR$ as a subtype of $\exG$, with
$$   \exG \ = \ 
                 \Trec{t}.~\Tbra{ \exTM : \Tvar{t} ,  \exOVER : \Trec{t'}
                 . \Tsel{ \exTC : \Tvar{t'}  , \exDONE : \Tend } } 
$$
Following the approach taken in other definitions of asynchronous
subtyping~\cite{MY15, MariangiolaPreciness, CDY2014}, our definition
will require to decompose the candidate supertype 
($\exG$ in our case) as an input context, with holes filled with subtypes starting with output selections. Notice that the subterm 
$\Tsel{ \exTC : \Tvar{t'}  , \exDONE : \Tend }$ of $\exG$ which 
starts with an output selection is not a correct subtype because it 
contains the free occurrence of the recursive variable $\Tvar{t'}$. 
Our distinctive notion of unfolding, will replace such free variable
with its definition. More precisely, we define the
function $\selunf{T}$ to unfold type $T$ by replacing recursion
variables with their corresponding definitions only if they 
are guarded by an output selection.
In the definition, we  use the predicate $\oplus\mathit{g}(\Tvar t, T)$
which holds if all instances of variable $\Tvar t$ are output selection guarded,
i.e., $\Tvar t$ occurs free in $T$
only inside subterms ${\Tselect{l}{T}}$. 

\begin{defi}[Selective Unfolding]\label{def:selectiveUnf}
Given a term $T$, we define $\selunf T=$ 
\\
\[
  \begin{cases}
    \Tselect{l}{T} & \text{if } T = {\Tselect{l}{T}}
    \\
    \&\{l_i : \selunf{T_i}\}_{i\in I}   & \text{if } T = {\Tbranch{l}{T}}
    \\
    T'\subs{\Trec t.T'}{\Tvar{t}}
& \text{if }
    T = {\Trec t.T'}
    \text{, $\oplus\mathit{g}(\Tvar t, T')$}
    \\
    \Trec t.\selunf{\dec{T'}\subs{\Trec t.T'}{\Tvar{ \hat{t}}}} \
    \mathit{with}\ \Tvar{ \hat{t}}\ \mathit{fresh}   & \text{if }
    T = {\Trec t.T'}
    \text{,  $\lnot\oplus\mathit{g}(\Tvar t, T')$}
    \\
    \Tvar{t} & \text{if } T = {\Tvar t}
    \\
    \Tend & \text{if } T = {\Tend}
  \end{cases}
\]
where, $\dec{T'}$ is obtained from $T'$ by replacing the
free occurrences of $\Tvar t$ that are inside a subterm
$\Tselect{l}{S}$ of $T'$ by $\Tvar{\hat{t}}$.

\end{defi}

\begin{exa}
  Consider the type $T = \Trec t . \Tbra{ l_1 : \Tvar{t}, \, l_2 :
    \Tsel{l_3 : \Tvar{t}}  }$, then we have
\[ 
    \selunf{T} = 
    \Trec t . \Tbra{ l_1 : \Tvar{t}, \, l_2 :  \Tsel{l_3 : \Trec t .~\Tbra{ l_1 : \Tvar{t}, \, l_2 :
          \Tsel{l_3 : \Tvar{t}}  }  }}
  \]
  i.e., the type is only unfolded within output selection sub-terms.
Note that $\Tvar{\hat{t}}$ is used to identify where unfolding must take
  place, e.g.,  \\
$\dec{\Tbra{ l_1 : \Tvar{t}, \, l_2 : \Tsel{l_3 : \Tvar{t}} }} =
  {\Tbra{ l_1 : \Tvar{t}, \, l_2 : \Tsel{l_3 : \Tvar{\hat t}} }}$.
\end{exa}

The last auxiliary notation required to define our notion of
subtyping is that of \emph{input contexts}, which are used to record
inputs that may be delayed in a candidate super-type.
In contrast to previous works on asynchronous subtyping, these input
contexts may include recursive constructs.
\begin{defi}[Input Context]\label{def:context}
An input context $\mathcal A$ is a session type with several holes
  defined by the syntax: 
\[
    \mathcal A\ ::=\ \quad[\,]^k
    \quad\mid\qquad \Tbranch{l}{\mathcal A}
    \quad\mid\qquad
    \Trec t.{\mathcal A}
    \quad\mid\qquad
    \Tvar t 
  \]
where the holes $[\,]^k$, with $k \in K$, of an input context
  $\mathcal A$ are assumed to be pairwise distinct.
We assume that recursion is guarded, i.e., in an input context
  $\Trec t.{\mathcal A}$, the recursion variable $\Tvar t$ must occur
  within a subterm $\Tbranch{l}{\mathcal A}$.
  
  We write $\holes(\calA)$ for the set of hole indices in $\calA$.
Given a type $T_k$ for each $k \in K$, we write
  $\context {A} {T_k} {k\in K}$ for the type obtained by filling each
  hole $k$ in $\mathcal A$ with the corresponding $T_k$. 
\end{defi}

In contrast to previous works~\cite{MariangiolaPreciness,ESOP09,
  CDY2014, BravettiCZ17, sefm19, BCLYZ19}, these input contexts may
contain recursive constructs.
This is crucial to deal with examples such as
Figure~\ref{fig:runex-types}.

We are now ready to define the {\it fair asynchronous subtyping}
relation, written $\subtype$.
The rationale behind asynchronous session subtyping is that under
asynchronous communication it is unobservable whether or not an output is
anticipated before an input, as long as this output is executed along all
branches of the candidate super-type. Besides the usage of our new recursive input contexts
the definition of fair asynchronous subtyping differs from those in~\cite{MariangiolaPreciness, ESOP09, CDY2014, BravettiCZ17,
sefm19, BCLYZ19} in that controllability plays a fundamental role: the subtype is not required to mimic supertype inputs leading to uncontrollable behaviours.

\begin{defi}[Fair Asynchronous Subtyping,
  $\subtype$]\label{def:subtyping}

  A relation $\,\mathcal R\!\!\;$ on session types is a controllable subtyping relation
  whenever

$(T,S)\in\mathcal R$ implies: \begin{enumerate}
   \item \label{fs:end}
  
    if $T=\Tend$ then 
    $\unf S = \Tend$;

  \item \label{fs:rec} 
  
  if $T= \Trec t.{T'} $ then $(T'\subs{T}{\Tvar{t}}, S)\in\mathcal R$;

   \item \label{fs:external}
  
  if $T=\Tbranch{l}{T}$ then 
    $\unfold {}S = \Tbranchindex lSjJ$, $I \supseteq K$, and
    $\forall k\in K  \ldotp (T_k, S_k)\in\mathcal R$,
where $K = \{ k \in J \; | \; S_k \text{ is controllable} \}$;


\newcommand{\TselectindexCOMPACTED}[4]{\oplus\{{#1}_{#3}\!:\!{#2}_{#3}\}_{#3\in #4}}

   \item \label{fs:internal}
if $T=\Tselect{l}{T}$ then 
$\selunf {S} \! = \! \context {A} {\TselectindexCOMPACTED l{S_k}{i}{I}} {k\in K}$ and
      $\forall i\!\in \!I.\,
      (T_i,\context {A} {{S_{ki}}} {k\in K})\!\in\!\mathcal R$.
    \end{enumerate}
    $T$ is a controllable subtype of $S$ if there is a controllable
    subtyping relation $\mathcal R$ s.t.\ 
    $(T,S) \, \in \, \mathcal R$.
    
    \noindent
    $T$ is a \emph{fair asynchronous subtype} of $S$, written
    $T \, \subtype \, S$, 
    whenever: $S$ controllable implies that $T$  is a controllable subtype of $S$.
\end{defi}  
Notice that the top-level check for controllability in the above definition is consistent
with the inner controllability checks performed in Case $(3)$.

\paragraph{Subtyping simulation game}
Session type $T$ is a fair asynchronous subtype of $S$ if $S$ is not
controllable or if $T$  is a controllable subtype of $S$.
Intuitively, the above co-inductive definition says that it is
possible to play a simulation game between a subtype $T$ and its
supertype $S$ as follows.
Case~\eqref{fs:end} says that if $T$ is the $\Tend$ type,
then $S$ must also be $\Tend$.
Case~\eqref{fs:rec} says that if $T$ is recursively defined, then
$T$ is replaced by the unfolding of its definition, $S$ is left unchanged and the simulation game continues.
Case~\eqref{fs:external} says that if $T$ is an input branching, then
the sub-terms in $S$ that are controllable can reply by inputting at
most 
some of the labels $l_i$ in the branching (contravariance of
inputs), and the simulation game continues (see
Example~\ref{ex:contra-variance}).
Case~\eqref{fs:internal} says that if $T$ is an output selection,
then $S$ can reply by outputting \emph{all} the labels $l_i$ in the
selection, possibly after executing some inputs, after which the
simulation game continues.
We comment further on Case~\eqref{fs:internal} with
Example~\ref{ex:subtyping-finite}.

\begin{exa}\label{ex:contra-variance}
  Consider 
  $T = \Tbra{l_1 : \Tend, \ l_2 : \Tend}$
  and
  $S = \Tbra{l_1 : \Tend, \ l_3: \Trec{t} . \Tsel{l_4: \Tvar{t}}}$.
We have $T \subtype S$. 
Once branch $l_3$, that is
  uncontrollable, is removed from $S$, we can apply contravariance for input branching.  We have
  $ I = \{1, 2\} \supseteq \{1\} = K$ in
  Definition~\ref{def:subtyping}.
\end{exa}

\begin{exa}\label{ex:subtyping-finite}
  Consider $\exG$ and $\exGR$ from Figure~\ref{fig:runex-types}.
For the pair $(\exGR, \exG)$, we apply
  Case~\eqref{fs:internal} of Definition~\ref{def:subtyping} for which we
  compute
\[ \selunf {\exG}= \calA [
    \Tsel{ \exTC : \Trec{t'}
      . \Tsel{ \exTC : \Tvar{t'}  , \exDONE : \Tend }  , \exDONE : \Tend
    }]
  \]
  with $\calA =
  \Trec{t}.\Tbra{ \exTM : \Tvar{t} ,  \exOVER : [\,]^1} $.
Observe that $\calA$ contains a recursive sub-term, such contexts
  are not allowed in previous works~\cite{MariangiolaPreciness,ESOP09,
    CDY2014}.
  
  The use of selective unfolding makes it possible to express $\exG$
  in terms of a \emph{recursive} input context $\calA$ with holes
  filled by types (i.e., closed terms) that start with an output
  prefix.
Indeed selective unfolding does not unfold the recursion variable
  $\vt$ (\emph{not} guarded by an output selection), which
  becomes part of the input context $\calA$.
Instead it unfolds the recursion variable $\vt'$ (which is guarded
  by an output selection) so that the term that fills the hole, which
  is required to start with an output prefix, is a closed term.

  Case~\eqref{fs:internal} of Definition~\ref{def:subtyping} requires
  us to check that the following pairs are in the relation:
  ($i$)     $( \exGR , \calA[ \Trec{t'}  . \Tsel{ \exTC : \Tvar{t'} , \exDONE :
    \Tend } ])$ and
  ($ii$)  $( \Trec{t'} . 
  ~\Tbra{ \exTM : \Tvar{t'} ,  \exOVER : \Tend } , \calA[ \Tend ])$.
Observe that $\exG  = \calA[ \Trec{t'}  . \Tsel{ \exTC : \Tvar{t'} , \exDONE :
    \Tend } ]$.
Hence, we have $\exGR \leq \exG$ with 
\[
    \mathcal{R} \!  = \!
    \left\{
      (\exGR, \exG), 
      (\Tend, \Tend),  
      (\Trec{t'} . 
      \Tbra{ \exTM \! : \Tvar{t'} ,  \exOVER \! : \Tend }
      , \Trec{t} .
      \Tbra{ \exTM \! : \Tvar{t} ,  \exOVER \! : \Tend }
      )
    \right\}
  \]
  and $\mathcal{R}$ is a controllable subtyping relation.

\end{exa}

We show that fair asynchronous subtyping is sound w.r.t.\
fair refinement.
In fact, fair asynchronous subtyping can be seen as a sound
coinductive characterisation of fair refinement.
Namely this result gives an operational justification to the 
syntactical definition of fair asynchronous session subtyping.
Note that $\subtype$ is not complete w.r.t.\ $\refine$, see
Example~\ref{ex:sub-not-complete}.

The proof of soundness of fair asynchronous subtyping w.r.t.\
fair refinement is rather complex and can be found in Appendix \ref{app:soundnessProof},
here we report the two main results and a sketch of their proofs.

\begin{restatable}{prop}{subtypingPropSoundness}\label{prop:pathTuSuccess}
Given two session types $T$ and $S$, 
if $T \subtype S$ then,
for every $\word$, $R$, and $\word_R$ such that 
$\cnfg{S}{\word}{R}{\word_R}$ is a correct composition, 
there exist $T'$, $\word'$, $R'$, and $\word_R'$
such that $\cnfg{T}{\word}{R}{\word_R} \trans{}^* \cnfg{T'}{\word'}{R'}{\word_R'}$
and $\cnfg{T'}{\word'}{R'}{\word_R'}\surd$.
\end{restatable}
\noindent{\em Sketch of the proof.}
Given that $\cnfg{S}{\word}{R}{\word_R}$ is a correct composition,
there exist $S'$, $\word''$, $R''$, and $\word_R''$ such that 
$\cnfg{S}{\word}{R}{\word_R} \trans{}^* \cnfg{S'}{\word''}{R''}{\word_R''}$
and $\cnfg{S'}{\word'}{R''}{\word_R''}\surd$.
The thesis is proved by induction on the length of this sequence
of transitions.

If the length is 0, then $\cnfg{S}{\word}{R}{\word_R}\surd$,
that implies $\unf{S}=\Tend$, that also implies $\unf{T}=\Tend$
(because $T \subtype S$), from which we have $\cnfg{T}{\word}{R}{\word_R}\surd$.

If the length is greater than 0, we proceed by case analysis on the
 first possible transition $\cnfg{S}{\word}{R}{\word_R} \trans{} \cnfg{S''}{\word'''}{R'''}{\word_R'''}$.

If the transition is inferred by $R$ it is sufficient to observe that
$S''=S$ and $\cnfg{T}{\word}{R}{\word_R} \trans{} \cnfg{T}{\word'''}{R'''}{\word_R'''}$,
and then apply the inductive hypothesis because $\cnfg{S''}{\word'''}{R'''}{\word_R'''}$
is a correct composition in that it is reachable from a correct composition.

We now consider that the transition is inferred by $S$.\\
There are three possible cases:
\begin{enumerate}
\item $\unf{S}=\Tselect{l}{S}$,
\item $\unf{S}=\Tbranch{l}{S}$ and $T$ starts with
an input branching (i.e., $\unf{T}=\&\{l_j:T_j\}_{j\in J}$), 
\item $\unf{S}=\Tbranch{l}{S}$ and $T$ starts with
an output branching (i.e., $\unf{T}=\oplus\{l_j:T_j\}_{j\in J}$).
\end{enumerate}
\noindent 
In the first two cases we have that
the above initial transition is 
$\cnfg{S}{\word}{R}{\word_R} \trans{} \linebreak \cnfg{S_i}{\word'''}{R'''}{\word_R'''}$, for
some $i \in I$. 
Given that $T\subtype S$, 
it is possible to show that $i \in J$, that $T_i \subtype S_i$, and also
$\cnfg{T}{\word}{R}{\word_R} \trans{} \cnfg{T_i}{\word'''}{R'''}{\word_R'''}$.
Then we can apply the inductive hypothesis because $T_i \subtype S_i$
and $\cnfg{S_i}{\word'''}{R'''}{\word_R'''}$
is a correct composition.

In the third case, given that $T \subtype S$,  and $S$ is controllable,
we have that 
$\selunf {S} = \context {A} {\Tselectindex l{S_k}{i}{J}} {k\in K}$,
and $\unf{T}=\oplus\{l_j:T_j\}_{j\in J}$
with $T_j \subtype \context {A} {S_{kj}} {k\in K}$, for every $j \in J$.
We first observe that the sequence of transitions 
$\cnfg{S}{\word}{R}{\word_R} \trans{}^* \cnfg{S'}{\word''}{R''}{\word_R''}$,
with $\cnfg{S'}{\word''}{R''}{\word_R''}\surd$,
includes at least one output selection $l_j$ executed
by one of the output selections filling the holes in $\ctx{A}$.
This label $l_j$ is the first one emitted by the l.h.s. type
after it has executed input branchings in $\ctx{A}$.
We have that the same sequence of transitions, excluding the output 
of $l_j$, can be executed from the configuration 
$\cnfg{\context {A} {S_{kj}} {k\in K}}{\word}{R}{\word_R \append l_j}$.
Such a sequence is $\cnfg{\context {A} {S_{kj}} {k\in K}}{\word}{R}{\word_R \append l_j} 
\trans{}^* \cnfg{S'}{\word''}{R''}{\word_R''}$,
with $\cnfg{S'}{\word''}{R''}{\word_R''}\surd$; notice that it is shorter than the 
above one.
We now consider $\cnfg{T}{\word}{R}{\word_R} \trans{} \cnfg{T_i}{\word}{R}{\word_R \append{l_j}}$.
We can now apply the inductive hypothesis on the shorter sequence 
$\cnfg{\context {A} {S_{kj}} {k\in K}}{\word}{R}{\word_R \append l_j} 
\trans{}^* \cnfg{S'}{\word''}{R''}{\word_R''}$, because $T_j \subtype \context {A} {S_{kj}} {k\in K}$
(and because it is possible to prove that  
$\cnfg{\context {A} {S_{kj}} {k\in K}}{\word}{R}{\word_R \append l_j}$
is also a correct composition, see Proposition \ref{prop:anticipation} in Appendix \ref{app:soundnessProof}).
\qed

\begin{restatable}{thm}{subtypingSoundness}\label{thm:soundness-refine-sub}
Given two session types $T$ and $S$, 
if $T \subtype S$ then $T \refine S$.
\end{restatable}
\noindent{\em Sketch of the proof.}
If $S$ is not controllable, then the thesis trivially holds
because $T \refine S$ for every $T$.

Consider now $S$ controllable.
The thesis is proved by showing that if $T \subtype S$ then,
for every $\word$, $R$, and $\word_R$ such that 
$\cnfg{S}{\word}{R}{\word_R}$ is a correct composition, we have that
the following holds:

  \noindent
if $\cnfg{T}{\word}{R}{\word_R} \rightarrow \cnfg{T'}{\word'}{R'}{\word_R'}$
then there exists $S'$ such that $T' \subtype S'$ and
$\cnfg{S'}{\word'}{R'}{\word_R'}$ is a correct composition.

The above implies the thesis because, given $T \subtype S$ and the 
correct composition $\cnfg{S}{\lempty}{R}{\lempty}$,
if there exists a computation 
$\cnfg{T}{\lempty}{R}{\lempty} \trans{}^* 
\cnfg{T'}{\word'}{R'}{\word_R'}$,
we can apply the above result
on each step of the computation 
to prove that there exists $S'$ such that $T' \subtype S'$
and $\cnfg{S'}{\word'}{R'}{\word_R'}$ is a correct composition.
Then, by Proposition \ref{prop:pathTuSuccess}, we have that 
there exist $T''$, $\word''$, $R''$, and $\word_R''$
such that $\cnfg{T'}{\word'}{R'}{\word_R'} \trans{}^* \cnfg{T''}{\word''}{R''}{\word_R''}$
and $\cnfg{T''}{\word''}{R''}{\word_R''}\surd$.
\qed

\begin{exa}\label{ex:sub-not-complete}
  Let $T = \Tsel{l_1 : \Tbra {l_3: \Tend}}$ and $S = \Tbra{l_3: \!
    \Tsel{l_1 : \Tend, \ l_2 : \Tend}}$.
We have $T \refine S$, but $T$ is not a fair asynchronous subtype of
  $S$ since $\{l_1\} \neq \{l_1,l_2\}$, i.e., covariance of outputs is
  not allowed.
\end{exa}

\subsection{Undecidability of fair asynchronous session subtyping}

\newcommand{\semTcont}[1]{{\{\!\!\{{#1}\}\!\!\}}}
\newcommand{\BsemT}[2]{{[\![\![{#1}]\!]\!]}^{#2}}

In this section we address the problem of checking fair asynchronous
session subtyping, and we show that it is actually undecidable. 
We have already proved that the fair refinement relation $\refine$
is undecidable (Corollary \ref{cor:subtype-undec}) and 
that the fair asynchronous subtyping relation $\subtype$ is a subset of the refinement relation 
$\refine$ (Theorem \ref{thm:soundness-refine-sub}). From these results we cannot 
immediately conclude that fair asynchronous subtyping is also undecidable;
hence we need a specific proof for this additional undecidability result. 
The approach we take has some commonalities with the
one adopted in Section~\ref{subsec:undecidabilityRefinement}, as 
we also proceed by reduction from undecidability properties 
in queue machines.
Nevertheless, there are several 
relevant differences. First, we consider termination in queue machines
instead of state reachability. Then we need to slightly modify the encodings 
of both the finite control and of the queue of the considered machine.
And finally, the proof of correctness of the encoding is significantly 
different as subtyping is defined on the syntax of types,
while refinement is defined on the operational semantics of 
(the parallel composition of) session types.

As anticipated above, we reduce the problem of checking
the (non)termination of a queue machine to the problem of checking
subtyping between two session types.
In Definition~\ref{queue_computation} we have defined 
$(q,\gamma) \rightarrow _{M} (q',\gamma')$ denoting
computation steps of a queue machine. We have that one queue machine $M$
terminates if and only if there exists a configuration with
empty queue that is reachable from the initial configuration, i.e., 
$(s,\$) \rightarrow _{M}^* (q',\lempty)$. This holds because the 
transition function is total in queue machines, hence if the queue
is not empty there is always a possible transition. In case
the queue machine does not terminate, we have that 
$(q,\$) \rightarrow _{M}^* (q',\gamma')$ implies
the existence of an additional computation step 
$(q',\gamma') \rightarrow _{M} (q'',\gamma'')$.

Given a queue machine $M=(Q , \Gamma , \$ , s , \delta )$ 
and an additional ending symbol $E \not\in \Gamma$,
we now define the types $T = \semthree{M,\_,E}$ and $S = \semthree{M,E}$
in such a way that $M$ 
does not terminate if and only if $T \subtype S$.
The encodings $\semthree{M,\_,E}$ and $\semthree{M,E}$ 
are similar to the corresponding encodings $\sem{M,q_f,E}$ and $\sem{M,E}$ 
defined in Definitions \ref{def:controlEncoding} and  \ref{def:queueenc},
but with the following differences:
\begin{itemize}
\item
there is no specific target state $q_f$;
\item
the encoding $\semthree{M,E}$
starts with an input branching with only one branch
labeled with the initial queue symbol $\$$ and continuation
corresponding to the producer/consumer $\sem{M,E}$
as defined in Definition \ref{def:queueenc};
\item
in order to be a potential subtype of $S = \semthree{M,E}$, all of
the output selections in $T = \semthree{M,\_,E}$ must have branchings 
for all of the symbols in $\Gamma \cup \{E\}$ (because these are the labels in the
output selection in the potential supertype); among all of these branchings only one
will be consistent with the encoding of the finite control, while 
the continuations in the other branchings are guaranteed to be always
good subtypes (this is guaranteed by
a type that nondeterministically produces symbols,
and that after producing the ending symbol $E$ it is able to recursively consume 
all possible symbols in $\Gamma$, and then become $\Tend$ after consuming the ending symbol $E$).
\end{itemize}

\begin{defi}[New Finite Control Encoding]\label{def:controlEncoding2}
  Let $M= (Q , \Gamma , \$ , s , \delta )$ be a queue machine
  and let $E \not\in \Gamma$ be the additional ending symbol.
  We define $\semthree{M,\_,E}$ as follows:
  $$
  \semthree{M,\_,E}\ = \semTt{s}{\emptyset}
  $$
  with, given $q\in Q$ and $\mathcal S\subseteq Q$,  
  $\semTt{q}{\mathcal S}$ is defined as follows:
  $$
    \begin{array}{l}
      \semTt{q}{\mathcal S} = 
      \left \{
      \begin{array}{l}
        \Trec{q}.\Tbranchset{A}{\semTcont{B^A_1\cdots B^A_{n_A}}_{q'}^{\mathcal S \cup \{q\}}}{\Gamma}
        \\[1mm]
        \hspace{0.9cm}\text{if }q\not\in {\mathcal S} \text{ and } \delta(q,A)=(q',B^A_1\cdots B^A_{n_A})
        \\
        \\
        \Tvar{q}\qquad \mbox{if $q \in {\mathcal S}$}
      \end{array}
      \right.
    \end{array}
  $$
  where
  $$
  \begin{array}{l}
           \semTcont{B_1\cdots B_{m}}_{r}^{\mathcal T} \!=\! \left\{\!\!
          \begin{array}{ll}
            \!\BsemT{r}{\mathcal T} 
            & \text{if }m=0\\
\begin{array}{ll}
            \!\!\!\!\oplus & 
            \!\!\!\!\big( \big\{B_1:  \semTcont{B_2\ldots B_m}_{r}^{\mathcal
            T}\big\} \cup
            \\
            & \! \big\{{A:V}\big\}_{A\in\Gamma\setminus\{B_1\}} \cup
            \{E: V'\}
             \big) 
            \end{array} 
            & \text{otherwise}
          \end{array}
              \right.
  \end{array}
  $$
  with $V=\Trec{\Tvar t}. \big( \Tselectset{A}{\Tvar t}{\Gamma} \cup \{E:V'\} \big)$
  and $V'=\Trec{\Tvar t}. \big( \Tbranchset{A}{\Tvar t}{\Gamma} \cup \{E:\Tend\} \big)$.
\end{defi}

\begin{defi}[New Producer/consumer]\label{def:queueenc2}
Let $M=
  (Q , \Gamma , \$ , s , \delta)$ be a queue machine and $E \not\in \Gamma$ be 
  the ending symbol.
  We define $\semthree{M,E}$ as
  $$
  \semthree{M,E} = \&\{\$ : \sem{M,E}\}
  $$
  with $\sem{M,E}$ as defined in Definition \ref{def:queueenc}.
\end{defi} 

We now prove that the above two types $T=\semthree{M,\_,E}$ and $S=\semthree{M,E}$
are such that $T \subtype S$ if and only if the machine $M$ does not terminate.
We report a sketch of the proof, the details are in Appendix \ref{app:undecidabilitySubtyping}.

\begin{restatable}{thm}{undecidabilitySubtyping}\label{th:undecidabilitySubtyping}
Given a queue machine $M$ and the ending symbol $E$,
consider $T=\semthree{M,\_,E}$ and $S=\semthree{M,E}$. We have that $T \subtype S$ if and only if 
$M$ does not terminate.
\end{restatable}

\noindent{\em Sketch of the proof.}
The only-if part is proved by considering the contrapositive statement,
that is, if the queue machine $M$ terminates then $T \not\!\!\!\,\subtype S$.
If the queue machine terminates, we have that $(s,\$) \rightarrow _{M}^* (q',\lempty)$.
Consider now the pair of types $(T,S)$ with $T=\semthree{M,\_,E}$ and $S=\semthree{M,E}$.
If, by contradiction, $T \subtype S$, since $S$ is controllable
(it is compliant, e.g., with its dual) we have that 
by Definition \ref{def:subtyping} there exists a fair asynchronous
subtyping relation $\mathcal R$ such that $(T,S) \in \mathcal R$.
By applying the definition of fair asynchronous
subtyping relation we have that $\mathcal R$ will have to include other pairs of
types $(T'',S'')$ corresponding with configurations $(q'',\gamma'')$ 
reachable in the queue machine $M$.
The types $T''$ represent the corresponding state $q''$,
while the types $S''$ represent the corresponding queue $\gamma''$.
Consider now the pair of types $(T_f,S_f)$ corresponding with the final configuration $(q',\lempty)$:
$T_f$ starts with an input branching (representing the
willingness to consume one symbol from the queue)
while $S_f$ starts with an output selection
(in fact, the representation of the queue starts with
a sequence of input branchings, one for each symbol in the queue, 
followed by an output selection and, given that it represents the empty queue,
the initial sequence of input branching is absent). 
Summarising, we have that $(T_f,S_f) \in \mathcal R$, $T_f$ starts with an input 
branching, and $S_f$ with an output selection: hence there is a pair 
in $\mathcal R$ which does not satisfy the item for input selection 
in Definition~\ref{def:subtyping}, thus contradicting the initial
assumption about $\mathcal R$ being a fair asynchronous subtyping
relation.

The if part is proved by showing that if the queue machine $M$ does not terminate
then there exists a fair asynchronous subtyping relation $\mathcal R$
that contains the pair $(T,S)$, hence $T \subtype S$.
There are two kinds of pairs in $\mathcal R$: (i) the pairs discussed in the above
only-if part of the proof that corresponds to the path in the subtyping
simulation game that reproduces the computation of the queue machine $M$,
and (ii) other pairs corresponding to alternative paths.
Here, we only comment the new pairs of kind (ii).
The l.h.s. types in these pairs are generated by
considering the alternative branches in the types  
$\semTcont{B_1\cdots B_{m}}_{r}^{\mathcal T}$ in Definition~\ref{def:controlEncoding2},
namely those involving the types denoted with $V$ and $V'$.
These types are of two kinds: (a) they are able to recursively perform all possible
outputs until the label $E$ is selected (type $V$),
or (b) they are able to recursively perform all possible inputs until the label 
$E$ is selected (type $V'$). 
All of these
pairs satisfy the constraints in Definition \ref{def:subtyping}
(under the assumption that also a final pair $(\Tend,\Tend)$ belongs
to $\mathcal R$). 
Summarising, there exists a fair asynchronous subtyping relation  $\mathcal R$
such that $(T,S) \in \mathcal R$ in that this is 
the first pair of the kind (i) above. Hence we can conclude that $T \subtype S$.
\qed

\medskip
\noindent 
As a direct consequence of the above theorem and the undecidability of termination 
in queue machines, we can conclude that fair asynchronous
subtyping (Definition~\ref{def:subtyping}) 
is also undecidable.

\begin{restatable}{cor}{subtypingUndecidable}\label{thm:subtype-undec}
Given two session types $T$ and $S$, it is in general undecidable to 
check whether $T \subtype S$.
\end{restatable}

\section{A Sound Algorithm for Fair Asynchronous Subtyping}\label{sec:algorithm}
We propose an algorithm which soundly verifies whether a session type is a
fair asynchronous subtype of another.
The algorithm relies on building a tree whose nodes are labelled by
configurations of the simulation game induced by
Definition~\ref{def:subtyping}.
The algorithm analyses the tree to identify \emph{witness} subtrees
which contain input contexts that are growing following a recognisable
pattern.

\begin{exa}\label{ex:super-spacecraft}
Recall the satellite communication example (Figure~\ref{fig:runex-types}).
The spacecraft with protocol $\exS$ may be a replacement for an
  older generation of spacecraft which follows the more complicated
  protocol $\exSR$, see Figure~\ref{fig:runex-big}.
Type $\exSR$ notably allows the reception of telecommands to be
  interleaved with the emission of telemetries.
The new spacecraft may safely replace the old one because
  $\exS \subtype \exSR$.
  
  However, checking $\exS \subtype \exSR$ leads to an infinite
  accumulation of input contexts, hence it requires to consider infinitely many pairs of session types. E.g.,
after $\exS$ selects the output label $\exTM$
  twice, the subtyping simulation game considers the pair $(\exS, \exSRR)$, where  $\exSRR$ is given in
Figure~\ref{fig:runex-big}.
The pairs generated for this example illustrate a common
  recognisable pattern where some branches grow infinitely (the
  $\exTC$-branch), while others stay stable throughout the derivation
  (the $\exDONE$-branch).
The crux of our algorithm is to use a finite parametric
  characterisation of the infinitely many pairs occurring in the check
of $\exS \subtype \exSR$.
\end{exa}

The \emph{simulation tree} for $T \subtype S$, written $\simtree{T}{S}$, is
the labelled tree representing the simulation game for $T \subtype S$,
i.e.,
$\simtree{T}{S}$ is a tuple $(N, n_0, \treetrans, \tlab)$
where $N$ is its set of nodes, $n_0 \in N$ is its root, $\treetrans$
is its transition relation, and $\tlab$ is its labelling function, such that
$\tlab(n_0) = (S,T)$. We omit the formal definition of $\treetrans$, as it
is straightforward from Definition~\ref{def:subtyping} following the
subtyping simulation game discussed after that definition.
We give an example below.

Notice that the simulation tree $\simtree{T}{S}$ is defined only when
$S$ is controllable, since $T \subtype S$ holds without needing to
play the subtyping simulation game if $S$ is not controllable.
We say that a branch of $\simtree{T}{S}$ is \emph{successful} if it is
infinite or if it finishes in a leaf labelled by $(\Tend, \Tend)$. All
other branches are \emph{unsuccessful}.
Under the assumption that $S$ is controllable, 
we have that all branches of $\simtree{T}{S}$  are
successful if and only if $T \subtype S$. 
As a consequence checking whether all branches of $\simtree{T}{S}$ are
successful is generally undecidable.
It is possible to identify a branch as successful if it visits
finitely many pairs (or node labels), see
Example~\ref{ex:subtyping-finite}; but in general a branch may
generate infinitely many pairs, see Examples~\ref{ex:super-spacecraft}
and~\ref{ex:space-infinite}.

\begin{figure}[t]
  \centering
  \begin{tabular}{c}
\begin{tikzpicture}[mycfsm, node distance = 0.5cm and 0.9cm
      ,scale=0.85, every node/.style={transform shape}]
      \node[state, initial, initial where=above] (s0) {$0$};
      \node[state, left =of s0] (s1) {$1$};
      \node[state, left=of s1] (s2) {$2$};
      \node[state, left=of s2] (s3) {$3$};
\node[state, right =of s0] (s4) {$4$};
      \node[state, right =of s4] (s5) {$5$};
\path
      (s0) edge [bend right] node [above] {$\rcv{\exTC}$} (s1)
      (s0) edge node {$\rcv{\exDONE}$} (s4)
      (s1) edge [bend right]  node {$\snd{\exTM}$} (s0)
      (s1) edge node [above] {$\snd{\exOVER}$} (s2)
      (s2) edge [loop above] node [above] {$\rcv{\exTC}$} (s2)
      (s2) edge node [below] {$\rcv{\exDONE}$} (s3)
      (s4) edge [loop above] node {$\snd{\exTM}$} (s4)
      (s4) edge node {$\snd{\exOVER}$} (s5)   
      ; 
    \end{tikzpicture}
    \\
    \begin{tabular}{l}
      \begin{tabular}{lcllll}
        $\exSR$ & = & $\Trec  t$ & $.\& \big\{$ &  
                                                  $\exTC :$
        & $\Tsel{\exTM : \Tvar{t} 
          ,
          \exOVER:
          \Trec{t'} . 
          ~\Tbra{ \exTC : \Tvar{t'} ,  \exDONE : \Tend }
          } , $
        \\
                &&& &
                      $
                      \exDONE :$ & $\Trec{t''} . \Tsel{
                                   \exTM : \Tvar{t''} 
                                   , \exOVER : \Tend
                                   }
                                   \big\} $
      \end{tabular}
      \\
      \begin{tabular}{lclllllllllllllll}
        $\exSRR$ & = & \phantom{$\Trec  t$} & $\phantom{.}\& \big\{$
        & $\exTC :$ 
        &   $\& \{$ & $\exTC :$ &

                                  $\exSR$, &
        \\
                &&&&&& $ \exDONE : $ & $\Trec{t''} . \Tsel{\exTM : \Tvar{t''}, \exOVER : \Tend}$ & $\}$,
        \\
                &&&& $ \exDONE : $  & \multicolumn{3}{l}{$\Trec{t''} . \Tsel{\exTM : \Tvar{t''}, \exOVER : \Tend}$ }& $\big\}$    
      \end{tabular}
    \end{tabular}
  \end{tabular}
\caption{$\exSR$ is an alternative session type for
    $\exS$, see Example~\ref{ex:super-spacecraft}.}\label{fig:runex-big}
\end{figure}

In order to support types that generate unbounded accumulation, we
characterise finite subtrees --- called witness subtrees, see
Definition~\ref{def:witness-tree} --- such that all the branches that
traverse these finite subtrees are guaranteed to be successful.

\paragraph{Notation}
We give a few auxiliary definitions and notations.
Hereafter $\calA$ and $\calA'$ range over \emph{extended} input
contexts, i.e., input contexts that may contain distinct holes
with the same index. These are needed to deal with
unfoldings of input contexts, see Example~\ref{ex:extended-ctxt}.

The set of \emph{reductions} of an input context $\calA$ is the minimal set
$\mathcal S$ s.t.
\begin{enumerate*}[label={($\roman*$)}]
\item $\calA \in \mathcal S$;
\item if $\&\{l_i : \calA_i\}_{i\in I} \in \mathcal S$ then
$\forall i\in I. \calA_i \in \mathcal S$ and \item\label{ICunfold} if $\Trec t.\calA' \in \mathcal S$ 
then $\calA'\subs{\Trec t.\calA'}{\Tvar t} \in \mathcal S$.
\end{enumerate*}
Notice that due to unfolding (item \ref{ICunfold}), the reductions of an input context may
contain extended input contexts.
Moreover, given a reduction $\calA'$ of $\calA$, we have that
$\holes(\calA') \subseteq \holes(\calA)$.

\begin{exa}\label{ex:extended-ctxt}
  Consider the following extended input contexts:
  \begin{mathpar}
    \calA_1 = \Trec{t} .~\Tbra{
      l_1 : [\,]^1, \
      l_2 : \Tbra{
        l_3 : \Tvar{t}
      }
    }

   \calA_2 =
    \Tbra{
      l_3 : \Trec{t} .~\Tbra{
        l_1 : [\,]^1, \
        l_2 : \Tbra{
          l_3 : \Tvar{t}
        }
      }
    }

    \unf{\calA_1} = \Tbra{
      l_1 : [\,]^1, \
      l_2 : \Tbra{
        l_3 : \Trec{t} .~\Tbra{
          l_1 : [\,]^1, \
          l_2 : \Tbra{
            l_3 : \Tvar{t}
          }
        }
      } 
    } 
  \end{mathpar}  
    Context $\calA_2$ is a reduction of $\calA_1$, i.e., one can reach
    $\calA_2$ from $\calA_1$, by unfolding $\calA_1$ and executing the input $l_2$. 
Context $\unf{\calA_1}$ is also a reduction of $\calA_1$.  Observe
    that $\unf{\calA_1}$ contains two distinct holes indexed by $1$.
\end{exa}

Given an extended context $\calA$ and a set of hole indices $K$ such
that $K \subseteq \holes(\calA)$, we use the following shorthands.
Given a type $T_k$ for each $k \in K$, we write
$\crepl{\calA}{T_k}^{k\in K}$ for the extended context obtained by replacing
each hole $k \in K$ in $\calA$ by $T_k$.
Also, given an extended context $\calA'$ we write
$\grepl{\calA}{\calA'}{K}$ for the extended context obtained by replacing each
hole $k \in K$ in $\calA$ by $\calA'$.
When $K = \{ k\}$, we often omit $K$ and write, e.g.,
$\grepl{\calA}{\calA'}{k}$ and $\crepl{\calA}{T_k}^{k}$.

\begin{figure}[t]\centering
  \begin{tikzpicture}    [node distance=1.2cm and 1.2cm
    ,scale=0.88, every node/.style={transform shape}
    ]
    \node[cnode] (n0) {$\treepair{\exS}{   \crepl{\calA}{\exSR,
          T'_1}^{ \{1,2\} }}$};
\node[cnode, below= of  n0] (n1) {$\treepair{  \Trec{t'} .  \Tbra{ \exTC : \Tvar{t'}  , \exDONE : \Tend }   }{  
        \crepl{
          \grepl{\calA}{
            \crepl{\calA}{T''_1}^{1}
          }{1}
        }
        {\Tend}^{2}}$};
    \node[cnode, below= of  n1] (n2) {$\treepair{  \Trec{t'} . \Tbra{ \exTC : \Tvar{t'}  , \exDONE : \Tend } }{   \crepl{
          \crepl{\calA}{T''_1}^{1}}{\Tend}^{2}}$};
    \node[cnode, below right = of  n1,xshift=-1cm] (n3) {$\treepair{\Tend}{\Tend}$};
\node[cnode, right = of  n0] (n20) {$\treepair{\exS }{  
        \crepl{
          \grepl{\calA}{
            \crepl{\calA}{\exSR}^{1}
          }{1}
        }
        {T'_1}^{2}}$};
\draw[silentedge] (n0) edge node [left] {$\snd{\exOVER}$} (n1);
    \draw[silentedge] (n1) edge node [right] {$\rcv{\exTC}$} (n2);
    \draw[ancestor] (n2) edge [bend left] node {} (n1);
    \draw[silentedge] (n1) edge node [above right] {$\rcv{\exDONE}$} (n3);
\draw[silentedge] (n0) edge node [above] {$\snd{\exTM}$} (n20);
      \draw[ancestor] (n20) edge [bend left] node {} (n0);
\node[draw=black,dashed, rounded corners=2pt,
      fit=(n0) (n1) (n2) (n20) (n3) ,inner sep=6pt] (t1) {};
\node[cnode, above= of n0,very thick,xshift=-3cm] (root) {$\treepair{\exS}{ \exSR}$};
      \draw[silentedge] (root) |- node [left] {$\snd{\exTM}$} (n0);
\node[cnode, right = of  root,xshift=2.5cm] (na) {$\treepair{  \Trec{t'} . \Tbra{ \exTC : \Tvar{t'}  , \exDONE : \Tend } }{   \crepl{
            \crepl{\calA}{T''_1}^{1}}{\Tend}^{2}}$};
\node[cnode, above left = of  na] (nb)
      {$\treepair{\Tend}{\Tend}$};
\node[cnode, above = of  na] (nc) {$\treepair{  \Trec{t'}
          . \Tbra{ \exTC : \Tvar{t'}  , \exDONE : \Tend } }{T''_1}$};
\node[cnode, above left = of  nc] (nd) {$\treepair{  \Tend  }{ \Tend
        }$};
\node[cnode, above = of  nc] (ne) {$\treepair{  \Trec{t'}
          . \Tbra{ \exTC : \Tvar{t'}  , \exDONE : \Tend } }{T''_1}$};

      \draw[silentedge] (root) edge node [above] {$\snd{\exOVER}$} (na);
      \draw[silentedge] (na) edge node [above right] {$\rcv{\exDONE}$} (nb);
      \draw[silentedge] (na) edge node [right] {$\rcv{\exTC}$} (nc);
      \draw[silentedge] (nc) edge node [above right] {$\rcv{\exDONE}$} (nd);
      \draw[silentedge] (nc) edge node [right] {$\rcv{\exTC}$} (ne);
      
      \draw[ancestor,sloped] (ne) edge [bend right] node [below] {=} (nc);
      
\node[draw=black, densely dotted, rounded corners=2pt,
    fit=(na) (nb) (nc) (ne) (nd) ,inner sep=6pt] (t2) {};    
\node[draw=black,right=of n1,xshift=-0.1cm,yshift=-0.3cm,fill=white,fill opacity=1,drop
    shadow] (txt) {
      {
        \begin{tabular}{ll@{ $ = \; $}l}
          & 
                  $\calA$ & $\Tbra{ \exTC : [\, ]^1 ,  \exDONE : [\,]^2 }$
          \\
                & $T'_1$ &  $\Trec{t''} . \Tsel{
                                 \exTM : \Tvar{t''} 
                                 , \exOVER : \Tend
                                 }$
          \\
                & $T''_1$ &  $   \Trec{t'} . ~\Tbra{ \exTC : \Tvar{t'} ,  \exDONE : \Tend }$
        \end{tabular}
      }
    };
  \end{tikzpicture}
    \caption{Simulation tree for $\exS \leq \exSR$
      (Figures~\ref{fig:runex-types} and~\ref{fig:runex-big}), the
      root of the tree is in bold.}
    \label{fig:simtree}
  \end{figure}

\begin{exa}\label{ex:algo-notation-space}
  Using the above notation and 
  posing $\calA = \Tbra{ \exTC : [\, ]^1 , \exDONE : [\,]^2 }$, we can
  rewrite $\exSRR$ (Figure~\ref{fig:runex-big}) as
  $ \crepl{ \grepl{\calA}{ \crepl{\calA}{\exSR}^{1} }{1} }
  {
    \Trec{t''} . \Tsel{
      \exTM : \Tvar{t''} 
      , \exOVER : \Tend
    }
  }^{2}$.
\end{exa}

\begin{exa}\label{ex:algo-notation}
  Consider the session type below
 \[
  S= \Tbra{ 
    l_1 : \Tbra{ l_1 : T_1 ,\  l_2 : T_2 ,\ l_3: T_3}  , \;
    l_2 : \Tbra{ l_1 : T_1 ,\  l_2 : T_2 ,\ l_3: T_3} , \;
    l_3 :T_3 }.
  \]
  Posing  $\calA = \Tbra{ l_1 : [\, ]^1 ,  l_2 : [\,]^2 ,  l_3 : [\,]^3 }$ we
  have $\holes(\calA) =  \{1,2,3\}$.
Assuming $J = \{1,2\}$ and $K=\{3\}$, we can rewrite $S$ as
$\crepl{\grepl{\calA}{ \crepl{\calA}{T_j}^{j \in J} }{J}}{T_k}^{k
    \in K}$.
\end{exa}

\begin{exa}\label{ex:space-infinite}
Figure~\ref{fig:simtree} shows the partial simulation tree for
$\exS \leq \exSR$, from Figures~\ref{fig:runex-types}
and~\ref{fig:runex-big} (ignore the dashed edges for now).
Notice how the branch leading to the top part of the tree visits only
finitely many node labels (see dotted box), however the bottom part of
the tree generates infinitely many labels, see the path along the
$\snd{\exTM}$ transitions in the dashed box.
\end{exa}

\paragraph{Witness subtrees}
Next, we define witness trees which are finite subtrees of a
simulation tree which we prove to be successful.
The role of the witness subtree is to identify branches that satisfy a
certain accumulation pattern.
It detects an input context $\calA$ whose holes fall in two
categories:
($i$) growing holes (indexed by indices in $J$ below) which lead to an
infinite growth and
($ii$) constant holes (indexed by indices in $K$ below) which stay
stable throughout the simulation game.
The definition of witness trees relies on the notion of
\emph{ancestor} of a node $n$, which is a node $n'$ (different
from $n$) on the path from the root $n_0$ to $n$.
We illustrate witness trees with Figure~\ref{fig:simtree} and
Example~\ref{ex:full-algo-ex}.
\begin{defi}[Witness Tree]\label{def:witness-tree}
  A finite tree $(N, n_0, \treetrans, \tlab)$ is a \emph{witness tree} for
  $\calA$, such that $\holes(\calA) = I$, with
  $\emptyset \subseteq K \subset I$ and $J = I \setminus K$, if all the
  following conditions are satisfied:
\begin{enumerate}
\item \label{cond:nodes}
  for all $n \in N$ either
$\tlab(n) =   (T,    \crepl{\grepl{\calA'}{\crepl{\calA}{S_j}^{j \in J}}{J}}{S_k}^{k
    \in K})$ or \newline
$\tlab(n) = 
(T,    \crepl{   \grepl{\calA'}{   \grepl{\calA}{
        \crepl{\calA}{S_j}^{j\in J} }{J}}{J}  }{S_k}^{k\in K} ) $, 
  where
$\calA'$ is a reduction of $\calA$, and it holds that
  \begin{itemize}
  \item 
  $\holes(\calA')\subseteq K$ implies that $n$ is a leaf and 
  \item if $\tlab(n) = (T,    \calA[S_i]^{i\in I})$ and $n$ is not a leaf
  then $\unf{T}$ starts with an output selection;
  \end{itemize}
\item \label{algo:leaves}
 each leaf $n$ of the tree satisfies one of the following conditions:
  \begin{enumerate}
\item  \label{algo:loop}
    $ \tlab(n) = (T, S)$  and  $n$
    has an ancestor $n'$ s.t.\ $\tlab(n') = (T,S)$
\item  \label{algo:increase}
    $ \tlab(n) = (T,    \crepl{\grepl{\calA}{\crepl{\calA}{S_j}^{j \in J}}{J}}{S_k}^{k
      \in K})$  
    and 
    $n$ has an ancestor $n'$ s.t.\ $\tlab(n')\! = \!(T,
    \calA[S_i]^{i\in I})$
\item \label{algo:decrease}
$ \tlab(n) = (T,    \calA[S_i]^{i\in I})$ 
and
$n$ has an ancestor $n'$ s.t.\
$\tlab(n')\! = \!(T, \crepl{\grepl{\calA}{\crepl{\calA}{S_j}^{j \in
      J}}{J}}{S_k}^{k \in K})$
\item \label{algo:const}
    $ \tlab(n) = (T, \calA'[S_k]^{k \in K'})$ where $K' \subseteq K$
  \end{enumerate}
  and for all leaves $(T,S)$ of type~\eqref{algo:decrease}
  or~\eqref{algo:const} $T \subtype S$ holds.
\end{enumerate}
\end{defi}

\noindent 
Intuitively Condition~\eqref{cond:nodes} says that a witness subtree
consists of nodes that are labelled by pairs $(T,S)$ where $S$
contains a fixed context $\calA$ (or a reduction/repetition thereof)
whose holes are partitioned in growing holes ($J$) and constant holes
($K$).
Whenever all growing holes have been removed from a pair (by reduction
of the context) then this means that the pair is labelling a leaf of
the tree. In addition, if the initial input is limited
to only one instance of $\calA$, the l.h.s.\ type starts with an output
selection so that this input cannot be consumed in the subtyping simulation
game.

Condition~\ref{algo:leaves} says that all leaves of the tree must
validate certain conditions from which we can infer that their
continuations in the full simulation tree lead to successful branches.
Leaves satisfying Condition~\eqref{algo:loop} straightforwardly lead
to successful branches as the subtyping simulation game, starting from 
the corresponding pair, has been already checked starting from its
ancestor having the same label.
Leaves satisfying Condition~\eqref{algo:increase} lead to an infinite
but regular ``increase'' of the types in $J$-indexed holes ---
following the same pattern of accumulation from their ancestor.
The next two kinds of leaves must additionally satisfy the subtyping
relation --- using witness trees inductively or based on the fact they
generate finitely many labels.
Leaves satisfying Condition~\eqref{algo:decrease} lead to regular
``decrease'' of the types in $J$-indexed holes --- following the same
pattern of reduction from their ancestor.
Leaves satisfying Condition~\eqref{algo:const} 
use only constant $K$-indexed holes because, by reduction
of the context $\calA'$, the growing holes containing 
the accumulation $\calA$ have been removed.

\begin{rem}
  Definition~\ref{def:witness-tree} is parameterised
  by an input context $\calA$.
We explain how such contexts can be identified while building a
  simulation tree in Section~\ref{sec:tool}.
\end{rem}

\begin{exa}\label{ex:full-algo-ex}
  In the tree of Figure~\ref{fig:simtree} we highlight two subtrees.
The subtree in the dotted box is not a witness subtree
  because it does not validate Condition~\eqref{cond:nodes} of
  Definition~\ref{def:witness-tree}, i.e., there is an intermediary
  node with a label in which the r.h.s\ type does not contain $\calA$.

The subtree in the dashed box is a witness subtree with 3 leaves,
  where the dashed edges represent the ancestor relation,
  $\calA =\Tbra{ \exTC : [\, ]^1 , \exDONE : [\,]^2 }$, $J = \{1\}$
  and $K = \{2\}$.
We comment on the leaves clockwise, starting from $(\Tend, \Tend)$,
  which satisfies Condition~\eqref{algo:const}.
The next leaf satisfies condition~\eqref{algo:decrease}, while the
  final leaf satisfies Condition~\eqref{algo:increase}.
\end{exa}

\paragraph{Algorithm}
Given two session types $T$ and $S$ we first check whether $S$ is uncontrollable.
If this is the case we immediately conclude that $T \subtype S$. Otherwise, we 
proceed in four steps.
\begin{enumerate}[label=\textbf{S\arabic*}, leftmargin=2.5em,wide, labelwidth=!, labelindent=0pt]\item \label{step:build}
We compute a finite fragment of $\simtree{T}{S}$, stopping
  whenever
($i$) we encounter a leaf (successful or not),
($ii$) we encounter a node that has an ancestor as defined in
  Definition~\ref{def:witness-tree} (Conditions~\eqref{algo:loop},
  \eqref{algo:increase}, and~\eqref{algo:decrease}),
($iii$) or the length of the path from the root of $\simtree{T}{S}$
  to the current node exceeds a bound set to two times the depth of
  the AST of $S$.
This bound allows the algorithm to
  explore paths that will traverse the super-type at least twice.
We have empirically confirmed that it is sufficient for all
  examples mentioned in Section~\ref{sec:tool}.

\item \label{step:prune}
We remove subtrees from the tree produced in~\ref{step:build}
  corresponding to successful branches of the simulation game which
  contain finitely many labels.
Concretely, we remove each subtree whose each leaf $n$ is either
  successful or has an ancestor $n'$ such that $n'$ is in the same
  subtree and $\tlab(n) = \tlab(n')$.
\item \label{step:divide}
We extract subtrees from the tree produced in~\ref{step:prune} that
  are potential \emph{candidates} to be subsequently checked.
The extraction of these finite candidate subtrees is done by
  identifying the forest of subtrees rooted in ancestor nodes which do
  not have ancestors themselves.
\item \label{step:check} 
We check that each of the candidate subtrees from~\ref{step:divide} is a witness tree.
\end{enumerate}
If an unsuccessful leaf is found in~\ref{step:build}, then the
considered session types are not related.
In~\ref{step:build}, if the generation of the subtree reached the
bound before reaching an ancestor or a leaf, then the algorithm is
unable to give a decisive verdict, i.e., the result is \emph{unknown}.
Otherwise, if all checks in~\ref{step:check} succeed then the session
types are in the fair asynchronous subtyping relation.
In all other cases, the result is \emph{unknown} because a candidate
subtree is not a witness.

\begin{exa}
  We illustrate the algorithm above with the tree in
  Figure~\ref{fig:simtree}.
After~\ref{step:build}, we obtain the whole tree in the figure (11
  nodes).
After~\ref{step:prune}, all nodes in the dotted boxed are removed.
After~\ref{step:divide} we obtain the (unique) candidate subtree
  contained in the dashed box.
This subtree is identified as a witness subtree in~\ref{step:check},
  hence we have $\exS \subtype \exSR$.
\end{exa}

\paragraph{Soundness of the algorithm}

The soundness of our algorithm w.r.t. fair asynchronous session subtyping relies on proving that 
 given a \emph{witness tree} $(N, n_0, \treetrans, \tlab)$
such that $\tlab(n_0)=(T,S)$, then $T \subtype S$. 
We formalize this in  Theorem \ref{thm:algo-soundess} further down below.

The definition of witness tree consider nestings of input contexts $\calA$.
In the proof of Theorem \ref{thm:algo-soundess} we need the notation $\crepl{\calA^h}{S_j}^{j \in J}$,
to generalize to nestings of input contexts with parametric depth, defined as follows:
\begin{itemize} 
\item $\crepl{\calA^1}{S_j}^{j \in J}$ is $\crepl{\calA}{S_j}^{j \in J}$
\item
$\crepl{\calA^h}{S_j}^{j \in J}$ is $\grepl{\calA}{
\crepl{\calA^{h-1}}{S_j}^{j \in J}
}{J}$, when $h>1$.
\end{itemize}
\noindent 
Given a witness tree for $\calA$, we define a family of isomorphic trees
with labels in which the r.h.s. type has incrementally increased 
nestings of the input context $\calA$ in the growing holes.

\begin{defi}[$h$-th Witness Tree]Given a witness tree $\mathcal T=(N, n_0, \treetrans, \tlab)$ for $\calA$,
  and $h \geq 1$, we inductively define $\mathcal T^h$ as follows:
  \begin{itemize}
  \item
  $\mathcal T^1 = \mathcal T$;
  \item
  for $h>1$,
  given $\mathcal T^{h-1}=(N^{h-1}, n_0^{h-1}, \treetrans^{h-1}, \tlab^{h-1})$
  we define $\mathcal T^{h}=(N^{h}, n_0^{h}, \treetrans^{h}, \tlab^{h})$
  with $N^{h}=N^{h-1}$, $n_0^{h}=n_0^{h-1}$, $\treetrans^{h}=\treetrans^{h-1}$,
  and\\ $\tlab^{h}(n)=\crepl{\grepl{\calA'}{\crepl{\calA^h}{S_j}^{j \in J}}{J}}{S_k}^{k\in K}$
  if $\tlab^{h-1}(n)=\crepl{\grepl{\calA'}{\crepl{\calA^{h-1}}{S_j}^{j \in J}}{J}}{S_k}^{k\in K}$.
  \end{itemize}
\end{defi}
\noindent 
We now present a preliminary Lemma stating that, 
given a witness subtree $\mathcal T$ of a simulation tree,
all the trees in the family  $\mathcal T^h$ faithfully represent the subtyping
simulation game (proof in Appendix \ref{subsec:algorithmsoundness}).
\begin{restatable}{lem}{lemalgosoundness}\label{lem:algo}
  Consider a witness tree $\mathcal T^1=(N^1, n_0^1, \treetrans^1, \tlab^1)$
  contained in a simulation tree.
  For every $h \geq 1$, we have that $\treetrans^{h}$ in 
  $\mathcal T^{h}=(N^{h}, n_0^{h}, \treetrans^{h}, \tlab^{h})$ 
  is compatible with the subtyping simulation game, i.e.,
  $n \treetrans^{h} n'$ is present in $\mathcal T^{h}$
  if and only if there exists a simulation tree 
  $(M, m_0, \treetrans, \tlab)$
  including $m \treetrans^{h} m'$ with $\tlab(m)=\tlab^{h}(n)$
  and $\tlab(m')=\tlab^{h}(n')$.
\end{restatable}

\noindent 
We now move to a proposition stating that, given a witness subtree $\mathcal T$ 
of a simulation tree, we have that all branches in the simulation tree that traverse
$\mathcal T$ follows paths also present in the family of trees $\mathcal T^h$
or in simulation trees $\simtree{T'}{S'}$ where $(T',S')$ is a leaf of $\mathcal T$
for which we know that $T' \subtype S'$ (proof in Appendix \ref{subsec:algorithmsoundness}).
In the statement of this proposition we use $\treetrans{}\!\!^*$ to denote
the reflexive and transitive closure of $\treetrans$.

\begin{restatable}{prop}{propalgosoundness}\label{prop:algo}
  Let $T$ and $S$ be two session types
with $\simtree{T}{S}=(N, n_0, \treetrans, \tlab)$.
If $\simtree{T}{S}$ contains a witness tree $\mathcal T$ with root $n$, then for
  every node $n' \in N$ such that $n \treetrans{}\!\!^*\, n'$ 
  we have that $\tlab(n')$ is a label present either  
  in $\mathcal T^h$, for some $h$,
  or in $\simtree{T'}{S'}=(N', n_0', \treetrans, \tlab')$
  with $T' \subtype S'$.
\end{restatable}
\noindent 
We can now present the main result needed to prove the soundness of our
algorithm.
\begin{restatable}{thm}{thmalgosoundness}\label{thm:algo-soundess}
  Let $T$ and $S$ be session types
s.t.\ $\simtree{T}{S}=(N, n_0, \treetrans, \tlab)$.
If $\simtree{T}{S}$ contains a witness subtree with root $n$ then for
  every node $n' \in N$ s.t.\ $n \treetrans{}\!\!^*\, n'$, either $n'$ is
  a successful leaf, or there exists $n''$ s.t.\
  $n' \treetrans{} n''$.
\end{restatable}

\noindent 
In the light of this last theorem, we can finally conclude that if the candidate subtrees of
$\simtree{T}{S}$ identified with the steps \textbf{S1-3} explained above are
also witness subtrees (check done in the step \textbf{S4}), then we have $T \subtype S$.

\section{Implementation}\label{sec:tool}
To evaluate our algorithm, we have produced a Haskell implementation
of it, which is available on GitHub~\cite{tool}.  It implements a
version of the algorithm presented in Section \ref{sec:algorithm},
which internally represents session types as automata (LTS) (see,
e.g., \cite{BravettiZ21}).
In this context it is also natural to use bisimulation in place of the
syntactic equality for session types.
These design choices helped us to concretise an implementation of the
algorithm in Section ~\ref{sec:algorithm} and allowed us to implement
an optimisation which minimises the input types. We comment on this
below.

Using automata internally makes it easier to identify candidate input contexts
as we can keep track of states that correspond to the input context
computed when applying Case~\eqref{fs:internal} of
Definition~\ref{def:subtyping}.
In particular, we augment each local state in the automata
representation of the candidate supertype with two counters:
the $c$-counter keeps track of how many times a 
state has been used in an input \textit{c}ontext;
the $h$-counter keeps track of how many times a state has occurred
within a \textit{h}ole of an input context.
We illustrate this with Figure~\ref{fig:ctxta} which depicts the
internal data structures our tool manipulates when checking
$\exS \leq \exSR$ from Figures~\ref{fig:runex-types}
and~\ref{fig:runex-big}.
The state indices of the automata in Figure~\ref{fig:ctxta} correspond
to the ones in Figure~\ref{fig:runex-types} (2\textsuperscript{nd}
column) and Figure~\ref{fig:runex-big} (3\textsuperscript{rd} column).

The first row of Figure~\ref{fig:ctxta} represents the root of the
simulation tree, where both session types are in their respective
initial state and no transition has been executed.
We use state labels of the form $\ivar{n}{c}{h}$ 
where $n$ is the original identity of the state, $c$ is the value of
the $c$-counter, and $h$ is the value of the $h$-counter.
The second row depicts the configuration after firing transition
$\snd{\exTM}$, via Case~\eqref{fs:internal} of
Definition~\ref{def:subtyping}.
While the candidate subtype remains in state $0$ (due to a self-loop)
the candidate supertype is unfolded with $\selunf{\exSR}$
(Definition~\ref{def:selectiveUnf}).
The resulting automaton contains an additional state and two
transitions.
All previously existing states have their $h$-counter incremented,
while the new state has its $c$-counter incremented.
The third row of the figure shows the configuration after firing
transition $\snd{\exOVER}$, using Case~\eqref{fs:internal} of
Definition~\ref{def:subtyping} again.
In this step, another copy of state $0$ is added. Its $c$-counter is
set to $2$ since this state has been used in a context twice; and the
$h$-counters of all other states are incremented.

Using this representation, we construct a candidate input context by
building a tree whose root is a state $\ivar{q}{c}{h}$ such that
$c > 1$.
The nodes of the tree are taken from the states reachable from
$\ivar{q}{c}{h}$, stopping when a state $\ivar{q'}{c'}{h'}$ such that
$c'< c$ is found.
A leaf $\ivar{q'}{c'}{h'}$ becomes a hole of the input context.
The hole is a constant ($K$) hole when $h'=c$, and growing ($J$)
otherwise.
Given this strategy and the configurations in Figure~\ref{fig:ctxta},
we successfully identify the context
$\calA = \Tbra{ \exTC : [\, ]^1 , \exDONE : [\,]^2 }$ with $J = \{1\}$
and $K = \{2\}$.

Thanks to our automata representation, it is also possible to minimise
(up-to bisimulation) each session-type automaton \emph{before}
performing Steps~\ref{step:build}-\ref{step:check}.
Concretely our tool accepts an optional command-line flag that turns
on the minimisation of each session type after it has been transformed
into an automaton.
We discuss the benefits of this optimisation in the next section.

\begin{figure}[t]\centering
  \begin{tabular}{c|c|c}
    \toprule
    Last  transition \; & \; State of  $\exS$ \; & Representation of $\exSR$ \\
    \midrule
    $\epsilon$
               &
                 $0$
                         & \;
                           \begin{tikzpicture}[mycfsm, node distance = 0.5cm and 0.9cm
                             ,scale=0.75, every node/.style={transform shape}]
                             \node[state, initial, initial where=above] (s0)
                             {$\ivar{0}{0}{0}$};
                             \node[state, left =of s0] (s1) {$\ivar{1}{0}{0}$};
                             \node[state, left=of s1] (s2)  {$\ivar{2}{0}{0}$};
                             \node[state, left=of s2] (s3)  {$\ivar{3}{0}{0}$};
\node[state, right =of s0] (s4)  {$\ivar{4}{0}{0}$};
                             \node[state, right =of s4] (s5)  {$\ivar{5}{0}{0}$};
\path
                             (s0) edge [bend right] node [above] {$\rcv{\exTC}$} (s1)
                             (s0) edge node {$\rcv{\exDONE}$} (s4)
                             (s1) edge [bend right]  node {$\snd{\exTM}$} (s0)
                             (s1) edge node [above] {$\snd{\exOVER}$} (s2)
                             (s2) edge [loop above] node [left] {$\rcv{\exTC}$} (s2)
                             (s2) edge node [below] {$\rcv{\exDONE}$} (s3)
                             (s4) edge [loop above] node [left] {$\snd{\exTM}$} (s4)
                             (s4) edge node {$\snd{\exOVER}$} (s5)   
                             ; 
                           \end{tikzpicture}  
    \\
    \midrule
    $\snd{\exTM}$
               &
                 $0$
                         & 
                           \begin{tikzpicture}[mycfsm, node distance = 0.5cm and 0.9cm
                             ,scale=0.75, every node/.style={transform shape}]
                             \node[state] (s0) {$\ivar{0}{0}{1}$};
                             \node[state, left =of s0] (s1)  {$\ivar{1}{0}{1}$};
                             \node[state, left=of s1] (s2)  {$\ivar{2}{0}{1}$};
                             \node[state, left=of s2] (s3) {$\ivar{3}{0}{1}$};
\node[state, right =of s0] (s4) {$\ivar{4}{0}{1}$};
                             \node[state, right =of s4] (s5) {$\ivar{5}{0}{1}$};
\node[state, below =of s4,initial] (c0)
                             {$\ivar{0}{1}{0}$};
\path
                             (c0) edge node {$\rcv{\exTC}$} (s0)
                             (c0) edge [right] node {$\rcv{\exDONE}$} (s4)
(s0) edge [bend right] node [above] {$\rcv{\exTC}$} (s1)
                             (s0) edge node {$\rcv{\exDONE}$} (s4)
                             (s1) edge [bend right]  node {$\snd{\exTM}$} (s0)
                             (s1) edge node [above] {$\snd{\exOVER}$} (s2)
                             (s2) edge [loop above] node [left] {$\rcv{\exTC}$} (s2)
                             (s2) edge node [below] {$\rcv{\exDONE}$} (s3)
                             (s4) edge [loop above] node [left] {$\snd{\exTM}$} (s4)
                             (s4) edge node {$\snd{\exOVER}$} (s5)   
                             ; 
                           \end{tikzpicture}  
    \\
    \midrule
    $\snd{\exOVER}$
               &
                 $1$
                         &
                           \begin{tikzpicture}[mycfsm, node distance = 0.5cm and 0.9cm
                             ,scale=0.75, every node/.style={transform shape}]
                             \node[state,white] (s0) {$\ivar{0}{0}{2}$};
                             \node[state, left =of s0,white] (s1)  {$\ivar{1}{0}{2}$};
                             \node[state, left=of s1] (s2)  {$\ivar{2}{0}{2}$};
                             \node[state, left=of s2] (s3) {$\ivar{3}{0}{2}$};
\node[state, right =of s0,white] (s4) {$\ivar{4}{0}{2}$};
                             \node[state, right =of s4] (s5) {$\ivar{5}{0}{2}$};
\node[state, below =of s4] (c0)
                             {$\ivar{0}{1}{1}$};
\node[state, below =of c0,initial] (cc0) {$\ivar{0}{2}{0}$};
\path
                             (cc0) edge node {$\rcv{\exTC}$} (c0)
                             (cc0) edge [bend right] node [left] {$\rcv{\exDONE}$} (s5)
(c0) edge node {$\rcv{\exTC}$} (s2)
                             (c0) edge node {$\rcv{\exDONE}$} (s5)
(s2) edge [loop above] node [left] {$\rcv{\exTC}$} (s2)
                             (s2) edge node [below] {$\rcv{\exDONE}$} (s3)
                             ; 
                           \end{tikzpicture}
    \\
    \bottomrule
  \end{tabular}
  \caption{Internal representation of the simulation tree for $\exS \leq
    \exSR$ (fragment).}\label{fig:ctxta}
\end{figure}

We have run our tool on a dozen of examples handcrafted to test the
limits of our algorithm (inc.\ the examples discussed in this paper),
as well as on the 174 tests taken from~\cite{BCLYZ19}. All of these
tests terminate under a second.

Additionally, for debugging and illustration purposes, the tool can optionally
generate graphical representations of the subtyping simulation game and of
witness trees.

\section{Empirical Evaluation on Synthetic Benchmarks}\label{sec:eval}

\newcommand{\TTop}[3]{T_{\textit{R}}(#1,#2,#3)}
\newcommand{\TBot}[2]{T_{\textit{L}}({#1,#2)}}

\newcommand{\TR}[3]{\mathit{TBran}({#1,#2,#3})}
\newcommand{\TSL}[1]{\mathit{TSelL}({#1})}
\newcommand{\TRL}[1]{\mathit{TBranL}({#1})}
\newcommand{\TS}[2]{\mathit{TSel}({#1,#2})}
\newcommand{\TEST}[3]{\mathit{Test}(#1,#2,#3)}

To evaluate the cost of our algorithm and its implementation, wrt.\
runtime and memory usage, we have performed an empirical evaluation
based on a family of pairs of sub/supertype of increasing sizes.
We perform our evaluation with and without our minimisation-based
optimisation and discuss the results.

\paragraph{Experimental setup}
The family of types we consider is based on variants from our spacecraft example:
the subtype is based on variants of $\exS$ in
Figure~\ref{fig:runex-types}, while the supertype is based on variants of
$\exSR$ in Figure~\ref{fig:runex-big}.
The shape and size of each variant is determined by three parameters
which respectively affect the number of choices in branches (branching
width), the number of inputs that can be accumulated in the supertype
(input depth), and the number of choices in selections (selection
width).

\begin{figure}
\[
    \begin{array}{lcl}
      \TEST{n}{m}{k} & = &  \TBot{n}{k}  \leq       \TTop{n}{m}{k}
      \\[0.2cm]
\TBot{n}{k} & = & \Trec{t} . \TselectindexNoidx{\exTM}{  \Tvar t }{i}{ 1 \leq i \leq k}{\exOVER: \TRL{n}}
      \\[0.2cm]
\TTop{n}{m}{k} & = & \Trec{t} . \TR{n}{m}{k}
      \\[0.2cm]
\TR{n}{m}{k} & = &
                         \begin{cases}
                           \TbranchindexNoidx{\exTC}{\TR{n}{m{-}1}{k}}{i}{ 1 \leq i \leq n }{\exDONE: \TSL{k}} & \text{if } m > 0
                           \\
                           \TbranchindexNoidx{\exTC}{\TS{n}{k}}{i}{ 1 \leq i \leq n }{\exDONE: \TSL{k}} & \text{otherwise}
                         \end{cases}
\\[0.2cm]
\TS{n}{k} & = &
                      \TselectindexNoidx{\exTM}{  \Tvar t }{i}{ 1 \leq i \leq k }{\exOVER: \TRL{n}}
      \\[0.2cm]
\TRL{n} & = &
                    \Trec{t'} . \TbranchindexNoidx{\exTC}{  \Tvar t' }{i}{ 1 \leq i \leq n }{\exDONE: \Tend}
      \\[0.2cm]
\TSL{k} & = &
                    \Trec{t''} . \TselectindexNoidx{\exTM}{  \Tvar t'' }{i}{ 1 \leq i \leq k }{\exOVER: \Tend} 
    \end{array}
  \]
  \caption{Generation of parameterised sub-type/super-type pairs.
    Function $\TTop{n}{m}{k}$ is the super-type and $\TBot{n}{k}$ is the sub-type, where
$n$ is the branching width (the number of messages the type can receive at a given point),
$m$ is the branching depth (the number of messages the type can receive consecutively),
and $k$ is the selection width (the number of messages the type can send at a given point).
  }\label{fig:param-defs}
\end{figure}

Given values $n$, $m$, and $k$ for each of these parameters, we
generate a \emph{subtyping problem} $ \TEST{n}{m}{k}$ as described in
Figure~\ref{fig:param-defs}. We assume that $n \geq 1$, $m \geq 0$,
and $k \geq 1$ --- the branching/selection parameters need to provide
at least one branch, while input depth could be zero (no
anticipation).
Each test applies our algorithm to verify that $\TBot{n}{k}$ is a fair
asynchronous subtype of $\TTop{n}{m}{k}$ (by construction the test
always succeeds).

We describe Figure~\ref{fig:param-defs} in more details.
The subtype $\TBot{n}{k}$ only depends on two parameters: branching
width ($n$) and selection width ($k$).
It is similar to $\exS$ in Figure~\ref{fig:runex-types} except that it
can send (resp.\ receive) different telemetry (resp.\ telecommand)
messages. 
It is a recursive type that immediately chooses between sending one of the $k$
telemetries ($\exTM_i$) then recurse, or send a termination signal ($\exOVER$).
In the latter case, the behaviour continues with $\TRL{n}$, i.e.,
another recursive definition followed by a branching construct where
the type expects to receive either one of the $n$ telecommands
($\exTC_i$) then recurse, or receive the termination signal $\exDONE$.

The supertype $\TTop{n}{m}{k}$ depends on three parameters: branching
width ($n$), input depth ($m$), and selection width ($k$).
This type is similar to $\exSR$ in Figure~\ref{fig:runex-big} but can
send (resp.\ receive) different telemetry (resp.\ telecommand)
messages and allows the reception of $m$ telecommands to
precede the emission of a telemetry message.
$\TTop{n}{m}{k}$ relies on four additional definitions.
$\TR{n}{m}{k}$ encodes the sequence of $m+1$ inputs that can
precede the emission of telemetries.
$\TS{n}{k}$ performs the selections that precede the final series of
inputs in $\TRL{n}$.
$\TSL{k}$ performs the final series of outputs.

\begin{figure}[t]
  \centering
  \begin{minipage}{.25\textwidth}
     \begin{tikzpicture}[mycfsm, node distance = 0.9cm and 1.5cm
      ,scale=0.95, every node/.style={transform shape}]
      \node[state, initial, initial where=above] (s1) {$0$};
      \node[state, below =of s1] (s2) {$1$};
      \node[state, below =of s2] (s3) {$2$};
\path
      (s1) edge [loop left,looseness=50] node [above,yshift=0.2cm] {$\{\snd{\exTM}_1, \snd{\exTM}_2, \snd{\exTM}_3, \snd{\exTM}_4 \}$} (s1)
      (s1) edge  node [right] {$\snd{\exOVER}$} (s2)
      (s2) edge  node [right] {$\rcv{\exDONE}$} (s3)
      (s2) edge [loop left,looseness=50] node [above,yshift=0.2cm] {$\{\rcv{\exTC}_1,  \rcv{\exTC}_2 \}$} (s2)
      ; 
    \end{tikzpicture}
  \end{minipage}
  \begin{minipage}{.6\textwidth}
  \begin{tikzpicture}[mycfsm, node distance = 0.9cm and 2.3cm
      ,scale=0.95, every node/.style={transform shape}]
      \node[state, initial, initial where=above] (s1) {$0_0$};
      \node[state, below =of s1] (s2) {$0_1$};
      \node[state, below =of s2] (s3) {$0_2$};
      \node[state, below =of s3] (s4) {$0_3$};
      \node[ state, left =of s4] (s8) {$4$};
      \node[ state, right =of s4] (s5) {$1$};
      \node[ state, below =of s5] (s6) {$2$};
      \node[ state, below =of s6] (s7) {$3, 5$};
\path
      (s1) edge node [right] {$\{ \rcv{\exTC}_1,\rcv{\exTC}_2\} $} (s2)
      (s2) edge node [right] {$\{ \rcv{\exTC}_1,\rcv{\exTC}_2\} $} (s3)
      (s3) edge node [right] {$\{ \rcv{\exTC}_1,\rcv{\exTC}_2\} $} (s4)
      (s4) edge [] node (rlab) [below, xshift=-0.1cm] {$\{ \rcv{\exTC}_1,\rcv{\exTC}_2\} $} (s5)
(s1) edge [ bend right=90] node [left, near start,yshift=2] {$\rcv{\exDONE}$} (s8)
      (s2) edge [ bend right=70] node [left, near start,yshift=3] {$\rcv{\exDONE}$} (s8)
      (s3) edge [ bend right=45] node [left, near start,yshift=4] {$\rcv{\exDONE}$} (s8)
      (s4) edge [] node [above] {$\rcv{\exDONE}$} (s8)
(s8) edge [ loop below,looseness=50] node [below] {$\{\snd{\exTM}_1, \snd{\exTM}_2, \snd{\exTM}_3, \snd{\exTM}_4 \}$} (s8)
      (s5) edge [ bend right=90] node [right] {$\{\snd{\exTM}_1, \snd{\exTM}_2, \snd{\exTM}_3, \snd{\exTM}_4 \}$} (s1)
(s8) edge  [ bend right=20] node [above,yshift=0.2cm] {$\snd{\exOVER}$} (s7)
      (s5) edge  [] node [thick, right] {$\snd{\exOVER}$} (s6)
      (s6) edge  [] node [thick, right] {$\rcv{\exDONE}$} (s7)
(s6) edge [ loop right,looseness=50] node [right] {$\{\rcv{\exTC}_1,  \rcv{\exTC}_2 \}$} (s6)
      ;
      \begin{pgfonlayer}{background}        
        \node[draw,densely dotted,fit=(s1) (s4) (rlab), inner sep = 3pt,fill=black!8] {};
      \end{pgfonlayer}
\end{tikzpicture}
  \end{minipage}
  \caption{Minimised versions of $\TBot{2}{4}$ (subtype, left) and
    $\TTop{2}{3}{4}$ (supertype, right).}\label{fig:param-sup-min}
\end{figure}

Figure~\ref{fig:param-sup-min} gives a graphical representation of the
session-type automata generated by the definitions in
Figure~\ref{fig:param-defs} \emph{after} minimisation up to
bisimulation.
The figure shows a subtype (left) that can send four different $\exTM_i$
messages ($k=4)$, then can receive two different $\exTC_i$ messages
($n=2$). The state labels correspond to the ones of $\exS$ in
Figure~\ref{fig:runex-types}. 

The supertype (right) is more complex. It can also send four different
$\exTM_i$ messages ($k=4)$, and receive two different $\exTC_i$
messages ($n=2$). Additionally, it may postpone the emission of
telemetries and receive up to $4$ telecommands first ($m+1=4$).
The state labels correspond to the ones of $\exSR$ in
Figure~\ref{fig:runex-big}. Note that because of minimisation the two
final states of $\exSR$ are merged into their $3,5$ counterpart in
Figure~\ref{fig:param-sup-min}.
Since the emission of $\exTM_i$ in $\TTop{2}{3}{4}$ is further
postponed compared to $\exSR$, we also obtain several variants of
state $0$, labelled by $0_i$ and highlighted in gray in
Figure~\ref{fig:param-sup-min}.

\begin{figure}
  \begin{minipage}{.48\textwidth}
  \includegraphics[width=\textwidth]{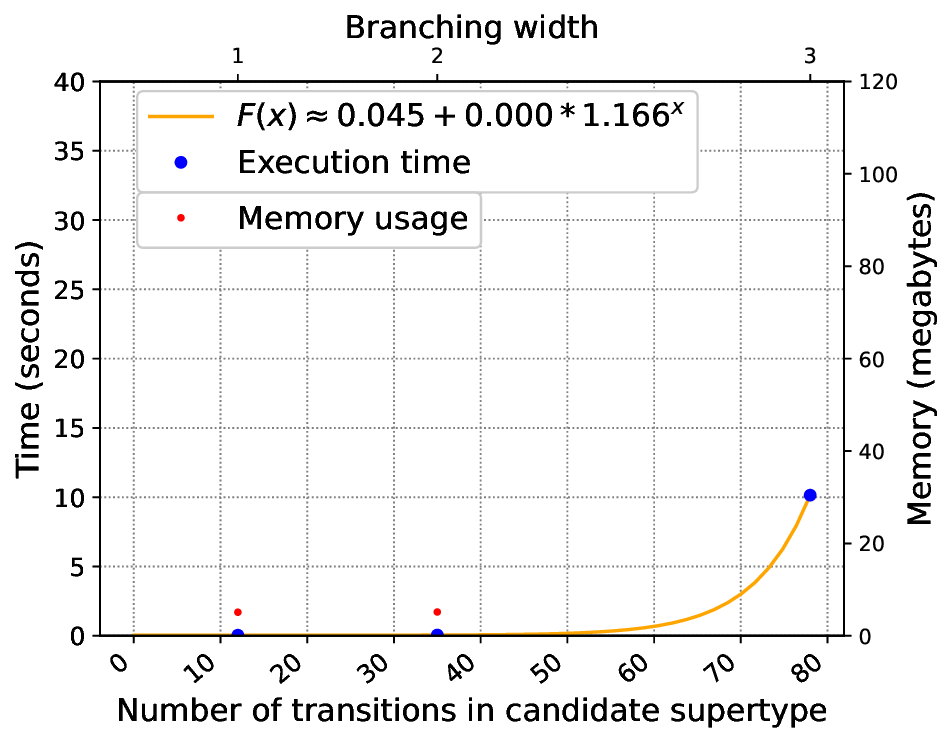}
\end{minipage}
\begin{minipage}{.48\textwidth}
  \includegraphics[width=\textwidth]{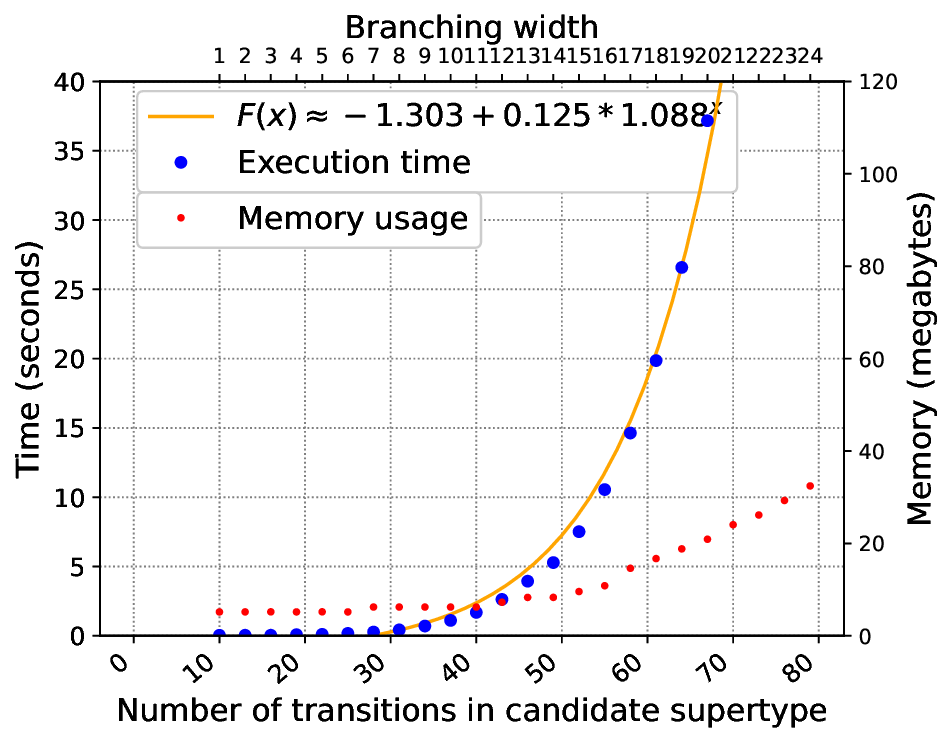}
\end{minipage}
\caption{Increasing branching width, without (left) and with minimisation (right)}\label{fig:exp-bra-w}
\begin{minipage}{.48\textwidth}
  \includegraphics[width=\textwidth]{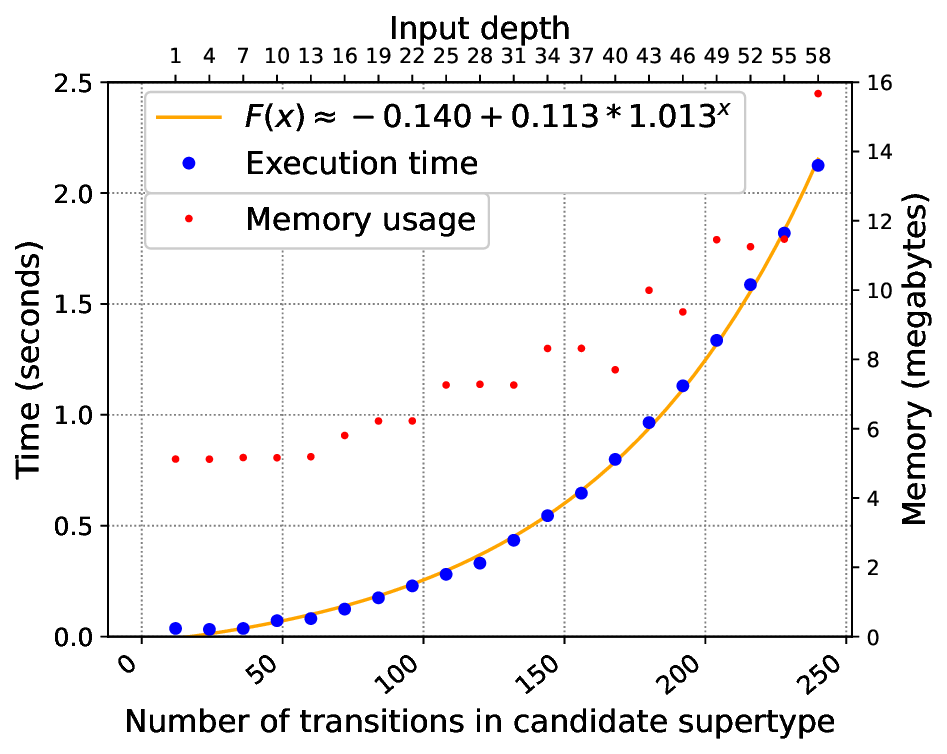}
\end{minipage}
\begin{minipage}{.48\textwidth}
  \includegraphics[width=\textwidth]{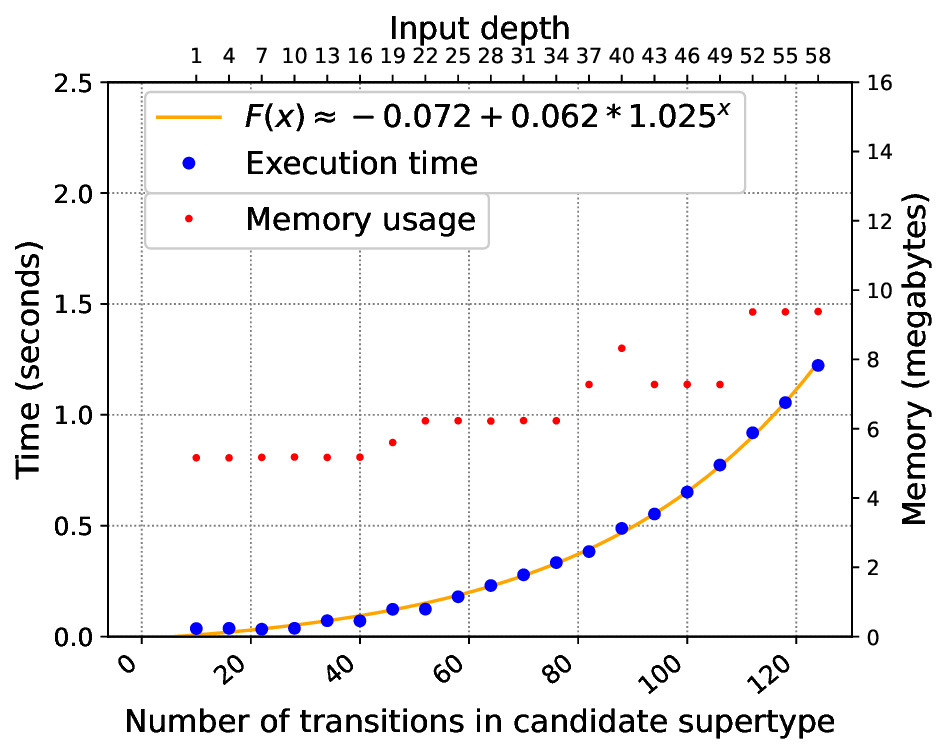}
\end{minipage}
\caption{Increasing input depth, without (left) and with minimisation (right)}\label{fig:exp-in-d}
\begin{minipage}{.48\textwidth}
  \includegraphics[width=\textwidth]{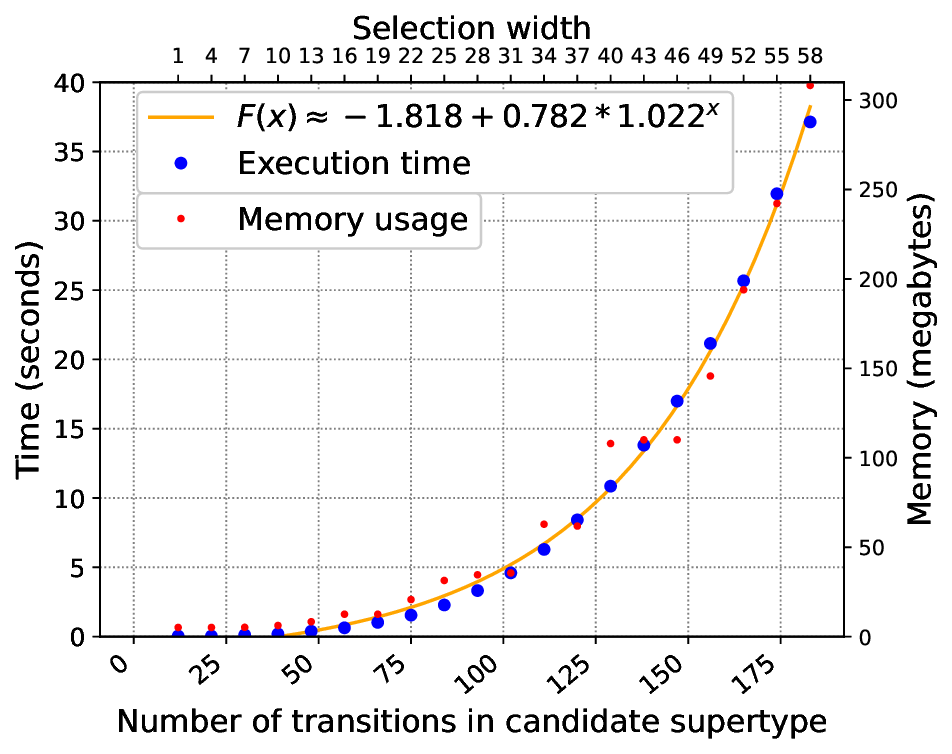}
\end{minipage}
\begin{minipage}{.48\textwidth}
  \includegraphics[width=\textwidth]{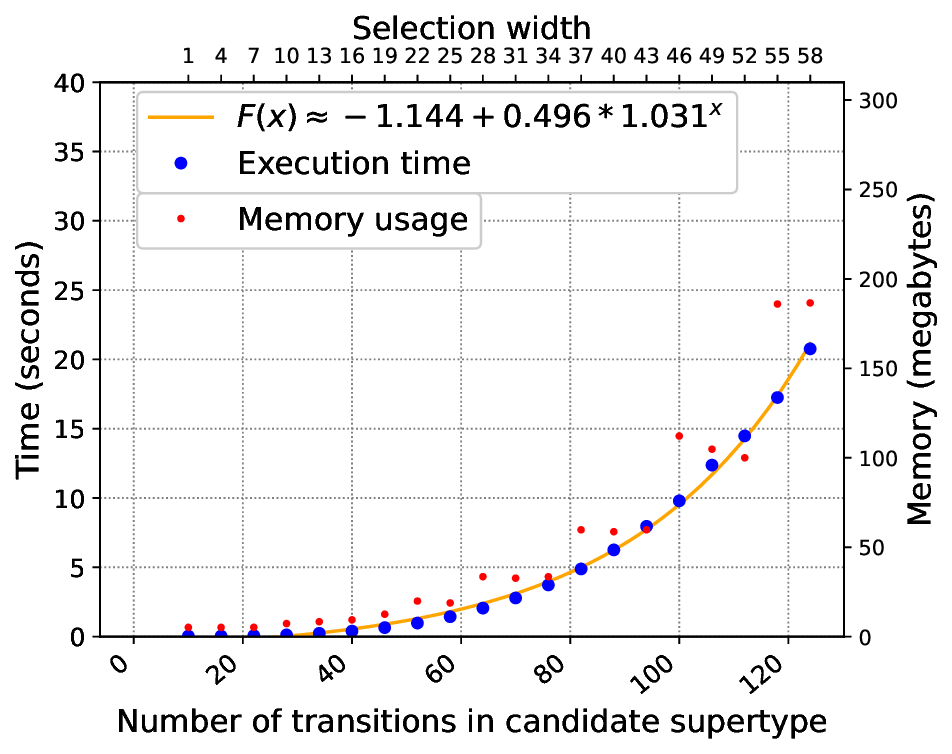}
\end{minipage}
\caption{Increasing selection width, without (left) and with minimisation (right)}\label{fig:exp-sel-w}
\end{figure}

\paragraph{Experimental results}
Figures~\ref{fig:exp-bra-w},~\ref{fig:exp-in-d},
and~\ref{fig:exp-sel-w} give the results of running the implementation
of our algorithm on increasingly large instances of the subtyping
problem $\TEST{n}{m}{k}$.
Each figure shows the runtime (larger data points in blue, left
y-axis) and peak memory usage (smaller data points in red, right
y-axis) for each instance of the problem.
Each figure includes two x-axes: the bottom one represents the number
of transitions in the automata representation of the candidate
supertype (which we consider a good measure of the size of the subtyping
problem); the top one represents the value of the variable parameter
for each experiment (e.g., branching with).
Plots on the left show the result without minimisation, plots on the
right show results using minimisation up to bisimulation.
Each figure depicts 20 data points unless our implementation timed out
(more than 300 seconds).
The yellow curve highlights the runtime trend. It is
computed using SciPy's \texttt{curve\_fit} function.

All the benchmarks in this paper were run on a MacBook Pro with an
Intel i5 CPU with 16GB RAM running macOS 13.4. The time was measured
by taking the difference between the system clock before and after our
tool was invoked. The memory usage refers to the maximum resident set
size as reported by the \texttt{/usr/bin/time -l} command. Each test
was ran 3 times, the plots report the average time (resp. memory)
measurements. All our test data and infrastructure are available on
our GitHub repository~\cite{tool}.

Figure~\ref{fig:exp-bra-w} shows the result of checking
$\TEST{n}{1}{1}$, with $n$ (branching width) increasing by step of
$1$, from $1$ to $20$.
The left-hand side plot shows that the tool quickly runs out of
resource without optimisation: only $n \in \{1,2,3\}$ terminate in
reasonable time.
While the asymptotic cost of the algorithm with minimised automata is
still exponential, the tool can deal with much larger input using this
optimisation as show on the right.

Figure~\ref{fig:exp-in-d} shows the result of checking
$\TEST{1}{m}{1}$, with $m$ (input depth) increasing by step of $3$,
from $1$ to $58$ ($20$ data points).
Observe that minimisation nearly halves the number of transitions in
the candidate supertypes. As a consequence, the version of the tool
that minimises its input before applying the subtyping algorithm runs
much faster and uses much less memory than its non-optimised
counterpart.

Figure~\ref{fig:exp-sel-w} shows the result of checking
$\TEST{1}{1}{k}$, with $k$ (selection width) increasing by step of
$3$, from $1$ to $58$ ($20$ data points).
In this case minimisation has a lesser effect on the number of
transitions in the candidate supertypes, but it has still a
significant effect on runtime, e.g., the largest problem takes 20s on
the minimised automata and 37s on the non-minimised ones.

\section{Related and Future Work}\label{sec:ending}
\label{sec:concl}
\paragraph{Related work}
The relationship between refinement and subtyping in the context of
\emph{synchronous} session types has been thoroughly investigated both 
for binary and multiparty session types.
For instance, Bernardi and Hennessy~\cite{BernardiH16} establish a
correspondence between binary session subtyping and an observational
preorder on session types interpreted as contracts. A similar result
has been obtained in the context of multiparty session types by Severi
and Dezani-Ciancaglini~\cite{SeveriD19}, where the subtyping is dubbed
structural preorder, while the refinement is named observational
preorder.
Concerning \emph{asynchronous} communication we can mention
previous works on refinement for asynchronous 
communication by some of the authors of this paper.
The work in~\cite{wsfm08} also considers fair compliance,
however here we consider binary (instead of multiparty) communication and we
use a unique input queue for all incoming messages instead of 
distinct named input channels.
Moreover, in the present paper we provide a sound characterisation of fair refinement
using coinductive subtyping and
provide a sound algorithm and its implementation.  In~\cite{sefm19,BravettiZ21} 
the asynchronous subtyping of \cite{MY15} 
is used to characterise refinement for a notion 
of correct composition based on the impossibility
to reach a deadlock, instead of the possibility to reach a final 
successful configuration as done in the present paper.
The  refinement from~\cite{sefm19}
does not support examples such as those in Figure~\ref{fig:runex-types}.

Concerning 
fairness in the context of session types,
Padovani studied a notion of fair subtyping for \emph{synchronous}
multi-party session types in~\cite{Padovani16}. This work notably
considers the notion of \emph{viability} which corresponds,  in the synchronous multiparty setting, to our
notion of controllability.
We use the term controllability instead of viability following the tradition
of service contract theories like those based on Petri nets \cite{Loh08,Wei08}
or process calculi \cite{BZ09a}.
Compared to~\cite{Padovani16},
asynchronous communication makes it much more involved to prove
soundness and completeness of the decidable characterisation of
controllability, as we do in this paper. Indeed in the asynchronous
case, transition systems arising from the communication of two types
are, in general, infinite state (due to unbounded queues), while they
are always finite state in the synchronous case.
Fair refinement in~\cite{Padovani16} is characterised by defining
a coinductive relation on
normal form of types, obtained by removing inputs leading to uncontrollable continuations.
Instead of using normal forms, we remove these inputs
during the asynchronous subtyping check. 
A limited form of variance on output
is also admitted in~\cite{Padovani16}.
Covariance between the outputs of a subtype and 
those of a supertype is possible when the additional
branches in the supertype are not needed to
have compliance with potential partners.
In~\cite{Padovani16} this check is made possible
by exploiting a \emph{difference} operation~\cite[Definition 3.15]{Padovani16}
on types, which synthesises a new type representing branches of one type that are absent in the other.
We observe that 
the same approach cannot work to 
introduce variance on outputs in an asynchronous setting.
Indeed the interplay between output anticipation and recursion
could generate differences in the branches of a 
subtype and a supertype that cannot be statically 
represented by a (finite) session type.

Padovani also studied an alternative notion of fair \emph{synchronous} subtyping
in~\cite{Padovani13}. Although the contribution of that paper refers to
session types, the formal framework therein seems to deviate from the usual session type approach. In particular, it considers
shared channel communication instead of binary channels: when a
partner emits a message, it is possible to have a race among several
potential receivers for consuming it.  As a consequence of this
alternative semantics, the subtyping in~\cite{Padovani13} does not
admit variance on input. Another difference with respect to session
type literature is the notion of \emph{success} among interacting
sessions: a composition of session is successful if at least one
participant reaches an internal successful state.  This approach has
commonalities with testing \cite{DH84}, where only the test composed
with the system under test is expected to succeed, but differs from
the typical notion of success considered for session types.
In~\cite{Barbanerad10,BernardiH16} (resp.\ \cite{MariangiolaPreciness})
it was proved that the Gay-Hole synchronous session subtyping
(resp. orphan message free asynchronous subtyping) coincides with
refinement induced by a successful termination notion requiring interacting processes to be {\it both} in the $\Tend$ state
(with empty buffers, in the asynchronous case).

More recently, van Glabbeek et al.~\cite{GlabbeekHH21} introduce a
type system for multiparty sessions that assumes
fairness. Nevertheless, the notion of fairness used in that paper is
different with respect to the notion considered by
Padovani~\cite{Padovani16} (in the synchronous case) and in this paper
(in the asynchronous case). In fact, in \cite{GlabbeekHH21}
\emph{weak} fairness is considered, consisting of a minimal fairness
assumption that ``guarantees only that concurrent transitions cannot
prevent each other from happening''.
On the other hand, Padovani~\cite{Padovani16} and ourselves consider a stronger notion
of fairness, namely, according to the terminology in \cite{GlabbeekH19},
we consider the composition of two session types correct
if their successful termination is a liveness property which holds
under the assumption of full fairness.
In \cite{GlabbeekH19} it is proved that, for finite state transition
systems, full fairness collapses to strong fairness of transitions, i.e., 
a transition which is (relentlessly) enabled infinitely many times during 
a computation, it is also executed infinitely often in such computation.
Session types are finite states, but we consider asynchronous
communication via unbounded FIFO buffers, hence our transition system 
(Definition \ref{def:transrel}) describing
the composition of two session types is not finite because buffers can store 
an unbounded amount of messages. On the contrary, in the context of
synchronous communication the transition system describing the composition 
of two session types is finite state, hence the above correspondence
result between full fairness and strong fairness applies.
A strong fair session subtyping has been recently used in a type
system that guarantees fair termination of sessions for a $\pi$-calculus like
language with binary sessions~\cite{CicconeP22}. The subtype defined in that paper differs
from previous strong fair subtypings because it also deals with higher-order types (useful
to type process languages including primitives for session creation and delegation) 
and because it is only sound but not complete w.r.t.\ fair session type
refinement. More precisely, it is complete only for bounded processes and it does
not capture subtypes like those discussed in Example \ref{ex:contra-variance}, where
the supertype has an uncontrollable (infinite) branch.

Several variants of asynchronous session subtyping have been proposed
in~\cite{ESOP09,MariangiolaPreciness, CDY2014, MY15, GhilezanPPSY21} and further studied in
our earlier work~\cite{BravettiCZ17,BravettiCLYZ21,sefm19,BCLYZ19}. All these
variants have been shown to be
undecidable~\cite{BCZ18,LY17,BravettiCZ17}.
Moreover, all these subtyping relations are (implicitly) based on an unfair notion of compliance. 
Some of these papers consider binary session types~\cite{MariangiolaPreciness, CDY2014, MY15}
as we do in this paper. 
An interesting technical difference with these papers is that they use finite input contexts
(i.e. without recursion) while we also consider infinite input contexts which may contain recursion
--- this is necessary to 
obtain 
$\exGR \subtype \exG$  and 
$\exS \subtype \exSR$ (see Figures~\ref{fig:runex-types} and~\ref{fig:runex-big}).
Moreover, the 
papers~\cite{MariangiolaPreciness, CDY2014} impose additional constraints in the
definition of asynchronous subtyping to guarantee 
absence of orphan-messages. Such constraints require the subtype not to have output loops whenever an output anticipation 
is performed, thus guaranteeing that at least one input
is performed in all possible paths.
In this paper, absence of orphan messages between compatible types
is guaranteed as successful termination is enforced under the assumption
of full-fairness.
Notice that not imposing this orphan-message-free
constraint is consistent with our recursive input contexts
that allows for input loops in the supertype whenever an output anticipation 
is performed.
The other papers~\cite{ESOP09, GhilezanPPSY21} consider asynchronous subtyping
for multiparty session types. In the binary case, a subtype can only anticipate
(under some specific conditions) outputs w.r.t input. In the multiparty
context additional differences are allowed, for instance, a subtype can anticipate
also an input w.r.t. other inputs of messages coming from other partners. Intuitively, this 
is possible because in the considered operational model messages coming from 
different partners are stored in distinct message queues.
A difference between~\cite{ESOP09} and~\cite{GhilezanPPSY21}
is that the former concentrates on deadlock freedom, while
the latter considers also orphan message freedom.
Notably, the subtyping in~\cite{GhilezanPPSY21} is proved to 
be precise (i.e. sound and complete), w.r.t. a notion of refinement 
that preserves orphan message freedom, deadlocks, and starvation,
for a $\pi$-calculus like language with multiparty sessions.  

In~\cite{BCLYZ19,BravettiCLYZ21}, we proposed a sound algorithm for the
(unfair) asynchronous subtyping in~\cite{MariangiolaPreciness}. 
The sound algorithm that we present in this paper substantially differs from that
of~\cite{BCLYZ19,BravettiCLYZ21}. Here we use witness trees that take under consideration both
increasing and decreasing of accumulated input.
In~\cite{BCLYZ19,BravettiCLYZ21}, instead, only regular growing accumulation is considered.
It is worth mentioning that in the context of multiparty session types there exist
alternative sound (but not complete) algorithmic approaches. 
In particular, in~\cite{DagninoGD23} a multiparty
approach is adopted: they study properties of networks of communicating 
end-point types instead
of studying a subtyping relation on binary session types in isolation, 
as we do in this paper. 
A first phase of their algorithm infers global 
types from networks, and a second phase checks the well formedness of the inferred
global types. Using techniques similar to ours (i.e. reduction from
queue machines) well formedness is proved to be undecidable, but a sound algorithmic
characterisation is proposed which is based on the notion of balancing.
The authors of that paper show that, following their approach, 
one of the examples not captured by the algorithm in~\cite{BCLYZ19,BravettiCLYZ21}
can be managed.

Finally, we mention work about refinement/subtyping in the context of
asynchronous multiparty sessions, where the use of global types 
allows for the definition of decidable type systems. More precisely,
both Castellani et al.~\cite{CastellaniDG21} and Li et al.~\cite{LiSW24} study a 
notion of refinement for
(asynchronous) multiparty session types that ensures that the
implementation of a given role can be replaced by another in the
context of a specific global type. This means that the relation
considers not only the component being refined, but also the other
components of the system. Unlike most subtyping relation for
asynchronous session types, this relation is decidable --- this is notably 
due to the relation being restricted to the specific context of a given 
global type.

\paragraph{Future work}
In future work, we will investigate the possibility to characterize 
a notion of fair asynchronous session subtyping 
which is complete with respect to our notion of fair refinement,
in particular, we are interested in a less restrictive 
subtyping which includes also some form of output variance.
We also plan to lift our study of fairness from binary to multiparty session types; 
in fact, the notions of fair compliance and refinement extend naturally 
to several partners.
Finally, we will investigate a more refined termination condition for our
algorithm using ideas from~\cite[Theorem 3.8]{BravettiCLYZ21}.
In particular, we plan to identify conditions similar to those
in Definition~\ref{def:witness-tree} such that it is always guaranteed
to find, during the computation of each branch of the simulation tree, a node 
with an ancestor satisfying such conditions.
Then, the initial phase of the algorithm dedicated to the
identification of the candidate
subtrees can terminate when such nodes are
detected, and the subsequent phase will continue to check
whether such candidate subtrees are also witness subtrees.

\section*{Acknowledgments}
We thank the anonymous reviewers for their valuable feedback and
insightful suggestions, which have improved the quality of this work.

\bibliographystyle{alphaurl}
\bibliography{biblio}

\iflong

\newpage 

\appendix

\section{Proofs}\label{app:proofs}
\subsection{Undecidability of Fair Refinement}\label{subsec:queuemachines}

Let $T=\sem{M,q_f,E}$ and $S=\sem{M,E}$; we
have that $T \refine S$ if and only if $q_f$ is reachable in $M$. 
To prove this,
we first characterize the set of types that are compliant with $S$.

\bisimilarityAndDual*
\begin{proof}
Let $S=\sem{M,E}$.

We first prove the if part. Let $S'$ be a session type 
with input/output labels in $\Gamma \cup \{E\}$ s.t. $S' \sim \dual{S}$.
We now prove that $S'$ is compliant with $S$.
It is trivial to see that $\dual S$ is compliant
with $S$; this holds because in the configuration $\cnfg{S}{\lempty}{\dual S}{\lempty}$
the two parties alternate inputs and outputs in such a way
that their buffers have maximal length 1, and moreover the possibility to
successfully terminate by selecting the ending label $E$ is never disallowed.
By Corollary \ref{cor:bisimilar} we have that also all types $S' \sim \dual{S}$
are compliant with $S$.

We now move to the only-if part.
Let $S'$ be a session type 
with input/output labels in $\Gamma \cup \{E\}$ s.t. 
$S'$ is compliant with $S$, i.e.,
$\cnfg{S}{\lempty}{S'}{\lempty}$ is a correct composition.
We have that $\unf{S'}$ cannot start with an output selection; in fact, if, for instance,
it starts with an output selection and it selects any label $A$, the type $S$ can select 
a branch with a different label $A'$, thus blocking. 
The initial input branching of $\unf{S'}$ must have branchings
labeled with all the symbols in $\Gamma$ plus the ending symbol $E$,
in that these are the labels that can be initially selected by $S$.
In each continuation of $S'$, the unfolding of the type should start with
an output selection, otherwise the entire
system is blocked in that the continuation of $S$ after the initial output
selection starts with an input branching.
Moreover, given that these input branchings of the continuation of $S$ have only 
the initially selected label, the output selection in the
continuation of $S'$ can have only such label. 
After each of these output selections of the continuation of $S'$,
the same reasoning can be applied, excluding the case in which the label $E$
was initially selected. In this case, the continuation of $S'$ should be
such that its unfolding is $\Tend$. This because, the continuation of $S$
becomes $\Tend$ after executing the input branching labeled with $E$.
These constraints that we have just proved holding for the type $S'$ 
guarantee that $S' \sim \dual{S}$.
\end{proof}
\noindent 
In order to prove the undecidability of refinement,
we first show that $T$ is compliant with $\dual S$
if and only if $q_f$ is reachable in $M$.

\encodingQueue*
\begin{proof}
Consider the queue machine $M$, the types $T=\sem{M,q_f,E}$ and $S=\sem{M,E}$
and the initial configuration 
$\cnfg{\semT{s}{\emptyset}}{\lempty}{\dual S}{\lempty}$.
The first transition is 
$\cnfg{T}{\lempty}{\dual S}{\lempty} \trans{}
\cnfg{\semT{s}{\emptyset}}{\lempty}{\dual S}{\$}$.

We now define a partial mapping function $\semcurly{\,}$ from configurations 
(reachable from the initial configuration
$\cnfg{\semT{s}{\emptyset}}{\lempty}{\dual S}{\$}$) to configurations
in the queue machine computation:
\begin{itemize}
\item
$\semcurly{\cnfg{\semT{q}{\emptyset}}{\word_T}{S'}{\word_S'}}=
(q,\word_T\cdot \omega \cdot(\word_S')^R)$ where 
\begin{itemize}
\item
$\omega = \lempty$
if $S'$ starts with an input branching, or $\omega = A$ if 
$S'$ starts with an output selection with unique label $A$,
\item
the operator $\cdot$ stands for concatenation, and
\item
and $\beta^R$ is the reverse of $\beta$. 
\end{itemize}
\end{itemize}
Notice that $\semcurly{\cnfg{\semT{s}{\emptyset}}{\lempty}{\dual S}{\$}}$ is defined
and it coincides with the initial 
configuration of the queue computation $(s,\$)$. 
In the following we use the following notation: 
\begin{itemize}
\item
$
\cnfg{\semT{q}{\emptyset}}{\word_T}{S'}{\word_S'}
\Rightarrow
\cnfg{\semT{q'}{\emptyset}}{\word_T'}{S''}{\word_S''}
$
if 
\begin{itemize}
\item
$
\cnfg{\semT{q}{\emptyset}}{\word_T}{S'}{\word_S'}
\trans{}^*
\cnfg{\semT{q'}{\emptyset}}{\word_T'}{S''}{\word_S''}
$
and 
\item
all intermediary traversed configurations are not
in the domain of the partial mapping function $\semcurly{\,}$.
\end{itemize}
\end{itemize}
\noindent 
Given that, excluding the final state $q_f$, for each state $q$ 
of the queue machine $\semT{q}{\emptyset}$ reproduces the dequeue/enqueue
actions of state $q$ and
$\dual S$ is a simple forwarder that repeatedly produces and consumes
the same labels, we have that 
given $q \neq q_f$ we have 
$(q,\gamma)\rightarrow _{M}(q',\gamma')$
if and only if
$
\cnfg{\semT{q}{\emptyset}}{\word_T}{S'}{\word_S'}
\Rightarrow
\cnfg{\semT{q'}{\emptyset}}{\word_T'}{S''}{\word_S''}
$
with 
$\semcurly{\cnfg{\semT{q}{\emptyset}}{\word_T}{S'}{\word_S'}}=(q,\gamma)$
and 
$\semcurly{\cnfg{\semT{q'}{\emptyset}}{\word_T'}{S''}{\word_S''}}=(q',\gamma')$.

We now prove the only-if part of the theorem.
Assume that $T$ is compliant with $\dual S$.
This means that there exists a computation leading to the
final successful configuration. The unique occurrence of $\Tend$
is inside the type $\semT{q_f}{\mathcal S}$,
hence we have 
$\cnfg{\semT{s}{\emptyset}}{\lempty}{\dual S}{\$}
\Rightarrow \ldots
\Rightarrow
\cnfg{\semT{q_f}{\emptyset}}{\word_T}{S'}{\word_S'}
$ 
thus implying that
state $q_f$ is reachable in $M$.

We now prove the if part. Assume that $q_f$ is reachable in $M$.
Consider $\cnfg{\semT{s}{\emptyset}}{\lempty}{\dual S}{\$} \trans{}^* 
\cnfg{T'}{\word_T'}{S'}{\word_S'}$.
There are two possible cases: either (i) it is possible to extend the 
sequence of transitions as follows
$\cnfg{T'}{\word_T'}{S'}{\word_S'} \trans{}^*  
\cnfg{\semT{q}{\emptyset}}{\word_T''}{S''}{\word_S''}$, for some state $q$,
(ii) or during the sequence of transitions 
$\cnfg{\semT{s}{\emptyset}}{\lempty}{\dual S}{\$} \trans{}^* 
\cnfg{T'}{\word_T'}{S'}{\word_S'}$ a configuration is traversed in 
which the l.h.s. type is $\semT{q_f}{\emptyset}$.

In the first case (i), we have that 
$(s,\$) \rightarrow _{M}^* \semcurly{\cnfg{\semT{q}{\emptyset}}{\word_T''}{S''}{\word_S''}}$;
moreover, in this computation of the queue machine the state $q_f$ is not 
traversed. This means that such a queue machine computation can be
extended to reach $q_f$, hence
the sequence of transitions 
$\cnfg{\semT{s}{\emptyset}}{\lempty}{\dual S}{\$} \trans{}^*
\cnfg{\semT{q}{\emptyset}}{\word_T''}{S''}{\word_S''}$
can be additionally extended to reach a configuration where the
l.h.s. type is $\semT{q_f}{\emptyset}$. From such a configuration,
we have that there are only finitely many transitions leading
to the final successful configuration (in this final transitions
both the queues are emptied and both types become $\Tend$).

In the second case (ii), we have that a configuration whose
l.h.s. type is $\semT{q_f}{\emptyset}$. As just observed,
this means that the
configuration $\cnfg{T'}{\word_T'}{S'}{\word_S'}$ is an intermediary
configuration in the final sequence of transitions 
leading
to the final successful configuration (in which
both the queues are emptied and both types are $\Tend$).
\end{proof}
\noindent 
By combining Theorem \ref{th:encodingQueue} with Lemma \ref{lem:dualQueue}, 
we can finally prove that our encoding of queue machines into session types
correctly reduces state reachability into refinement.

\correfinementundec*
\begin{proof}
We first prove the only-if part. Let $T \refine S$. By Lemma \ref{lem:dualQueue}
we have that $S$ is compliant with $\dual S$. Given that $T \refine S$,
also $T$ is compliant with $\dual S$. By Theorem \ref{th:encodingQueue} this implies
that $q_f$ is reachable in $M$.

We now prove the if part. Assume that $q_f$ is reachable in $M$. 
As discussed in Section~\ref{sec:refinement} (see footnote \ref{foot:labels}) 
our encoding of queue machines
assumes that the set $\mathcal{L}$ of labels in the Definition \ref{def:sessiontypes} of session types includes the symbols 
in the queue machine
alphabet $\Gamma$ plus the symbol $E$. We now consider a queue machine
$M'=(Q' , \Sigma , \Gamma'\supseteq \Gamma , 
 \$ , s , \delta' \supseteq \delta)$ 
obtained by replacing the queue alphabet $\Gamma$
with a richer alphabet $\Gamma'$ such that
$\mathcal{L} = \Gamma' \cup \{E\}$,
and by extending $\delta$ with 
a new transition relation $\delta'$ which includes 
also the additional queue symbols in its domain. 
The behaviour of $\delta'$ on these additional symbols
is irrelevant because these symbols will never be placed
in the queue, given that the input alphabet is still $\Sigma$.
We have that $q_f$ is reachable in $M'$, simply because 
$M'$ reproduces the same computations of $M$.
By Theorem \ref{th:encodingQueue} 
we have that $T$ is compliant with $\dual S$. By Corollary \ref{cor:bisimilar} we have that
$T$ is compliant with all $S'$ such that $S' \sim \dual S$.
Under the assumption that $\mathcal{L} = \Gamma' \cup \{E\}$,
by Lemma \ref{lem:dualQueue} we have that the set of
types $S'$ such that $S' \sim \dual S$ precisely corresponds
with the types with which $S$ is compliant. We have observed 
that $T$ is compliant with all such $S'$, hence we can conclude that
$T \refine S$.
\end{proof}

 \subsection{Controllability Characterisation}\label{subsec:controlcharacter}

In this section we will prove the following theorem about controllability characterisation.

\thcontrolcharacter*

\noindent 
We start by introducing some notions and definitions that will be needed in the proof.

First of all we present an equivalent definition, based on purely structural induction, of the $\unicontro$ predicate introduced in Definition \ref{def:controllability} characterizing session type controllability.
\begin{defi}\label{def:unicontropure}
  Given a session type $T$, we define the judgment $T\unicontro$
  inductively as follows:
  \begin{mathpar}
    \inferrule
    {\,}
    {\Tvar t \unicontro}

    \inferrule
    {\,}
    {\Tend \unicontro}
    
    \inferrule
    {
      \Tend \!\in\! T \vee \exists \Tvar t' \!:\!  \Tvar t' \!\neq\! \Tvar t \wedge \Tvar t' \!\in\! \free{T}
      \and 
      T \unicontro
    }
    {\Trec t.T \unicontro}
\\
    \inferrule
    {
      T\unicontro
    }
    { \Tbranchsingle{l}{T}  \unicontro}

        \inferrule
    {
      \forall i \in I \dst T_i \unicontro
    }
    {  \Tselect{l}{T}   \unicontro}
  \end{mathpar}
where $\free{T}$ is the set of variables $t$ occurring free in $T$.
\end{defi}

In the following we will use a reformulation of session types in terms of equation sets. In equation set notations we
will use terms $T$ that have the same syntax as those used to denote session types, excluding the $\Trec t.\_$ recursion operator. Notice that in such notations we consider possibly open terms $T$ (i.e.\ such that $\free{T}$ is not empty). Session types are, thus, denoted by $T\{\Tvar t=T_{\Tvar t} \mid \Tvar t \in \mathsf{Vars} \}$, with $\mathsf{Vars}$ being a set of variables $\vt$ that includes all variables in $\free{T}$ and also in $\free{T_\vt}$ for all $\vt \in \mathsf{Vars}$.

Formally, given a session type $T$ (we assume with loss of generality that each of its recursions uses a variable with a different name) we consider its equivalent equation set notation $\mathsf{esn}(T)= T_{\mathsf{init}}\{\Tvar t=T_{\Tvar t} \mid \Tvar t \in \mathsf{Vars} \}$, defined as follows:
\begin{itemize}
\item $\mathsf{Vars}$ is the set of variable names used in the recursions of $T$
\item $T_\mathsf{init}$ is the only term without recursion operators satisfying: there exists a set of terms $T'_\vt$, one for each variable $\vt \in \free{T_\mathsf{init}}$, such that $T_\mathsf{init}\{T'_\vt/\vt \mid \vt \in \free{T_\mathsf{init}}\}=T$
\item each $T_{\Tvar t}$, with $\Tvar t \in \mathsf{Vars}$, is the only term without recursion operators satisfying: there exists a 
set of variables $\mathsf{Vars}_\vt \subseteq \free{T_{\Tvar t}}$ and a
set of terms $T'_{\vt'}$, one for each variable $\vt' \in \mathsf{Vars}_\vt$, such that $T_{\Tvar t}\{T_{\vt'}/\vt' \mid \vt' \in \mathsf{Vars}_\vt \}=T''$ with $\Trec t.T''$ occurring in $T$.
\end{itemize}

\begin{defi}[Unfolding]Given session type in equation set notation we define its unfolding $\unf{T\{\Tvar t=T_{\Tvar t} \mid \Tvar t \in \mathsf{Vars} \}}$  
as follows: 
\[
\unf{T\{\Tvar t=T_{\Tvar t} \mid \Tvar t \in \mathsf{Vars} \}} = 
\begin{cases}
  \unf{T_{\vt'} \{\Tvar t=T_{\Tvar t} \mid \Tvar t \in \mathsf{Vars} \}} & \text{if $T=\vt'$}
  \\
  T\{\Tvar t=T_{\Tvar t} \mid \Tvar t \in \mathsf{Vars} \} & \text{otherwise}
\end{cases}
\]
\end{defi}
\noindent 
Notice that unfolding is well defined because we consider session types with guarded recursion in equation set notation.

The transition relation for configurations $\cnfg{T_1\{\Tvar t=T_{1,\Tvar t} \mid \Tvar t \in \mathsf{Vars}_1 \}}{\word_1}{T_2\{\Tvar t=T_{2,\Tvar t} \mid \Tvar t \in \mathsf{Vars}_2 \}}{\word_2}$, with $T_i\{\Tvar t=T_{i,\Tvar t} \mid \Tvar t \in \mathsf{Vars}_i \}$, for $i \in \{1,2\}$, being session types in equation set notation, is defined as in Definition \ref{def:transrel} by using the above definition of unfolding (and by assuming that the $\{\Tvar t=T_{i,\Tvar t} \mid \Tvar t \in \mathsf{Vars}_i \}$ equational part is copied, for both $T_1$ and $T_2$, after every transition). 

Given $T_1$ and $T_2$ session types, it obviously holds (by standard arguments) that the transition system of $\cnfg{T}{\lempty}{S}{\lempty}$ is bisimilar to that of $\cnfg{\mathsf{esn}(T)}{\lempty}{\mathsf{esn}(S)}{\lempty}$, hence
that: $T$ and $S$ are compliant if and only if $\mathsf{esn}(T)$ and $\mathsf{esn}(S)$ are compliant.

We now define predicate $\contro$ for session types in equation set notation. $\contro$ is defined as in Definition \ref{def:controllability}, by assuming that predicate $\unicontro$ is, instead, defined as follows.
$T\{\Tvar t=T_\vt \mid \Tvar t \in \mathsf{Vars} \} \unicontro$ if there exists an indexing (total order) ${\Tvar t}_i$ on the variables of $\mathsf{Vars}$  such that $\{ \vt_i \mid 1 \leq i \leq n \}=\mathsf{Vars}$ and, for all $i$, with $1 \leq i \leq n$, it, holds:
$$\Tend \!\in\! T_i \vee \exists \Tvar t_j \!:\!  j<i \wedge \Tvar t_j \!\in\! \free{T_i}$$
Moreover, as in Definition \ref{def:controllability}, in order to establish $\contro$ of a session type $T\{\Tvar t=T_\vt \mid \Tvar t \in \mathsf{Vars} \}$ input prefix replacement must preliminarily be performed, so to obtain session types $T'\{\Tvar t=T'_\vt \mid \Tvar t \in \mathsf{Vars}' \}$ where $\mathsf{Vars}' \subseteq \mathsf{Vars}$ and in both term $T' \in \mathsf{sin}(T)$ and all terms $T'_\vt \in \mathsf{sin}(T_\vt)$, with $\Tvar t \in \mathsf{Vars}'$, all input prefixes have a single label.

\begin{prop}\label{prop:esncontrollable}
$T$ $\contro$ if and only if $\mathsf{esn}(T)$ $\contro$.
\end{prop}
\begin{proof}
We first show that $T$ $\contro$ implies $\mathsf{esn}(T)$ $\contro$. Given $T'$ obtained by input prefix replacement from $T$ (so to have input prefixes with single choices) that satisfies the $\unicontro$ predicate, we correspondingly consider $\mathsf{esn}(T')$, which is 
an input prefix replacement of  $\mathsf{esn}(T)$. $\mathsf{esn}(T') \unicontro$ is an immediate consequence of $T' \unicontro$ by considering the indexing  ${\Tvar t}_i$ of variable names used in the recursions of $T$ obtained as follows. 
We incrementally assign indexes to variables (starting from $1$) according to a depth-first visit of the syntax tree of $T$ 
as follows. When we are at a 
 $\Trec t.T''$ node, we have two cases. Either $\vt$ has already an assigned index (not possibile at the beginning) or not.
In the latter case: we consider all $\Trec t'.\_$ operators occurring in $T''$, if any, that syntactically include $\Tend$ or variable $\Tvar t''$ such that $\Tvar t'' \!\neq\! \Tvar t \wedge \Tvar t'' \!\in\! \free{T''}$ and we assign an index to all such $\vt'$ (incrementing the last assigned index) in increasing order from the innermost to the outermost; then we assign  an index to $\vt$ (incrementing the last assigned index). Finally, in both cases, we visit all the $\Trec t'.\_$ descendants (with no other recursion node in-between) of the $\Trec t.\_$ node, if any.

We now show that $\mathsf{esn}(T)$ $\contro$ implies $T$ $\contro$. Given $T_\mathsf{init}\{\Tvar t=T_\vt \mid \Tvar t \in \mathsf{Vars} \}$ obtained by input prefix replacement from $\mathsf{esn}(T)$ that satisfies the $\unicontro$ predicate, we correspondingly consider the only term $T'$ which is an input prefix replacement of $T$ such that $\mathsf{esn}(T')=T_\mathsf{init}\{\Tvar t=T_\vt \mid \Tvar t \in \mathsf{Vars} \}$.
We show that $T' \unicontro$ (Definition \ref{def:unicontropure} above) by structural induction:
\begin{itemize}
\item  For the base cases ${\Tvar t \unicontro}$ and ${\Tend \unicontro}$ we have nothing to show.
\item  $\Tbranchsingle{l}{T''}  \unicontro$ and $\Tselect{l}{T''}   \unicontro$ are a direct consequence of the induction
hypothesis, i.e.\ $T''  \unicontro$ and $\forall i \in I \dst T''_i \unicontro$, respectively.
\item  $\Trec t.T'' \unicontro$ is  a direct consequence of the induction hypothesis $T'' \unicontro$ and of the fact that: 
$\Tend \!\in\! T'' \vee \exists \Tvar t' \!:\!  \Tvar t' \!\neq\! \Tvar t \wedge \Tvar t' \!\in\! \free{T''}$. The latter is shown  as follows. 
From $T_\mathsf{init}\{\Tvar t=T_\vt \mid \Tvar t \in \mathsf{Vars} \} \unicontro$ we know that there exists a variable indexing ${\Tvar t}_i$ such that,
for all $i\in I$ it, holds: $\Tend \!\in\! T_i \vee \exists \Tvar t_j \!:\!  j<i \wedge \Tvar t_j \!\in\! \free{T_i}$.
So, given index $i$ such that $\vt_i=\vt$, we have to show: $\Tend \!\in\! T'' \vee \exists z \!:\!  z \!\neq\! i \wedge \Tvar \vt_z \!\in\! \free{T''}$. What we know  is that $\Tend \!\in\! T_i \vee \exists \Tvar t_j \!:\!  j<i \wedge \Tvar t_j \!\in\! \free{T_i}$, so there are two cases:
\begin{enumerate}
\item Either it holds $\Tend \!\in\! T'' \vee \Tvar t_j \!\in\! \free{T''}$ and we are done (with $z=j$).
\item Or $\Trec t_j.T'''$, for some $T'''$, is a subterm of $T''$. In this case we show that: $\Tend \!\in\! T''' \vee \exists z \!:\!  z \!\neq\! i \wedge \Tvar \vt_z \!\in\! \free{T'''}$. To do this we consider index $j$ and the defining term $T_j$ in its equation: we know that $\Tend \!\in\! T_j \vee \exists \Tvar t_k \!:\!  k<j \wedge \Tvar t_k \!\in\! \free{T_j}$.
Now again we have the same two cases, considering index $k$ instead of $j$ and term $T'''$ instead of term $T''$. Notice that we cannot proceed like this forever because the syntax of $T''$ is finite, hence case $1.$ must eventually apply. Moreover when this happens, we are sure that the variable $\vt_z$ that we detect is different from $\vt=\vt_i$ (i.e.\ $z \neq i$) because the indexing of the variables that we consider are always strictly smaller than $i$.  
\qedhere
\end{enumerate}
\end{itemize}
\end{proof}
\noindent 
We are now in a position to prove the desired theorem. We prove implications in the two opposite directions one at a time.

\begin{thm}If there exists a  
session type $S$ such that $T$ and $S$ are compliant then $T$ $\contro$.
\end{thm}
\begin{proof}
Since $T$ and $S$ are compliant, as observed above, we have also that $\mathsf{esn}(T)$ and $\mathsf{esn}(S)$ are compliant.
Therefore (the transition system of) configuration $\cnfg{\mathsf{esn}(T)}{\lempty}{\mathsf{esn}(S)}{\lempty}$ is a correct  composition according to Definition \ref{def:compliance}.
 
We now show that $\mathsf{esn}(T)$ $\contro$: by Proposition \ref{prop:esncontrollable}
this implies that $T$ $\contro$. In order to do this we need to enrich the transition system representation of the behaviour of
configurations $\cnfg{T_1\{\Tvar t=T_{1,\Tvar t} \mid \Tvar t \in \mathsf{Vars_1} \}}{\word_1}{T_2\{\Tvar t=T_{2,\Tvar t} \mid \Tvar t \in \mathsf{Vars_2} \}}{\word_2}$. We assume the transition relation $\rightarrow$ defined in Definition 
\ref{def:transrel} to be enriched as follows:  $\rightarrow$  transitions originated from outputs of $T_1$ (rule $1.$ of Definition 
\ref{def:transrel}) are assumed to be decorated with the label $l_j$ of the performed output (denoted by $\arrow{\overline{l_j}}$), while $\rightarrow$  transitions originated from inputs of $T_1$ (rule $2.$ of Definition 
\ref{def:transrel}) are assumed to be decorated with the label $l_j$ of the performed input (denoted by $\arrow{l_j}$). Notice that, in case of transitions originated from inputs or outputs of $T_2$ no decoration is added to transitions $\rightarrow$. Moreover, rule $3.$ (about recursion unfolding) of Definition \ref{def:transrel} is assumed to just copy the decoration labeling the transition (if there is any). 

We now consider such an enriched transition system over configurations  $\cnfg{T_1\{\Tvar t=T_{1,\Tvar t} \mid \Tvar t \in \mathsf{Vars_1} \}}{\word_1}{T_2\{\Tvar t=T_{2,\Tvar t} \mid \Tvar t \in \mathsf{Vars_2} \}}{\word_2}$. We use $\conf$ to range over these configurations. We say that a configuration $s= \cnfg{T_1\{\Tvar t=T_{1,\Tvar t} \mid \Tvar t \in \mathsf{Vars_1} \}}{\word_1}{T_2\{\Tvar t=T_{2,\Tvar t} \mid \Tvar t \in \mathsf{Vars_2} \}}{\word_2}$ exposes variable $\vt' \in \mathsf{Vars_1}$ if $T_1=\vt'$. Moreover, we denote transition systems paths starting from a given configuration $\conf$, i.e.\ finite sequences of transitions $\conf \arrow{\alpha_1} \conf_1 \arrow{\alpha_2} \conf_2 \dots \arrow{\alpha_n} \conf_n$ (where $\alpha_i$ decorations can be $\varepsilon$ in case of non decorated $\rightarrow$ transitions), by means of strings 
$\langle \alpha_1,\conf_1 \rangle \langle \alpha_2,\conf_2 \rangle \dots \langle \alpha_n,\conf_n \rangle$ (strings over pairs $\langle \alpha',\conf' \rangle$ with $\alpha'$ being a decoration or $\varepsilon$ and $\conf'$ a configuration).

Assuming $\mathsf{esn}(T)= T_{\mathsf{init}}\{\Tvar t=T_{\Tvar t} \mid \Tvar t \in \mathsf{Vars} \}$,
we now construct an indexing on the variables in the subset $\mathsf{Vars'}$ of $\mathsf{Vars}$, which includes variables $\vt$ such that: a configuration $\conf$ that exposes $\vt$ is reachable from the initial configuration $\cnfg{\mathsf{esn}(T)}{\lempty}{\mathsf{esn}(S)}{\lempty}$. We proceed as follows.
If $\mathsf{Vars'} \neq \emptyset$, then we consider
any reachable configuration $\conf$ that exposes some variable $\vt \in \mathsf{Vars}$. 
Since $\cnfg{\mathsf{esn}(T)}{\lempty}{\mathsf{esn}(S)}{\lempty}$ is a correct  composition, the configuration $\conf$ must 
reach a configuration $\conf'$ such that $\conf'\surd$. We consider the path from $\conf$ to $\conf'$ and the last configuration $\conf''$ of such a path that exposes a variable. We denote such a variable with $\vt_1$, the configuration $\conf''$ that exposes it with $\conf_1$, and the path (string) from $\conf_1$ that leads to $\conf'$ (part of the path from $\conf$ to $\conf'$ considered above) with $\mathsf{path}_1$.
In any subsequent $k$-th step, with $k \geq 2$, we consider the set $\mathsf{Vars}_k = \mathsf{Vars'} - \{ \vt_h \mid h < k\}$.
If $\mathsf{Vars}_k \neq \emptyset$, then we consider
any reachable configuration $\conf$ that exposes some variable $\vt \in \mathsf{Vars}_k$. 
Since $\cnfg{\mathsf{esn}(T)}{\lempty}{\mathsf{esn}(S)}{\lempty}$ is a correct  composition, the configuration $\conf$ must 
reach a configuration $\conf'$ such that $\conf'\surd$. We consider the path from $\conf$ to $\conf'$ and the first configuration $\conf''$ of such a path that either exposes a variable in $\{ \vt_h \mid h < k\}$ or is such that $\conf'' \surd$. Again we consider the path from $\conf$ to $\conf''$ and the last configuration $\conf'''$ of such a path that: is different from $\conf''$ and exposes a variable (such a variable must exist, because $\conf$ exposes a variable, and belong to $\mathsf{Vars}_k$ because of the way we have selected $\conf''$). We denote such a variable with $\vt_k$, the configuration $\conf'''$ that exposes it with $\conf_k$,
and the path (string) from $\conf_k$ that leads to $\conf''$ (part of the path from $\conf$ to $\conf''$ considered above) with $\mathsf{path}_k$.

We now consider terms $T'_k$ for each variable $\vt_k \in \mathsf{Vars'}$. We build $T'_k$ terms inductively by taking $T'_k=\mathsf{term}(T_{\vt_k}, \conf_k, \mathsf{path}_k)$, where  $\mathsf{term}(T',\conf, \mathsf{optpath})$, with $\mathsf{optpath}$ being either a $\mathsf{path}$ or $*$ (that represents being outside the path), is defined as follows.

\begin{itemize}
\item $\mathsf{term}(\Tvar t, \conf, \varepsilon)=\Tvar t$
\item $\mathsf{term}(\Tend, \conf, \varepsilon)=\Tend$

\item $\mathsf{term}(\Tbranch{l}{T}, \conf, \langle l_j,\conf' \rangle \mathsf{path})=\Tbranchsingle{l_j}{\mathsf{term}(T_j, \conf', \mathsf{path})}$
\item $\mathsf{term}(\Tselect{l}{T}, \conf, \langle \overline{l_j},\conf' \rangle \mathsf{path})=\Tselect{l}{T'}$\\
where $T'_j\!=\!\mathsf{term}(T_j, \conf', \mathsf{path})$ and, for all $i \!\in\! I$,\ $i\!\neq\! j$: $T'_i\!=\!\mathsf{term}(T_i, \conf_i, *)$ with
$\conf \arrow{\overline{l_i}} \conf_i$

\item $\mathsf{term}(T', \conf, \langle \varepsilon,\conf' \rangle \mathsf{path})=\mathsf{term}(T', \conf', \mathsf{path})$

\item $\mathsf{term}(\Tvar t, \conf, *)=\Tvar t$
\item $\mathsf{term}(\Tend, \conf, *)=\Tend$

\item $\mathsf{term}(\Tbranch{l}{T}, \conf, *)=\Tbranchsingle{l_j}{\mathsf{term}(T_j, \conf_j, *)}$ if $\conf$ has some $\arrow{l}$ transition\\
where $j$ is any $i \in I$ such that $\conf \arrow{l_j} \conf_j$

\item $\mathsf{term}(\Tselect{l}{T}, \conf, *)=\Tselect{l}{\mathsf{term}(T_i, \conf_i, *)}$ if $\conf$ has some $\arrow{\overline{l}}$ transition\\
where, for all $i \!\in\! I$, $\conf \arrow{l_i} \conf_i$

\item $\mathsf{term}(T', \conf, *)=\mathsf{term}(T', \conf', *)$ if $T' \notin \{\Tvar t, \Tend\}$ and $\conf$ has neither $\arrow{l}$ nor $\arrow{\overline{l}}$ transitions\\
where $\conf'$ is the first configuration having some $\arrow{l}$ transition or some $\arrow{\overline{l}}$ transition
in the path from $\conf$ to a configuration $\conf''$ such that $\conf'' \surd$ (such a path must exist because $\cnfg{\mathsf{esn}(T)}{\lempty}{\mathsf{esn}(S)}{\lempty}$ is a correct  composition)

\end{itemize}
where we use $\varepsilon$ to represent the empty string.

We also take $T'_\mathsf{init} =  \mathsf{term}(T_\mathsf{init},\cnfg{\mathsf{esn}(T)}{\lempty}{\mathsf{esn}(S)}{\lempty},*)$.

We now have that $T'_{\mathsf{init}}\{\vt_k=T'_k \mid \vt_k \in \mathsf{Vars}' \}$ is a session type in equation notation: 
$\mathsf{Vars}'$ must include all variables in $\free{T'_{\mathsf{init}}}$ and also in $\free{T'_k}$ for all $\vt_k \in \mathsf{Vars}'$ because, otherwise, a configuration $\conf$ exposing the variable that is not included in $\mathsf{Vars}'$ would have been reachable from the initial configuration $\cnfg{\mathsf{esn}(T)}{\lempty}{\mathsf{esn}(S)}{\lempty}$ (which contradicts the definition of $\mathsf{Vars}'$). Moreover, due to the way $\mathsf{term}$ is defined, $T'_{\mathsf{init}}\{\vt_k=T'_k \mid \vt_k \in \mathsf{Vars}' \}$ is obtained from $T_{\mathsf{init}}\{\Tvar t=T_{\Tvar t} \mid \Tvar t \in \mathsf{Vars} \}$ by performing input replacement that yields input prefixes with single inputs.
Finally, being $\conf_k$ the last configuration exposing a variable inside a path ending with a configuration $\conf$ that either exposes a variable in $\{ \vt_h \mid h < k\}$ (and not having previous configurations exposing such variables) or is such that $\conf \surd$, each of the $T'_k$ satisfies the constraint $\Tend \!\in\! T'_k \vee \exists \Tvar t_h \!:\!  h<k \wedge \Tvar t_h \!\in\! \free{T'_k}$.
\end{proof}

\begin{thm}If $T$ $\contro$ then there exists a  
session type $S$ such that $T$ and $S$ are compliant.
\end{thm}
\begin{proof}
If $T$ $\contro$ then $\mathsf{esn}(T)=T_\mathsf{init}\{\Tvar t=T_\vt \mid \Tvar t \in \mathsf{Vars} \}$ $\contro$. 
That is, there exists an input prefix replacement that yields a session type $T'_\mathsf{init}\{\Tvar t=T'_\vt \mid \Tvar t \in \mathsf{Vars}' \}$ such that $\mathsf{Vars}' \subseteq \mathsf{Vars}$ (and in both term $T'_\mathsf{init} \in \mathsf{sin}(T_\mathsf{init})$ and all terms $T'_\vt \in \mathsf{sin}(T_\vt)$, with $\Tvar t \in \mathsf{Vars}'$, all input prefixes have a single label) and that satisfies the $\unicontro$  predicate, i.e.\ there exists an indexing $\vt_i$ of the $\mathsf{Vars}'$ variables, such that:
$\Tend \!\in\! T'_{\vt_i} \vee \exists \Tvar t_j \!:\!  j<i \wedge \Tvar t_j \!\in\! \free{T'_{\vt_i}}$.
We assume set $\mathsf{Vars}'$ to be minimal, i.e.\ to not include any defined but unused variable name and  we take $S$ to be the unique session type  such that $\mathsf{esn}(S)=\dual{T'_\mathsf{init}}\{\Tvar t=\dual{T'_\vt} \mid \Tvar t \in \mathsf{Vars}' \}$.

In the following we will consider configurations $\cnfg{T_1\{\Tvar t=T_\vt \mid \Tvar t \in \mathsf{Vars} \}}{\word_1}{T_2\{\Tvar t=\dual{T'_\vt} \mid \Tvar t \in \mathsf{Vars}' \}}{\word_2}$ that are reachable from the initial configuration $\conf_\mathsf{init}=\cnfg{\mathsf{esn}(T)}{\lempty}{\mathsf{esn}(S)}{\lempty}$. We say that any such configuration exposes variable $\vt' \in \mathsf{Vars}$ if $T_1=\vt'$. 
Now, given any configuration $\conf$ reachable from the initial configuration $\conf_\mathsf{init}$, we have that $\conf$ is such that:
\begin{itemize}
\item $\word_1 = \lempty \vee \word_2 = \lempty$ 
\item There exists a configuration $\conf_\lempty$, which is reached from $\conf$ with the transitions originated by performing either the non-empty $\word_1$ sequence of inputs in the lefthand type or the non-empty sequence $\word_2$ of inputs in the righthand type,  such that $\conf_\lempty = \cnfg{T'_1\{\Tvar t=T_\vt \mid \Tvar t \in \mathsf{Vars} \}}{\lempty}{\dual{T'_2}\{\Tvar t=\dual{T'_\vt} \mid \Tvar t \in \mathsf{Vars}' \}}{\lempty}$, with $T'_2 \in \mathsf{sin}(T'_1)$.
\end{itemize}
This property of $\conf$ is, indeed, an invariant property of all configurations reachable from the initial configuration $\conf_\mathsf{init}$ in that: it is satisfied by $\conf_\mathsf{init}$ itself and it is preserved both by transitions originated from outputs of the lefthand or righthand type (which, for a configuration satisfying the above property, can be done only if its own queue is empty, and have the effect of enqueuing in the righthand or lefthand type, respectively, a symbol that it can then, dually, dequeue with an input) and by transitions originated from inputs of the lefthand or righthand type (which just make the already existing input transition sequence to $\conf_\lempty$ shorter).

We now notice that it is possible to reach, from $s_\lempty$, by performing outputs of the lefthand or righthand type immediately followed by inputs dually executed by the righthand or lefthand type, respectively: either a configuration $\conf'$ such that $\conf' \surd$ (in case $\Tend \in T'_2$), or a configuration exposing  an indexed variable $\vt_i \in \mathsf{Vars}'$. In the latter case, we can, similarly, reach: either a configuration $\conf''$ such that $\conf'' \surd$ (in case $\Tend \in T'_{\vt_i}$), or a configuration exposing an indexed variable $\vt_j \in \mathsf{Vars}'$ with $j <i$. In the latter case, we repeat, again, the same step: we are guaranteed to eventually meet the case in which a $\surd$ configuration is reached in that variable indexes strictly decrease at each step.
We thus have that $\mathsf{esn}(T)$ and $\mathsf{esn}(S)$ are compliant, hence $T$ and $S$ are compliant.
\end{proof}

\subsection{Soundness of Fair Asynchronous Subtyping w.r.t. Fair Refinement}
\label{app:soundnessProof}

\begin{lem}\label{lemma1}
Consider the session type $T = \context {A} {\Tselectindex l{T_k}{j}{J}} {k\in K}$.
Let $P_2=\cnfg{T}{\word_T}{S}{\word_S}$ and 
$P_1^i=\cnfg{\context {A} {T_{ki}} {k\in K}}{\word_T}{S}{\word_S \append l_i}$,
for every $i \in J$. If $P_2$ is a correct composition
then one of the following holds: 
\begin{itemize}
\item
$\mathcal A$ does not contain any input branching and
$P_2 \trans{} P_1^i$, for every $i \in J$;
\item
$\mathcal A$ contains an input branching and 
$P_1^i$ (for every $i \in J$) and $P_2$ have at least one outgoing transition.\\
For every possible transition $P_1^i \trans{} P_1'$
we have that one of the following holds:
\begin{enumerate}
\item\label{item1}
$P_1^i$ does not consume the label $l_i$ and
there exist $\mathcal A'$, $W$, $T'_{wj}$ (for every $w\in W$, $j \in J$), $S'$, $\word_T'$ 
and $\word_S'$ s.t. $P_1' = [\context {A'} {T'_{wi}} {w\in W},\word_T']\pa [S',\word_S'\append l_i]$
        and \\
$P_2 \trans{} [\context {A'} {\Tselectindex l{T'_w}{j}{J}} {w\in W},\word_T']\pa [S',\word_S']$;
\item\label{item2}
$P_1^i$ consumes the label $l_i$, hence
$P_1' = [\context {A} {T_{ki}} {k\in K},\word_T] \pa [S',\word_S]$, and
$\exists j \in \{1,\ldots, m\}$ s.t.
$P_2 \trans{}^* [T_{ji},\word_T']\pa [S',\word_S]$
and
$\word_T=a_1 \append \dots \append  a_w \append \word_T'$,
where $a_1,\dots,a_w$ are the labels in one of the paths to $[\,]^j$
in $\mathcal A$.
\end{enumerate}
For every possible transition $P_2 \trans{} P_2'$
we have that there exist $\mathcal A'$, $W$, $T'_{wj}$ (for every $w\in W$, $j \in J$), $S'$, $\word_T'$ 
and $\word_S'$ s.t. \\
$P_2' = [\context {A'} {\Tselectindex l{T'_w}{j}{J}} {w\in W},\word_T']\pa [S',\word_S']$ and\\
$P_1^i \trans{}  [\context {A'} {T'_{wi}} {w\in W},\word_T']\pa [S',\word_S' \append l_i]$.
\end{itemize}
\end{lem}

\begin{lem}\label{lemma2}
Consider $P_1=[\context {A} {T_{k}} {k\in K},\word_T] \pa [S,\word_S]$ and
$P_2 = [T_{j},\word_T']\pa [S,\word_S]$
with
$\word_T = a_1 \append \dots \append  a_w \append \word_T'$,
where $a_1,\dots,a_w$ are the labels in one of the paths to $[\,]^j$
in $\mathcal A$. We have that if $P_2$ is a correct composition,
then also $P_1$ is a correct composition.
\end{lem}
\begin{proof}
By contraposition, assume $P_1$ is not a correct composition.
This implies the existence of $P_1'$, from which it is not
possible to reach a successful configuration, such that 
$P_1 \trans{}^* P_1'$. If the labels $a_1,\dots,a_w$
were not consumed, we extend $P_1 \trans{}^* P_1'$ to
$P_1 \trans{}^* P_1''$ by allowing the l.h.s. type to consume
all the labels $a_1,\dots,a_w$. We have that also from $P_1''$
is not possible to reach a successful configuration.
We now reorder the transitions in $P_1 \trans{}^* P_1''$
such that in the initial $w$ steps the l.h.s. type consumes
the labels $a_1,\dots,a_w$. After these transitions the
configuration $P_2$ is reached. This implies that also
$P_2 \trans{}^* P_1''$, but this is not possible because
$P_2$ is a correct composition and from $P_1''$ no successful 
configuration can be reached.
\end{proof}

\begin{lem}\label{lemma3}
Consider the session type $T = \context {A} {\Tselectindex l{T_k}{j}{J}} {k\in K}$.
Let $P_2=\cnfg{T}{\word_T}{S}{\word_S}$ and 
$P_1^i=\cnfg{\context {A} {T_{ki}} {k\in K}}{\word_T}{S}{\word_S \append l_i}$,
for every $i \in J$. If $P_2$ is a correct composition
then, for every $i\in J$, there exists $\cnfg{T'}{\word_T'}{S'}{\word_S'}$ such that
$P_1^i \rightarrow^* \cnfg{T'}{\word_T'}{S'}{\word_S'}$ and  
$\cnfg{T'}{\word_T'}{S'}{\word_S'}\surd$.
\end{lem}
\begin{proof}
Given that $P_2$ is a correct composition, we know that 
there exists 
$\cnfg{T'}{\word_T'}{S'}{\word_S'}$ s.t. $\cnfg{\context {A} {\Tselectindex l{T_k}{j}{J}} {k\in K}}{\word_T}{S}{\word_S} \rightarrow^* \cnfg{T'}{\word_T'}{S'}{\word_S'}$ and  
$\cnfg{T'}{\word_T'}{S'}{\word_S'}\surd$.
During this sequence of transitions, the input context $\ctx A$ will 
become without input branchings, because a configuration that 
contains one type with an input branching is not successful. 
In other terms there exist a prefix of the
sequence of transitions, at the end of which the input context 
becomes without input branchings. We proceed by induction on the
length of such a prefix. If the length is zero, we can apply the first
item of Lemma \ref{lemma1} to conclude that $P_2 \trans{} P_1^i$, for every $i \in J$,
hence also $P_1^i$ can reach a successful configuration.
In the inductive step, we consider the first transition of $P_2$,
we apply the last item of Lemma \ref{lemma1} to show that
also $P_1^i$, for every $i \in J$, can perform a transition
such that it is possible to apply again the hypothesis on the
reached configurations. This is possible because if $P_2$ is correct, also the
configurations it can reach are correct.
\end{proof}

\begin{prop}\label{prop:anticipation}
Consider the session type $T = \context {A} {\Tselectindex l{T_k}{j}{J}} {k\in K}$.
If 
$\cnfg{T}{\word_T}{S}{\word_S}$ is a correct composition then, for 
every $i \in J$, we have that also
$\cnfg{\context {A} {T_{ki}} {k\in K}}{\word_T}{S}{\word_S \append l_i}$
is a correct composition.
\end{prop}
\begin{proof}
By contraposition, assume $i \in J$ s.t.
$P_1^i=\cnfg{\context {A} {T_{ki}} {k\in K}}{\word_T}{S}{\word_S \append l_i}$
is not a correct composition.
This means the existence of $P_1^i \trans{}^* P'$ such that
$P'$ cannot reach a successful configuration.
By induction on the length of this sequence of transition
we show that, differently from what assumed, $P'$
can reach a successful configuration.
If the length is 0, we simply apply Lemma \ref{lemma3} to show
that $P_1^i=P'$ can reach a successful configuration.
If the length is not 0, we consider two possible cases:
(i) the initial transition of $P_1^i \trans{} P''$ of $P_1^i \trans{}^* P'$
consumes the label $l_i$ from the the queue of the r.h.s. type or
(ii) it does not.
In case (i) we use the corresponding item \ref{item2} in Lemma \ref{lemma1}
to see that we can apply Lemma \ref{lemma2} on $P_2$ and $P''$, in order to conclude
that $P''$ is a correct composition. Given that $P'' \trans{}^* P'$
we can conclude that $P''$ can reach a successful configuration.
In case (ii) we use the corresponding item \ref{item1} in Lemma \ref{lemma1}
to conclude that we can apply again the inductive hypothesis
on the shortest sequence of transitions $P'' \trans{}^* P'$. This is 
possible because $P_2$ has a corresponding transition to $P_2 \trans{} P_2'$,
such that $P''$ and $P_2'$ still satisfies the assumption in the
statement of the Lemma. In particular $P_2'$ is a correct composition 
because also $P_2$ is a correct composition.
\end{proof}

\begin{lem}\label{lemma:controllable}
If $\cnfg{S}{\word_S}{R}{\word_R}$ is a correct composition then
$S$ is controllable.
\end{lem}
\begin{proof}
We show the existence of a type $T$ such that $\cnfg{S}{\lempty}{T}{\lempty}$ is a correct composition.

Consider a type $T$ defined as follows. Assume $\word_S=l_1^S\cdots l_k^S$ and $\word_R=l_1\cdots l_w^R$.
The type $T$ initially performs $k$
outputs with single output labels $l_1$, $\cdots$, $l_k$, respectively. After such outputs,
it becomes like $R$, with the difference that along all of its paths, the initial $w$ input
branchings are replaced by one of its continuation as follows:
the $i$-th input branching is replaced by its continuation in the branch 
labeled with $l_i^R$.

We now show by contraposition that $\cnfg{S}{\lempty}{T}{\lempty}$ is a correct composition.
If $\cnfg{S}{\lempty}{T}{\lempty}$ is not correct, then there exists  
$\cnfg{S}{\lempty}{T}{\lempty} \trans{}^* \cnfg{S'}{\word_S'}{T'}{\word_T'}$
such that from $\cnfg{S'}{\word_S'}{T'}{\word_T'}$ it is not possible to reach a successful configuration.
It is not restrictive to assume that during 
$\cnfg{S}{\lempty}{T}{\lempty} \trans{}^* \cnfg{S'}{\word_S'}{T'}{\word_T'}$ 
the r.h.s. type has produced the queue $\word_S$ (in fact, if it has not
produced them, we continue the computation performing them).
We can also assume that outputs in $T$, corresponding to outputs in $R$
along an initial path with less than $w$ inputs have been all
performed (also in this case, if these outputs were not performed,
we continue the computation executing them). We have that also 
$\cnfg{S}{\word_S}{R}{\word_R}$ can perform a computation
$\cnfg{S}{\word_S}{R}{\word_R} \trans{}^* \cnfg{S'}{\word_S'}{T'}{\word_T'}$.
Given that $\cnfg{S}{\word_S}{R}{\word_R}$ is a correct composition,
we have that from $\cnfg{S'}{\word_S'}{T'}{\word_T'}$ will be possible
to reach a successful configuration, thus contradicting the 
above assumption. 
\end{proof}

\subtypingPropSoundness*
\begin{proof}
Given that $\cnfg{S}{\word}{R}{\word_R}$ is a correct composition,
there exist $S'$, $\word''$, $R''$, and $\word_R''$ such that 
$\cnfg{S}{\word}{R}{\word_R} \trans{}^* \cnfg{S'}{\word''}{R''}{\word_R''}$
and $\cnfg{S'}{\word'}{R''}{\word_R''}\surd$.
We proceed by induction on the length of this sequence
of transition.

If the length is 0, then $\cnfg{S}{\word}{R}{\word_R}\surd$,
that implies $\unf{S}=\Tend$, that also implies $\unf{T}=\Tend$
(because $T \subtype S$), from which we have $\cnfg{T}{\word}{R}{\word_R}\surd$.

If the length is greater than 0, we proceed by case analysis on the
possible first transition $\cnfg{S}{\word}{R}{\word_R} \trans{} \cnfg{S''}{\word'''}{R'''}{\word_R'''}$.

If the transition is inferred by $R$ it is sufficient to observe that
$S''=S$ and $\cnfg{T}{\word}{R}{\word_R} \trans{} \cnfg{T}{\word'''}{R'''}{\word_R'''}$,
and then apply the inductive hypothesis because $\cnfg{S''}{\word'''}{R'''}{\word_R'''}$
is a correct composition in that it is reachable from a correct composition.

We now consider that the transition is inferred by $S$.\\
We first discuss the case in which $\unf{S}=\Tselect{l}{S}$.
In this case, the above transition is 
$\cnfg{S}{\word}{R}{\word_R} \trans{} \cnfg{S_i}{\word'''}{R'''}{\word_R'''}$, for 
some $i \in I$.
Given that $T\subtype S$, and $S$ is controllable by Lemma \ref{lemma:controllable},
we have $\unf{T}=\Tselect{l}{T}$
with $T_i \subtype S_i$, for every $i\in I$. This ensures that 
$\cnfg{T}{\word}{R}{\word_R} \trans{} \cnfg{T_i}{\word'''}{R'''}{\word_R'''}$.
Then we can apply the inductive hypothesis because $T_i \subtype S_i$
and $\cnfg{S_i}{\word'''}{R'''}{\word_R'''}$
is a correct composition.

We now discuss the case in which $\unf{S}=\Tbranch{l}{S}$.
There are two possible subcases: (i) also $T$ starts with
an input branching, i.e., $\unf{T}=\&\{l_j:T_j\}_{j\in J}$,
or (ii) $T$ starts with
an output selection, i.e., $\unf{T}=\oplus\{l_j:T_j\}_{j\in J}$.

In case (i), the above transition is 
$\cnfg{S}{\word}{R}{\word_R} \trans{} \cnfg{S_i}{\word'''}{R'''}{\word_R'''}$, for 
some $i \in I$.
Given that $T\subtype S$, and $S$ is controllable by Lemma \ref{lemma:controllable},
we have $\unfold {}T = \Tbranchindex lTjJ$, $J \supseteq K$, and
$\forall k\in K  \ldotp T_k \subtype S_k$,
where $K = \{ k \in I \; | \; S_k \text{ is controllable} \}$.
Given that $\cnfg{S}{\word}{R}{\word_R}$ is a correct composition and
$\cnfg{S}{\word}{R}{\word_R} \trans{} \cnfg{S_i}{\word'''}{R'''}{\word_R'''}$,
also the latter configuration is a correct composition.
By Lemma \ref{lemma:controllable} we have that $S_i$ is controllable.
This implies that $i \in K$, hence also $i \in J$.
This ensures that 
$\cnfg{T}{\word}{R}{\word_R} \trans{} \cnfg{T_i}{\word'''}{R'''}{\word_R'''}$.
Then we can apply the inductive hypothesis because $T_i \subtype S_i$
and $\cnfg{S_i}{\word'''}{R'''}{\word_R'''}$
is a correct composition.

In case (ii), given that $T \subtype S$,  and $S$ is controllable,
we have that 
$\selunf {S} = \context {A} {\Tselectindex l{S_k}{i}{J}} {k\in K}$,
and $\unf{T}=\oplus\{l_j:T_j\}_{j\in J}$
with $T_j \subtype \context {A} {S_{kj}} {k\in K}$, for every $j \in J$.
We first observe that the sequence of transitions 
$\cnfg{S}{\word}{R}{\word_R} \trans{}^* \cnfg{S'}{\word''}{R''}{\word_R''}$,
with $\cnfg{S'}{\word''}{R''}{\word_R''}\surd$,
includes at least one output selection $l_j$ executed
by one of the output selections filling the holes in $\ctx{A}$.
This label $l_j$ is the first one emitted by the l.h.s. type
after it has executed input branchings in $\ctx{A}$.
We have that the same sequence of transitions, excluding the output 
of $l_j$, can be executed from the configuration 
$\cnfg{\context {A} {S_{kj}} {k\in K}}{\word}{R}{\word_R \append l_j}$.
Such a sequence is $\cnfg{\context {A} {S_{kj}} {k\in K}}{\word}{R}{\word_R \append l_j} 
\trans{}^* \cnfg{S'}{\word''}{R''}{\word_R''}$,
with $\cnfg{S'}{\word''}{R''}{\word_R''}\surd$; notice that it is shorter than the 
above one.
We now consider $\cnfg{T}{\word}{R}{\word_R} \trans{} \cnfg{T_i}{\word}{R}{\word_R \append{l_j}}$.
We can now apply the inductive hypothesis on the shorter sequence 
$\cnfg{\context {A} {S_{kj}} {k\in K}}{\word}{R}{\word_R \append l_j} 
\trans{}^* \cnfg{S'}{\word''}{R''}{\word_R''}$, because $T_j \subtype \context {A} {S_{kj}} {k\in K}$
and by Proposition \ref{prop:anticipation} $\cnfg{\context {A} {S_{kj}} {k\in K}}{\word}{R}{\word_R \append l_j}$
is a correct composition.
\end{proof}

\subtypingSoundness*
\begin{proof}
If $S$ is not controllable, then the thesis trivially holds
because $T \refine S$ for every $T$.

We now consider $S$ controllable, and
we prove the thesis by showing that if $T \subtype S$ then,
for every $\word$, $R$, and $\word_R$ such that 
$\cnfg{S}{\word}{R}{\word_R}$ is a correct composition, we have that
the following holds:
\begin{itemize}

\item \label{item:one_stepT}
if $\cnfg{T}{\word}{R}{\word_R} \rightarrow \cnfg{T'}{\word'}{R'}{\word_R'}$
then there exists $S'$ such that $T' \subtype S'$ and
$\cnfg{S'}{\word'}{R'}{\word_R'}$ is a correct composition.

\end{itemize}
The above implies the thesis because, given $T \subtype S$ and the 
correct composition $\cnfg{S}{\lempty}{R}{\lempty}$,
if there exists a computation 
$\cnfg{T}{\lempty}{R}{\lempty} \trans{}^* 
\cnfg{T'}{\word'}{R'}{\word_R'}$,
we can apply the above result
on each step of the computation 
to prove that there exists $S'$ such that $T' \subtype S'$
and $\cnfg{S'}{\word'}{R'}{\word_R'}$ is a correct composition.
Then, by Proposition \ref{prop:pathTuSuccess}, we have that 
there exist $T''$, $\word''$, $R''$, and $\word_R''$
such that $\cnfg{T'}{\word'}{R'}{\word_R'} \trans{}^* \cnfg{T''}{\word''}{R''}{\word_R''}$
and $\cnfg{T''}{\word''}{R''}{\word_R''}\surd$.

We now prove the above result.
The transition $\cnfg{T}{\word}{R}{\word_R} \trans{} \cnfg{T'}{\word'}{R'}{\word_R'}$
can be of four possible kinds:
\begin{enumerate}
\item
the consumption of a message from the r.h.s. queue,
i.e. $\cnfg{T}{\word}{R}{l \append \word_R'} \rightarrow \cnfg{T}{\word}{R'}{\word_R'}$;
\item
the insertion of a new message in the l.h.s. queue, 
i.e. $\cnfg{T}{\word}{R}{\word_R} \rightarrow \cnfg{T}{\word\append l}{R'}{\word_R}$;
\item
the consumption of a message from the l.h.s. queue, 
i.e. $\cnfg{T}{l \append q'}{R}{\word_R} \rightarrow \cnfg{T'}{\word'}{R}{\word_R}$;
\item
the insertion of a new message in the r.h.s. queue, 
i.e. $\cnfg{T}{\word}{R}{\word_R} \rightarrow \cnfg{T'}{\word}{R}{\word_R\append l}$.
\end{enumerate}
In the first two cases, we simply observe that there exists also
$\cnfg{S}{\word}{R}{l \append \word_R'} \rightarrow \cnfg{S}{\word}{R'}{\word_R'}$
(resp.
$\cnfg{S}{\word}{R}{\word_R} \rightarrow \cnfg{S}{\word\append l}{R'}{\word_R}$),
that $T \subtype S$, 
and also $\cnfg{S}{\word}{R'}{\word_R'}$ (resp. $\cnfg{S}{\word\append l}{R'}{\word_R}$)
is a correct composition because reachable from the correct composition
$\cnfg{S}{\word}{R}{l \append \word_R'}$ (resp. $\cnfg{S}{\word}{R}{\word_R}$).

In the third case we have that $\unf T$ starts with an input
branching. Given that $T \subtype S$, and $S$ is controllable,
also $\unf S$ must start
with an input branching, i.e. 
$\unf{S}=\Tbranch{l}{S}$.
By definition of $\subtype$ we have that 
$\unfold {}T = \Tbranchindex lTjJ$, $J \supseteq K$, and
$\forall k\in K  \ldotp T_k \subtype S_k$,
where $K = \{ k \in I \; | \; S_k \text{ is controllable} \}$.
Given
that $\cnfg{S}{l \append q'}{R}{\word_R}$ is a correct composition,
there exists $i \in I$ s.t. $l=l_i$ and 
$\cnfg{S}{l \append q'}{R}{\word_R} \trans{} 
\cnfg{S_i}{\word'}{R}{\word_R}$. The former configuration is a correct
composition, hence also the latter is such. This implies, by Lemma \ref{lemma:controllable},
that $S_i$ is controllable, hence $i \in K$ and also $i\in J$.
Thus, we have
$\cnfg{T}{l \append q'}{R}{\word_R} \trans{} 
\cnfg{T_i}{\word'}{R}{\word_R}$,
with $T_i \subtype S_i$.
We conclude this case by observing again that 
$\cnfg{S_i}{\word'}{R}{\word_R}$ is a correct composition in that
reachable from the correct composition $\cnfg{S}{l \append q'}{R}{\word_R}$.

In the fourth and last case, we have that $\unf T$ starts with an output
selection, and $T'$ is the continuation in the branch with label $l$.
Given that $T \subtype S$, and $S$ is controllable, we have
$\selunf {S} = \context {A} {\Tselectindex l{S_k}{j}{I}} {k\in K}$,
and $T' \subtype S_{km}$, for every $k \in K$ and some $m \in I$ such that $l_m=l$.
It remains to show
that $\cnfg{\context {A} {{S_{km}}} {k\in K}}{\word}{R}{\word_R\append l}$ is a 
correct composition, but this follows from Proposition \ref{prop:anticipation} 
and the fact that $\cnfg{\context {A} {\Tselectindex l{S_k}{j}{I}} {k\in K}}{\word}{R}{\word_R}$,
with $l=l_m$ for some $m \in I$,
is a correct composition. In fact $\selunf {S} = \context {A} {\Tselectindex l{S_k}{j}{I}} {k\in K}$
and $\cnfg{S}{\word}{R}{\word_R}$ is a correct composition.
\end{proof}

\subsection{Undecidability of Fair Asynchronous Subtyping}\label{app:undecidabilitySubtyping}

\undecidabilitySubtyping*
\begin{proof}

We first consider the only-if part, proving the contrapositive statement, that is,
if the queue machine $M$ terminates then $T \not\!\!\!\subtype S$.
If the queue machine terminates, we have that $(s,\$) \rightarrow _{M}^* (q',\lempty)$.
Consider now the pair of types $(T,S)$ with $T=\semthree{M,\_,E}$ and $S=\semthree{M,E}$.
If, by contradiction, $T \subtype S$, since $S$ is controllable
(it is compliant, e.g., with its dual) we have that 
by Definition \ref{def:subtyping} there exists a fair asynchronous
subtyping relation $\mathcal R$ such that $(T,S) \in \mathcal R$.
We now show that, by definition of fair asynchronous
subtyping relation, $\mathcal R$ will have to include other pairs of
types $(T'',S'')$ corresponding with configurations $(q'',\gamma'')$ 
reachable in the queue machine $M$.
Consider the type $T$:
$$
\Trec{s}.\Tbranchset{A}{\semTcont{B^A_1\cdots B^A_{n_A}}_{q'}^{\{s\}}}{\Gamma}
$$
assuming $\delta(s,A)=(q',B^A_1\cdots B^A_{n_A})$ and 
$$
  \begin{array}{l}
           \semTcont{B_1\cdots B_{m}}_{r}^{\mathcal T} \!=\! \left\{\!\!
          \begin{array}{ll}
            \!\BsemT{r}{\mathcal T} 
            & \text{if }m=0\\
\begin{array}{ll}
            \!\!\!\!\oplus & 
            \!\!\!\!\big( \big\{B_1:  \semTcont{B_2\ldots B_m}_{r}^{\mathcal
            T}\big\} \cup
            \\
            & \! \big\{{A:V}\big\}_{A\in\Gamma\setminus\{B_1\}} \cup
            \{E: V'\}
             \big) 
            \end{array} 
            & \text{otherwise}
          \end{array}
              \right.
  \end{array}
$$
It starts with an input branching, with labels for each queue alphabet symbol
including the initial queue symbol $\$$.
Then it has a sequence of output selections, including the sequence
of symbols to be emitted by the queue machine after having consumed $\$$.
Consider now the type $S$:
$$
\&\{\$ : \Trec{\Tvar t}.\Tselectset{A}{\&\{A:\Tvar t\}}{\Gamma} \cup \{E:\&\{E:\Tend\}\}\}
$$
It starts with an input branching with only
label $\$$, followed by an output selection on all symbols, including label
$E$ having continuation $\&\{E: \Tend\}$. The latter ensures that $S$ is controllable.
If we consider the constraints
imposed by the Definition \ref{def:subtyping} on fair asynchronous subtyping relations,
we can conclude that $\mathcal R$ should contain a pair of types $(T',S')$
where $T'$ is the type corresponding to the new state of the queue machine
(reached after the above sequence of output selections
$\semTcont{B^\$_1\cdots B^\$_{n_\$}}_{q'}^{\{s\}}$ to 
be emitted by the queue machine after having consumed $\$$)
and $S'$ is like $S$, with the difference that before the output selection 
there is a sequence of input branchings, each one with only one label, corresponding
with the sequence of symbols $B^\$_1\cdots B^\$_{n_\$}$ in the queue after the first computation step.
This reasoning can be repeatedly applied to prove that $\mathcal R$ should also contain 
other pairs of types $(T'',S'')$, one for each configuration $(q'',\gamma'')$ 
reachable in the queue machine $M$.
Consider now the pair $(T_f,S_f) \in \mathcal R$ 
corresponding to the terminating configuration  $(q',\lempty)$.
The type $T_f$, as all the types representing states in the queue machine,
starts with an input branching.
The type $S_f$, on the other hand, represents the empty queue, so it is
$\Trec{\Tvar t}.\Tselectset{A}{\&\{A:\Tvar t\}}{\Gamma} \cup \{E:\&\{E:\Tend\}\}$,
i.e. it is 
like $\sem{M,E}$ but without input branchings before the output
selection. This means that $(T_f,S_f)$ does not satisfy the item for
input selection in Definition \ref{def:subtyping}.
Hence $\mathcal R$ cannot be a fair asynchronous subtyping,
but this contradicts the above initial assumption about $\mathcal R$
being a fair asynchronous session subtyping.

We now move to the if part. Assume that the queue machine $M$ does not terminate.
We show that there exists a fair asynchronous subtyping relation $\mathcal R$
that contains the pair $(T,S)$, hence $T \subtype S$. There 
are two kinds of pairs in $\mathcal R$: 
(i) the pairs discussed in the above
only-if part of the proof that corresponds to the path in the subtyping
simulation game that reproduces the computation of the queue machine $M$,
and (ii) other pairs corresponding to alternative paths.
The pairs of types (i) satisfy the constraints imposed by 
Definition~\ref{def:subtyping} because output selections
of the l.h.s. type can always be mimicked by the r.h.s. type
(that always include an output selection after a sequence of
input branchings with only one label),
and input branchings can always be mimicked by the r.h.s. type
because under the assumption that the queue machine does not terminate,
the queue is always non-empty during the computation.
Also the pairs of type (ii) satisfy the constraints imposed by 
Definition~\ref{def:subtyping}. In fact, these pairs are generated
considering the alternative branches in the l.h.s. types
$\semTcont{B_1\cdots B_{m}}_{r}^{\mathcal T}$ in Definition \ref{def:controlEncoding2},
namely, the branches corresponding with the labels $A$ and $E$ in the definition,
that we report here for reader convenience:
  $$
  \begin{array}{l}
           \semTcont{B_1\cdots B_{m}}_{r}^{\mathcal T} \!=\! \left\{\!\!
          \begin{array}{ll}
            \!\BsemT{r}{\mathcal T} 
            & \text{if }m=0\\
\begin{array}{ll}
            \!\!\!\!\oplus & 
            \!\!\!\!\big( \big\{B_1:  \semTcont{B_2\ldots B_m}_{r}^{\mathcal
            T}\big\} \cup
            \\
            & \! \big\{{A:V}\big\}_{A\in\Gamma\setminus\{B_1\}} \cup
            \{E: V'\}
             \big) 
            \end{array} 
            & \text{otherwise}
          \end{array}
              \right.
  \end{array}
  $$
  with $V=\Trec{\Tvar t}. \big( \Tselectset{A}{\Tvar t}{\Gamma} \cup \{E:V'\} \big)$
  and $V'=\Trec{\Tvar t}. \big( \Tbranchset{A}{\Tvar t}{\Gamma} \cup \{E:\Tend\} \big)$.
The l.h.s. type in the pairs $(T',S')$ associated with these branches,
are of two kinds: (a) they are able to recursively perform all possible
outputs until the label $E$ is selected (type $V$),
or (b) they are able to recursively perform all possible inputs until the label 
$E$ is selected (type $V'$). 
In the first case (a), the constraints in Definition \ref{def:subtyping}
are satisfied because the r.h.s. type is always able to mimick output selections
(see the above observation).
In the second case (b), we have that the output $E$ has been previously
selected by the last pair of kind (a) considered. Hence, the r.h.s. type
is a sequence of input branchings, with only one label, where
all inputs excluding the last one are different
from $E$, and the last one, having label $E$,
has continuation $\Tend$. This guarantees that all these
pairs satisfy the constraints in Definition \ref{def:subtyping},
under the assumption that also a final pair $(\Tend,\Tend)$ belongs
to $\mathcal R$. We the conclude by observing that we have proved the existence
of a fair session subtyping relation $\mathcal R$ such that
$(T,S) \in \mathcal R$ (in that this is 
the first pair of the kind (i) above), hence we have that $T \subtype S$.
\end{proof}

 \subsection{Soundness of the Algorithm w.r.t. Fair Asynchronous Subtyping}\label{subsec:algorithmsoundness}

\lemalgosoundness*
\begin{proof}
We proceed by induction.
If $h=1$, the thesis directly follows from the fact that $\mathcal T^1$
is contained in a simulation tree.

If $h>1$, by inductive hypothesis we have that the thesis holds
for $\mathcal T^{h-1}$. We prove that the thesis holds also for 
$\mathcal T^{h}$ showing that there exists a simulation tree
including $m \treetrans{} m'$ with $m'$ labeled with
$(T',\crepl{\grepl{\calA''}{\crepl{\calA^{v'}}{S'_j}^{j \in J}}{J}}{S'_k}^{k \in K})$
if and only if
there exists a simulation tree including
$t \treetrans{} t'$ with $t'$ labeled with
$(T',\crepl{\grepl{\calA''}{\crepl{\calA^{v'+1}}{S'_j}^{j \in J}}{J}}{S'_k}^{k \in K})$.
The proof is by case analysis, considering the three possible steps in the subtyping
simulation game at the basis of the definition of $\treetrans{}$. 

If $T$ starts with a recursive definition, the thesis trivially
holds because $\treetrans{}$
simply modify the l.h.s. type by unfolding its initial recursion and leaves
the r.h.s. type unchanged.

If $T$ starts with an input branching, by Definition \ref{def:subtyping}
we have that the r.h.s. type contains an entire context $\calA$ in its
growing holes.
We initially consider
$m \treetrans{} m'$ with $m'$ labeled with
$(T',\crepl{\grepl{\calA''}{\crepl{\calA^{v'}}{S'_j}^{j \in J}}{J}}{S'_k}^{k \in K})$.
This means that by applying $\unfold{}$ to the r.h.s. type we obtain an input context
starting with an input branching satisfying the constraints imposed by Definition \ref{def:subtyping}.
The step of the subtyping simulation game corresponding to $m \treetrans{} m'$
selects a branch of the input branching
such that its continuation $\crepl{\grepl{\calA''}{\crepl{\calA^{v'}}{S'_j}^{j \in J}}{J}}{S'_k}^{k \in K}$
is controllable.
Now consider $t$ with label 
$(T,\crepl{\grepl{\calA'}{\crepl{\calA^{v+1}}{S_j}^{j \in J}}{J}}{S_k}^{k \in K})$.
The application of $\unfold{}$ modifies the outer context in the same way
thus obtaining a type starting with the same input branching, simply
with an additional nesting of $\calA$ in the holes in $J$.
The continuation $\crepl{\grepl{\calA''}{\crepl{\calA^{v'+1}}{S'_j}^{j \in J}}{J}}{S'_k}^{k \in K}$
is also controllable because it is an input contexts with the set of indexed holes,
hence the same set of types $S'_j$ and $S'_k$.
Hence it is possible to apply a corresponding step in the subtyping 
simulation game $t \treetrans{} t'$ with $t'$ labeled with
$(T',\crepl{\grepl{\calA''}{\crepl{\calA^{v'+1}}{S'_j}^{j \in J}}{J}}{S'_k}^{k \in K})$.
Notice that the same reasoning can be applied assuming that
$t \treetrans{} t'$ with $t'$ labeled with
$(T',\crepl{\grepl{\calA''}{\crepl{\calA^{v'+1}}{S'_j}^{j \in J}}{J}}{S'_k}^{k \in K})$
to prove that there exists also the corresponding step in the subtyping
simulation game $m \treetrans{} m'$. In this case we use the assumption that 
in the growing holes of the r.h.s. type of the label of $m$ we have an entire 
context $\calA$, thus guaranteeing the presence of the same $S'_j$ in all the continuations
of the initial input branching present in the outer context. 

If $T$ starts with an output selection, we initially consider
$m \treetrans{} m'$ with $m'$ labeled with
$(T',\crepl{\grepl{\calA''}{\crepl{\calA^{v'}}{S'_j}^{j \in J}}{J}}{S'_k}^{k \in K})$.
This means that by applying $\selunf{}$ to the r.h.s.
type we obtain an input context filled with types starting with output
selections satisfying the constraints imposed by Definition \ref{def:subtyping}.
Notice that the application of $\selunf{}$ to the outer input context 
does not remove holes, but at most replicates some of them.
Moreover, the application of $\selunf{}$ applies to the innermost
types $S_j$ and $S_k$ by unfolding the variables inside outputs
replacing them with their definitions (already present in $S_j$ and $S_k$ given
that these are closed terms). The considered step in the subtyping
simulation game modifies (the unfoldings of) $S_j$ and $S_k$ by resolving 
initial output selections, thus obtaining $S_j'$ and $S_k'$.
Now consider $t$ with label 
$(T,\crepl{\grepl{\calA'}{\crepl{\calA^{v+1}}{S_j}^{j \in J}}{J}}{S_k}^{k \in K})$.
What we have just observed about the step $m \treetrans{} m'$
of subtyping simulation game, holds also for this new pair of types.
The application of $\selunf{}$ respectively modifies the
outer input context and the inner types $S_j$ and $S_k$ in the
same way, and also the same resolution of the initial output selections 
in $S_j$ and $S_k$ is possible.
Hence there exists 
$t \treetrans{} t'$ with $t'$ labeled with
$(T',\crepl{\grepl{\calA''}{\crepl{\calA^{v'+1}}{S'_j}^{j \in J}}{J}}{S'_k}^{k \in K})$.
Notice that the same reasoning can be applied assuming that
$t \treetrans{} t'$ with $t'$ labeled with
$(T',\crepl{\grepl{\calA''}{\crepl{\calA^{v'+1}}{S'_j}^{j \in J}}{J}}{S'_k}^{k \in K})$
to prove that there exists also the corresponding step in the subtyping
simulation game $m \treetrans{} m'$.
\end{proof}

\propalgosoundness*
\begin{proof}
We proceed by induction on the length of $n \treetrans{}\!\!^*\, n'$.

If the length is 0, then $n'$ is the root of $\mathcal T$
hence its label is obviously in $\mathcal T^1$.

If the length is greater than 1, consider 
$n \treetrans{}\!\!^*\, n''\treetrans{}n'$.
By inductive hypothesis we have that $\tlab(n'')$
is a label present either  
in $\mathcal T^h$, for some $h$,
or in $\simtree{T'}{S'}=(N', n_0', \treetrans, \tlab')$
with $T' \subtype S'$.

We start from the latter case, i.e., there exists $m''$
in $\simtree{T'}{S'}=(N', n_0', \treetrans, \tlab')$
such that $\tlab'(m'')=\tlab(n'')$. We have that there exists
$m''\treetrans{} m'$ in $\simtree{T'}{S'}$ s.t. $\tlab'(m')=\tlab(n')$.

We now consider the former case, i.e., there exists one node in 
$\mathcal T^h$, for some $h$, labeled with $\tlab(n'')$.
Let $m''$ be such node.
There are two possibilities, 
either (i) the node $m''$ is a leaf in $\mathcal T^h$,
or (ii) it is not a leaf.
In the case (ii) we have that $\mathcal T^h$
contains $m''\treetrans{} m'$, with $m'$ labeled with $\tlab(n')$. 
If $m''$ is a leaf, we consider the four kinds of leaves separately.

If $m''$ is a leaf of type \ref{algo:loop}, then there exists an ancestor $m'''$
of $m''$ in $\mathcal T^h$ with the same label $\tlab(n'')$. Given that the ancestor 
is not a leaf, $\mathcal T^h$
contains $m'''\treetrans{} m'$, with $m'$ labeled with $\tlab(n')$. 

If $m''$ is a leaf of type \ref{algo:increase} in $\mathcal T$, 
we have $\tlab(n'')=$ $(T',\crepl{\crepl{\calA^{h+1}}{S_j}^{j \in J}}{S_k}^{k\in K})$.
The node $n''$ has an ancestor $n'''$ in $\mathcal T^h$
s.t. $\tlab(n''')=(T',\crepl{\crepl{\calA^{h}}{S_j}^{j \in J}}{S_k}^{k\in K})$.
Consider now the corresponding node $m'''$ in $\mathcal T^{h+1}$.
We have that $m'''$ is labeled with $(T',\crepl{\crepl{\calA^{h+1}}{S_j}^{j \in J}}{S_k}^{k\in K})=\tlab(n'')$.
Given that $m'''$ 
is not a leaf, $\mathcal T^{h+1}$
contains $m'''\treetrans{} m'$, with $m'$ labeled with $\tlab(n')$. 

If $m''$ is a leaf of type \ref{algo:decrease} in $\mathcal T$, 
we have $\tlab(n'')=(T',\crepl{\crepl{\calA^{h}}{S_j}^{j \in J}}{S_k}^{k\in K})$.
We have two cases. If $h=1$, by definition of witness tree, $T' \subtype \crepl{\crepl{\calA^{h}}{S_j}^{j \in J}}{S_k}^{k\in K}$.
The node $n''$ has the same label as the root of 
$\simtree{T'}{\crepl{\crepl{\calA^{h}}{S_j}^{j \in J}}{S_k}^{k\in K}}$.
Hence such a simulation tree includes a transition from its root
to a node labeled with $\tlab(n')$.
If $h>1$ the node $n''$ has an ancestor $n'''$ in $\mathcal T^h$
such that $\tlab(n''')=(T',\crepl{\crepl{\calA^{h+1}}{S_j}^{j \in J}}{S_k}^{k\in K})$.
Consider now the corresponding node $m'''$ in $\mathcal T^{h-1}$.
We have that $m'''$ is labeled with $(T',\crepl{\crepl{\calA^{h}}{S_j}^{j \in J}}{S_k}^{k\in K})=\tlab(n'')$.
Given that $m'''$ 
is not a leaf, $\mathcal T^{h-1}$
contains $m'''\treetrans{} m'$, with $m'$ labeled with $\tlab(n')$. 

If $m''$ corresponds to leaf of type \ref{algo:const} in $\mathcal T$, 
we have that the label $\tlab(n'')$ of $m''$ is the same as the label in 
the corresponding node in $\mathcal T$, i.e. $(T',\calA'[S_k]^{k \in K'})$.
In fact labels of
the leaves of type \ref{algo:const} in $\mathcal T$ do not change when moving 
to $\mathcal T^h$. This because the input context $\calA'$
does not include growing holes.
By definition of witness tree we have that
$T' \subtype \calA'[S_k]^{k \in K'}$.
The node $n''$ has the same label as the root of 
$\simtree{T'}{\calA'[S_k]^{k \in K'}}$.
Hence such a simulation tree includes a transition from its root
to a node labeled with $\tlab(n')$.
\end{proof}

\thmalgosoundness*
\begin{proof}
Let $\mathcal T$ be the witness subtree with root in $n$.
By Proposition \ref{prop:algo} we have that $\tlab(n')$ is a label present either  
in $\mathcal T^h$, for some $h$,
or in $\simtree{T'}{S'}=(N', n_0', \treetrans, \tlab')$
with $T' \subtype S'$.
In the latter case the thesis trivially holds because
all nodes $m'$ in $\simtree{T'}{S'}$ are either successful or 
there exists $m' \treetrans{} m''$.
In the former case there are two cases: either there exists
an intermediary node (non-leaf) in one $\mathcal T^h$, for some $h$,
labeled with $\tlab(n')$ is an intermediary, or such a node can be
only in leaf positions. In the first case the thesis trivially holds 
because all intermediary nodes have successors.
The second case can occur only for leaves of type \ref{algo:decrease} 
in $\mathcal T$, or corresponding to leaves of type \ref{algo:const} 
in $\mathcal T$. Both cases imply that $\tlab(n')=(T',S')$
with $T' \subtype S'$. Hence $n'$ has the same label as the root
of $\simtree{T'}{S'}$ and, as above, the thesis trivially holds because
all nodes $m'$ in $\simtree{T'}{S'}$ are either successful or 
there exists $m' \treetrans{} m''$.
\end{proof}

\end{document}